\documentclass [3p, 12pt, sort&compress]{elsarticle}
\usepackage[numbers]{natbib}
\usepackage{graphicx}
\usepackage[small]{subfigure,epsfig}
\usepackage{indentfirst}
\usepackage{amsmath,latexsym,enumerate}
\usepackage{amssymb}
\usepackage{amsfonts}

\usepackage{amsthm}

\DeclareMathOperator{\erfc}{erfc}

\theoremstyle{plain} \newtheorem{theorem}{Theorem}
\newtheorem{lemma}{Lemma}
\numberwithin{equation}{section} \numberwithin{lemma}{section} \numberwithin{theorem}{section}
\newtheorem{definition}{Definition}[section]

\makeatletter
\def\ps@pprintTitle{%
  \let\@oddhead\@empty
  \let\@evenhead\@empty
  \def\@oddfoot{\reset@font\hfil\thepage\hfil}
  \let\@evenfoot\@oddfoot
}
\makeatother

\begin{document}
\begin{frontmatter}

\title{Integrability and solvability of polynomial Li\'{e}nard differential systems}

\author{Maria V. Demina}

\address{HSE University, 34 Tallinskaya Street, 123458, Moscow, Russian Federation, maria\underline{ }dem@mail.ru}

\begin{abstract}

We provide the necessary and sufficient conditions of
Liouvillian integrability for  Li\'{e}nard differential systems describing
nonlinear oscillators with a polynomial damping and a polynomial restoring
force. We prove that Li\'{e}nard differential systems are not Darboux
integrable excluding
 subfamilies with certain restrictions on the degrees of  the polynomials arising in the systems.
 We demonstrate that if the degree of a polynomial responsible for the
 restoring force is greater  than the degree of a polynomial producing the
 damping, then a generic Li\'{e}nard differential system is not Liouvillian
 integrable with the exception of  linear Li\'{e}nard systems.
 However, for any fixed degrees of the polynomials
 describing the damping and the restoring force  we present subfamilies  possessing
 Liouvillian first integrals. As a by-product of our results, we find a number of novel  Liouvillian integrable subfamilies.
 In addition, we
 study the existence of non-autonomous Darboux first integrals and
non-autonomous Jacobi last multipliers with a time-dependent exponential
factor.
\end{abstract}


\end{frontmatter}

\section{Introduction}\label{Sec_I}

Performing classifications of integrable or solvable subfamilies for a given
multi-parameter	system of ordinary differential equations is a very difficult
problem. The aim of the present article is to solve the integrability problem
for the following systems of first-order ordinary differential equations
\begin{equation}
 \label{Lienard_gen}
 x_t=y,\quad  y_t=-f(x)y-g(x).
\end{equation}
We suppose that $f(x)$ and $g(x)$ are polynomials
\begin{equation}
 \label{Lienard_fg}
f(x)=f_0x^m+\ldots+f_m, \quad g(x)=g_0x^{n}+\ldots+g_{n},\quad f_0 g_0\neq0
\end{equation}
with coefficients from the field $\mathbb{C}$. Systems \eqref{Lienard_gen}
are named in honor of the French physicist and engineer Alfred--Marie
Li\'{e}nard \cite{Lienard01}. These systems   describe oscillators with a
polynomial damping $f(x)$ and a polynomial restoring force $g(x)$. In
addition,  Li\'{e}nard differential  systems have a variety of other
applications in physics, chemistry, biology, economics, etc. For example,
these systems \eqref{Lienard_gen}  arise as traveling wave reductions of the
following general families of reaction-convection-diffusion equations
\begin{equation}
 \label{PDE_RCD}
u_t=Du_{xx}+A(u)u_x+B(u),\quad u=u(x,t),
\end{equation}
where $D$ is a diffusion coefficient, $A(u)$ describes a nonlinear convective
flux and $B(u)$ is responsible for a reaction force.

Let the variable $y$ be privileged with respect to the variable $x$, then the
function $y(x)$ satisfies the following family of Abel differential equations of the second kind
\begin{equation}
 \label{Lienard_y_x}
yy_x+f(x)y+g(x)=0.
\end{equation}

The integrability properties of  Li\'{e}nard differential systems are
investigated by various methods and in the framework of various theories. Let
us enumerate the most important studies:

\begin{enumerate}

\item local analysis \cite{Cherkas01, Christopher_Lienard,
    LV_Lienard_int, Gasull_Lienard_int, Llibre04};

\item classical Lie symmetry analysis \cite{Lakshmanan01, Lakshmanan02}
    and $\lambda$ symmetries \cite{Ruiz01};

\item local  \cite{Chiellini01, Abel01, Abel02, Lienard-Riccati,
    Choudhury_Lienard} and non-local transformations \cite{Berkovich01,
    DS2019, Guha_Lienard, Sin2020, DS2021, Sin2021};

\item differential Galois theory \cite{Morales-Ruiz01};

\item extended Prelle--Singer method \cite{Lakshmanan03} and the Darboux
    theory of integrability~\cite{Llibre06,  Cheze01, Stachowiak,
    Demina07, Demina13, Demina17, Demina16, DS2019, DS2021}.

\end{enumerate}

 A collection of
integrable and solvable subfamilies of Li\'{e}nard differential systems is
presented by A. D. Polyanin and V.~F.~Zaitsev in  \cite{Polyanin}.
The transformation $y(x)=1/w(x)$ brings Abel differential equations of the second kind
\eqref{Lienard_y_x} to Abel differential equations of the first kind
\begin{equation}
 \label{Abel}
w_x=g(x)w^3+f(x)w^2.
\end{equation}
Consequently, certain results available for such Abel differential equations
can be transferred to equations \eqref{Lienard_y_x} and related Li\'{e}nard
differential systems \cite{Chiellini01, Abel01, Abel02}. Let us note that
many scientific works dealing with the global integrability problem present
sufficient conditions of integrability. Thus, these works do not provide
classifications of  integrable Li\'{e}nard differential systems and Abel differential equations. Meanwhile,
it is an important scientific problem to find all integrable families for
fixed degrees of the polynomials $f(x)$ and $g(x)$.

This article is devoted to the study of the general integrability properties
of polynomial Li\'{e}nard differential systems. We  focus on the necessary
and sufficient conditions of Darboux and Liouvillian integrability. J. Llibre and C. Valls \cite{Llibre06} proved that Li\'{e}nard differential systems
\eqref{Lienard_gen} under the condition $\deg g\leq \deg f$ do not have
Liouvillian first integrals excluding the trivial case $g(x)=\alpha f(x)$, $\alpha\in\mathbb{C}$.
Consequently, we only need to study systems \eqref{Lienard_gen} satisfying
the restriction $\deg g>\deg f$. Throughout this article it is supposed that
$f(x)\not\equiv0$. Any Li\'{e}nard differential system is Hamiltonian with a
polynomial first integral
\begin{equation}
 \label{Lienard_Hamiltonian}
I(x,y)=y^2+2\int_{0}^{x}g(s)ds
\end{equation}
whenever the relation $f(x)\equiv0$ is valid.

We shall prove that Li\'{e}nard differential systems satisfying the restrictions $\deg g> \deg f$ and
$\deg g\neq2 \deg f+1$ are not Darboux integrable, while there are
Liouvillian integrable subfamilies. In contrast, Li\'{e}nard differential
systems in the case $\deg g=2 \deg f+1$ exhibit a variety of
rational, Darboux, and Liouvillian first integrals existing under certain restrictions on the parameters. This fact is also
recognized by many scientists \cite{Lakshmanan01, Lakshmanan02, Ruiz01, Chiellini01,
Choudhury_Lienard, Guha_Lienard, Lakshmanan03}. Our main tools
include the modern Darboux theory of integrability \cite{Singer,
Christopher}, the method of Puiseux series \cite{Demina12, Demina11}, and the
local theory of invariants \cite{Demina_Gine_Valls}.  We do not impose any
non-trivial restrictions on the coefficients of the polynomials $f(x)$ and
$g(x)$ with the exception of Li\'{e}nard differential systems from the
families $\deg g=2 \deg f+1$. We mainly study systems that are non-resonant
near infinity provided that the following restriction  $\deg g=2 \deg f+1$ is
valid. To be more precise, we say that a system \eqref{Lienard_gen} is  resonant
near infinity whenever  the highest-degree coefficients
$f_0$ and $g_0$  satisfy a resonance condition. This condition is explicitly given in Section \ref{S:Lienard_IAC} and arises only in the case $\deg g=2 \deg f+1$. Let us note that the
subset of resonant systems is of zero Lebesgue measure  in the set of all
polynomial systems \eqref{Lienard_gen} with fixed degrees of the polynomials $f(x)$ and $g(x)$. Our results
are also valid in the resonant case, but they are not complete. For all
other  polynomial Li\'{e}nard differential systems we present a complete
classification of Liouvillian integrable subfamilies. In addition, we
classify polynomial Li\'{e}nard differential systems possessing
non-autonomous Darboux first integrals and non-autonomous Jacobi last
multipliers with a time-dependent exponential factor.

 We
demonstrate that the integrability properties of Li\'{e}nard differential systems are substantially different in
the following three cases:

\begin{center}
\begin{enumerate}[(A)]

\item $\qquad$ $\deg f<\deg g<2\deg f+1$;

 \item $\qquad$ $\deg g=2\deg f+1$;

 \item $\qquad$ $\deg g>2\deg f+1$.

\end{enumerate}
\end{center}

This article is organized as follows. Sections \ref{S:Darboux},
\ref{S:Local}, and \ref{S:Lienard_IAC} contain a review of the known results and
several preliminary observations on the methods we use in the
subsequent part. In Section \ref{S:Darboux}, we describe the Darboux theory of
integrability and consider some related questions. Section \ref{S:Local} is
devoted to the method of Puiseux series and to the local theory of
invariants. In Section \ref{S:Lienard_IAC}, the results on invariant
algebraic curves of Li\'{e}nard differential systems are described. In
Section \ref{S:Lienard}, we present some integrability properties valid for a
generic polynomial Li\'{e}nard differential system. In Sections
\ref{S:Lienard_A}, \ref{S:Lienard_B}, \ref{S:Lienard_C}, we investigate the
integrability and solvability of Li\'{e}nard differential systems from
families ($A$), ($B$), and ($C$), respectively.
In Section \ref{S:Example_L24}, we consider an example: we study the Li\'{e}nard differential systems satisfying the restrictions
$\deg f=2$ and $\deg g=4$.

\section{The Darboux theory of integrability}\label{S:Darboux}

The main aim of the present section is to describe some basic aspects of the
Darboux theory of integrability. We  focus on the problem of finding
Darboux and Liouvillian first integrals of polynomial differential systems in
the plane
\begin{equation}
 \label{DS}
 x_t=P(x,y),\quad y_t=Q(x,y),
\end{equation}
where $P(x,y)$ and $Q(x,y)$ are relatively prime elements of the ring
$\mathbb{C}[x,y]$. By $\mathbb{C}[x,y]$ we denote the ring of bivariate
polynomials with complex-valued coefficients.  The vector field related to
system \eqref{DS} is defined as
\begin{equation}
 \label{VF}
 \mathcal{X}=P(x,y)\frac{\partial}{\partial x}+Q(x,y)\frac{\partial}{\partial y}.
\end{equation}

\begin{definition}
A non-constant function $I(x,y)$: $D\subset\mathbb{C}^2\rightarrow\mathbb{C}$
is called a first integral of differential system \eqref{DS} and the related
vector field $ \mathcal{X}$ on an open  subset $D\subset
\mathbb{C}^2$ if $I(x(t), y(t)) = C$ with $C$ being a constant for all values
of $t$  such that the solution $(x(t), y(t))$ of system \eqref{DS} is defined
in $D$.
\end{definition}

If $I(x,y)$ is of a class at least $C^1$ in $ D$, then $I(x,y)$ is a first
integral of differential system~\eqref{DS} if and only if $\mathcal{X} I=0$.


\begin{definition}
A non-constant function $M(x,y)$: $D\subset\mathbb{C}^2\rightarrow\mathbb{C}$
is called an integrating factor
  of differential system \eqref{DS} and the related
vector field $ \mathcal{X}$ in an open  subset $D\subset
\mathbb{C}^2$ if the differential form $M(x,y)(P(x,y)dy-Q(x,y)dx)$ is exact
in $ D$. In other words, there exists a function $I(x,y)$  of a class at least
$C^1$ in $ D$ such that the following relation is valid
$M(x,y)(P(x,y)dy-Q(x,y)dx)=dI(x,y)$.
\end{definition}

If an integrating factor $M(x,y)$ is of a class at least $C^1$ in $ D$, then it
satisfies the following linear first-order partial differential equation $
\mathcal{X}M=-\text{div}\,\mathcal{X} M$, where
$\text{div}\,\mathcal{X}=P_x+Q_y$ is the divergence of the vector field $
\mathcal{X}$.

A function $I(x,y)$ is refereed to as a Liouvillian function of two variables
$x$ and $y$ if it belongs to a Liouvillian extension of the field of rational
functions $\mathbb{C}(x,y)$ over $\mathbb{C}$. Generally speaking, any Liouvillian function can
be represented as a finite superposition of algebraic functions,
antiderivatives, and exponentials \cite{Singer}. A function $\Phi(x,y)$ is called \textit{a
Darboux function} of two variables $x$ and $y$, if it can be presented in the
form
\begin{equation}
\begin{gathered}
 \label{Darboux_function}
\Phi(x,y)=F^{d_1}_1(x,y)\ldots F^{d_K}_K(x,y)\exp\{R(x,y)\},
\end{gathered}
\end{equation}
where $F_1(x,y)\in\mathbb{C}[x,y]$, $\ldots$, $F_K(x,y)\in\mathbb{C}[x,y]$,
$R(x,y)\in\mathbb{C}(x,y)$, $d_1,\ldots, d_K\in\mathbb{C}$.  We see that any
Darboux function is a Liouvillian function. The converse  is not generally
true.

A differential system \eqref{DS} is called
\textit{Darboux (Liouvillian) integrable} if it possesses a Darboux
(Liouvillian) first integral. It is known that the problem of establishing Darboux or Liouvillian
integrability   of a differential system
  \eqref{DS}  can be reduced to the problem of
constructing all  irreducible invariant algebraic curves of \eqref{DS} and
all exponential invariants of \eqref{DS}, for more details see \cite{Zhang,
Singer, Christopher}.


\begin{definition}\label{D:IAC}
The curve $F(x,y)=0$ with $F(x,y)\in \mathbb{C}[x,y]\setminus\mathbb{C}$ is
an invariant algebraic curve of a differential system \eqref{DS}  whenever the following condition
$F_t|_{F=0}=(PF_x+QF_y)|_{F=0}=0$ is valid.
\end{definition}

If the polynomial $F(x,y)$ producing the invariant algebraic curve $F(x,y)=0$
is irreducible in $\mathbb{C}[x,y]$, then the ideal generated by $F(x,y)$ is
radical. Consequently, there exists an element $\lambda(x,y)$ of the ring
$\mathbb{C}[x,y]$ such that  $F(x,y)$ satisfies the partial differential
equation $P(x,y)F_x+Q(x,y)F_y=\lambda(x,y) F$. The polynomial $\lambda(x,y)$
is called \textit{the cofactor} of the invariant algebraic curve $F(x,y)=0$.
It is straightforward to show that the degree of $\lambda(x,y)$ is at most
$L-1$, where $L$ is the maximum between the degrees of the polynomials
$P(x,y)$ and $Q(x,y)$. We conclude that an invariant algebraic curve of
differential system \eqref{DS} is formed from solutions of the latter. A
solution of differential system \eqref{DS} has either empty intersection with
the zero set of $F(x,y)$ or it is entirely contained in $F(x,y) = 0$. The generating polynomial $F(x,y)$
of an invariant algebraic curve $F(x,y) = 0$ is refereed to as \textit{a Darboux polynomial} or \textit{an algebraic invariant}.


\begin{definition}
 A function $E(x,y)=\exp[g(x,y)/f(x,y)]$ with
the relatively prime polynomials $g(x,y)$, $f(x,y)\in\mathbb{C}[x,y]$ is called
an exponential invariant  of a differential system~\eqref{DS}  whenever the following condition
$\mathcal{X}E=\varrho(x,y)E$ is valid, where $\varrho(x,y)\in\mathbb{C}[x,y]$.
\end{definition}

The polynomial $\varrho(x,y)$ is refereed to as \textit{the cofactor} of the
exponential invariant $E(x,y)$. It is straightforward to show that the
product of the exponential invariants $E_1(x,y)$ and $E_2(x,y)$ with the
cofactors $\varrho_1(x,y)$ and $\varrho_2(x,y)$, respectively, is an
exponential invariant possessing the cofactor
$\varrho(x,y)=\varrho_1(x,y)+\varrho_2(x,y)$. It is known that the polynomial
$f(x,y)\in\mathbb{C}[x,y]\setminus\mathbb{C}$ arising in an exponential
invariant $E(x,y)=\exp[g(x,y)/f(x,y)]$  produces an  invariant algebraic
curve $f(x,y)=0$ of the system  under consideration
\cite{Christopher}.

It turns out that the study of autonomous first integrals and autonomous
integrating factors is sometimes restrictive, even if an autonomous
differential system is under consideration. In this article we do not
consider the Darboux theory of integrability for non-autonomous systems in
the general case, for more details see \cite{Llibre11, Pantazi01, Demina15}.

A non-autonomous first integral $I(x,y,t)$ and  a non-autonomous integrating
factor $M(x$, $y$, $t)$ of a differential system \eqref{DS}  are defined similarly to the autonomous case.  They
satisfy the following linear partial differential equations
$I_t+\mathcal{X}I=0$ and $M_t+\mathcal{X}M=-\text{div}\,\mathcal{X} M$,
respectively, whenever $I(x,y,t)$ and $M(x,y,t)$ are function of a class at
least $C^1$ in $D\subset\mathbb{C}^3$. Non-autonomous integrating factors are
commonly refereed to as Jacobi last multipliers or simply Jacobi multipliers.

The following  theorems are the essence of the modern Darboux theory of
integrability.

\begin{theorem}\label{T:Darboux_rat} A polynomial differential system \eqref{DS} is Darboux integrable if and only if it has a rational
integrating factor.
\end{theorem}

The fact that a Darboux integrable differential system \eqref{DS} has a
rational integrating factor was derived  by J. Chavarriga et al., see
\cite{Chavarriga_rat}. The converse statement was established by C.
Christopher et al., see \cite{Christopher_elemntary_FI}.

\begin{theorem}\label{T:Liouville} A polynomial differential system \eqref{DS} is Liouvillian integrable if and only if it has a Darboux
integrating factor.
\end{theorem}

 Theorem \ref{T:Liouville} was proved by M. F. Singer \cite{Singer}.

\begin{theorem}\label{T:L23_Non_aut_FI}
A polynomial differential system \eqref{DS}
 possesses a  first integral of the form
\begin{equation}
\begin{gathered}
 \label{FI_t_gen}
I(x,y,t)=\prod_{j=1}^{K}F^{d_j}_j(x,y)\exp\left\{\frac{S(x,y)}{R(x,y)}\right\}\exp{(\omega t)},\quad
\omega,\, d_1,\, \ldots,\, d_K\in\mathbb{C},
\end{gathered}
\end{equation}
where $F_1(x,y)$, $\ldots$, $F_K(x,y)$ are pairwise relatively prime
irreducible bivariate polynomials from the ring $\mathbb{C}[x,y]$, $S(x,y)$
and $R(x,y)$ are relatively prime bivariate polynomials from the ring
$\mathbb{C}[x,y]$, if and only if $F_1(x,y)=0$, $\ldots$, $F_K(x,y)=0$,
$R(x,y)=0$ are invariant algebraic curves of the system
and $E(x,y)=\exp\{S(x,y)/R(x,y)\}$ is an exponential invariant of
the system such that the following condition
\begin{equation}
\begin{gathered}
 \label{NDFI_gen_cond}
\sum_{j=1}^{N}d_j\lambda_j(x,y)+\varrho(x,y)+\omega =0
\end{gathered}
\end{equation}
holds identically. In this expression $\lambda_j(x,y)$ is the cofactor of the
invariant algebraic curve $F_j(x,y)=0$ and $\varrho(x,y)$ is the cofactor of
the exponential invariant $E(x,y)$.
\end{theorem}

Theorem \ref{T:L23_Non_aut_FI}  follows from the classical theory of Darboux
integrability, see~\cite{Zhang}. We  name a first integral of Theorem
\ref{T:L23_Non_aut_FI} as a non-autonomous Darboux first integral provided
that $\omega\neq 0$. If $\omega=0$, then function \eqref{FI_t_gen} gives a
Darboux first integral.

\begin{theorem}\label{T:L23_Non_aut_JLM}
Under the assumptions of Theorem \ref{T:L23_Non_aut_FI} a polynomial
differential system \eqref{DS}
 possesses a Jacobi  last multiplier of the form
\begin{equation}
\begin{gathered}
 \label{JLM_gen}
M(x,y,t)=\prod_{j=1}^{K}F^{d_j}_j(x,y)\exp\left\{\frac{S(x,y)}{R(x,y)}\right\}\exp{(\omega t)},\quad
\omega,\, d_1,\, \ldots,\, d_K\in\mathbb{C},
\end{gathered}
\end{equation}
if and only if $F_1(x,y)=0$, $\ldots$, $F_K(x,y)=0$, $R(x,y)=0$ are invariant
algebraic curves of the system  and
$E(x,y)=\exp\{S(x,y)/R(x,y)\}$ is an exponential invariant of the
system such that the following condition
\begin{equation}
\begin{gathered}
 \label{JLM_gen_cond}
\sum_{j=1}^{N}d_j\lambda_j(x,y)+\varrho(x,y)+\omega =-\text{div}\,\mathcal{X}
\end{gathered}
\end{equation}
is identically valid. In this expression $\lambda_j(x,y)$ is the cofactor of
the invariant algebraic curve $F_j(x,y)=0$ and $\varrho(x,y)$ is the cofactor
of the exponential invariant $E(x,y)$.
\end{theorem}

 Theorem \ref{T:L23_Non_aut_JLM} with the restriction $\omega=0$ was
derived by C.~Christopher \cite{Christopher}. The case $\omega\neq0$ was
considered in article \cite{Demina15}. A Jacobi last multiplier of Theorem
\ref{T:L23_Non_aut_JLM} will be referred to as  a non-autonomous
Darboux--Jacobi last multiplier whenever $\omega\neq0$.

 These theorems suggest
the following algorithm for searching  autonomous and non-autonomous Darboux
first integrals and Jacobi last multipliers:

\begin{enumerate}

\item find all relatively prime irreducible invariant algebraic curves
    and all exponential invariants with linearly independent cofactors;

\item find,  or prove the non-existence of, complex numbers $d_1$,
    $\ldots$, $d_K$, $\omega$ such that  condition \eqref{NDFI_gen_cond} or
    \eqref{JLM_gen_cond} is identically satisfied; the polynomial
    $\varrho(x,y)$ arising in conditions \eqref{NDFI_gen_cond} and
    \eqref{JLM_gen_cond} equals the sum of the cofactors of exponential
    invariants found at the first step.

\end{enumerate}

Let us note that there exist certain estimates of the number of pairwise
distinct invariants which guarantees the existence of rational, Darbour or
Liouvillian first integrals in the autonomous case. For more details see books
\cite{Zhang, Ilyashenko} and the references therein.

The first step of this algorithm is extremely difficult. This is due to the
absence of \textit{a priori} upper bounds on the degrees of bivariate polynomials
giving irreducible invariant algebraic curves. It is shown in articles
\cite{Demina06, Demina11, Demina12} that the method of Puiseux series, which
is described in the next section, can  facilitate the first step.

It is straightforward to see that integrating factors and Jacobi last
multipliers are defined modulo to the multiplication by a non-zero constant.
Two integrating factors or Jacobi last multipliers producing a constant ratio
are supposed to be equivalent. We do not distinguish them. The uniqueness
of integrating factors and Jacobi last multipliers is understood exactly in
this sense.

In general, the existence of only one independent non-autonomous first
integral is not sufficient for the complete integrability of the system under
study in the framework of the Darboux theory. However, the knowledge of a
non-autonomous Darboux first integral can be used to derive the general solutions. Several examples are given in article \cite{Demina16}. Moreover, we  demonstrate that some Li\'{e}nard differential systems from family ($B$) simultaneously have independent autonomous and non-autonomous Darboux first integrals. Eliminating the variable $y=x_t$ from these first integrals, one can find explicit expressions of the general solutions; see also article \cite{Ruiz01}, where some similar examples are presented. To conclude this section let us note that trajectories lying
in a zero set of inverse Jacobi last multipliers are of importance in the
qualitative theory of differential systems~\cite{Zhang}. Consequently, the classification of Li\'{e}nard differential systems possessing Darboux--Jacobi last multipliers seem to be a significant problem.

\section{The method of Puiseux series and the local theory of invariants}\label{S:Local}

We start with a brief review of the theory of fractional-power (or Puiseux)
series. A Puiseux series around a point $x_0\in\mathbb{C}$ can be presented
as
\begin{equation}
\begin{gathered}
 \label{Puiseux_null}
y(x)=\sum_{l=0}^{+\infty}c_l(x-x_0)^{\frac{l_0+l}{n_0}},\quad c_0\neq0,
\end{gathered}
\end{equation}
where $l_0\in\mathbb{Z}$ and $n_0\in\mathbb{N}$. It is without loss of
generality to suppose that the number $n_0$ is relatively prime with the
greatest common divisor of the numbers $\{l_0+l$ : $c_l\neq0$,
$l\in\mathbb{N}\cup\{0\}\}$. While a Puiseux series around
the point $x=\infty$ takes the form
\begin{equation}
\begin{gathered}
 \label{Puiseux_inf}
y(x)=\sum_{l=0}^{+\infty}b_lx^{\frac{l_0-l}{n_0}},\quad b_0\neq0,
\end{gathered}
\end{equation}
where  $l_0\in\mathbb{Z}$ and $n_0\in\mathbb{N}$. Again we assume  that the number $n_0$ is relatively prime with the
greatest common divisor of the numbers $\{l_0-l$ : $b_l\neq0$,
$l\in\mathbb{N}\cup\{0\}\}$.

Let us consider the algebraic equation $F(x,y)=0$, where
$F(x,y)\in\mathbb{C}[x,y]\setminus\mathbb{C}[x]$. Giving preference to one of
the variables with respect to another, a solution $y$ of this equation
viewed as a function of $x$ can be locally expanded into a convergent Puiseux
series. This statement is known as the Newton--Puiseux theorem.

The set of all Puiseux series given by \eqref{Puiseux_null} or
\eqref{Puiseux_inf} forms an algebraically closed field, which we denote by
$\mathbb{C}_{x_0}\{x\}$ or $\mathbb{C}_{\infty}\{x\}$, respectively. The ring
of polynomials in one variable $y$ with coefficients from the fields
$\mathbb{C}_{x_0}\{x\}$ or $\mathbb{C}_{\infty}\{x\}$ is denoted as
$\mathbb{C}_{x_0}\{x\}[y]$ or $\mathbb{C}_{\infty}\{x\}[y]$, respectively.

Suppose $S(x,y)$ is an element of the ring $\mathbb{C}_{\infty}\{x\}[y]$. Let
us introduce two operators of projection acting on this ring. The first
operator $\{S(x,y)\}_+ $ gives the sum of the monomials of $S(x,y)$ with
non-negative integer powers. In other words, $\{S(x,y)\}_+$ yields the
polynomial part of $S(x,y)$. Analogously, the projection
$\{S(x,y)\}_{-}=S(x,y)-\{S(x,y)\}_+$ produces the non-polynomial part of
$S(x,y)$. It is straightforward to show that these projections are linear
operators. In addition, we see that the action of the projection operators
 can be extended to the ring of  the Puiseux series in $y$ near the
point $y=\infty$ with coefficients from the field $\mathbb{C}_{\infty}\{x\}$.
We denote this ring as $\mathbb{C}_{\infty}\{x\}[\{y\}]$. Thus, we get the
relation $\{S(x,y)\}_+\in\mathbb{C}[x,y]$, where
$S(x,y)\in\mathbb{C}_{\infty}\{x\}[\{y\}]$. We endow the
fields $\mathbb{C}_{x_0}\{x\}$ and $\mathbb{C}_{\infty}\{x\}$ with the differential operator $\partial_x$. Analogously, we endow the
rings $\mathbb{C}_{x_0}\{x\}[y]$ and $\mathbb{C}_{\infty}\{x\}[y]$ with the differential operators $\partial_x$ and $\partial_y$. The action of these differential operators is standard.

The following basic statements are  proved in articles \cite{Demina07,
Demina12, Demina11}.

\begin{lemma}\label{L_Puiseux series}

Let $y(x)$ be  a Puiseux series from one of the fields
$\mathbb{C}_{x_0}\{x\}$ or $\mathbb{C}_{\infty}\{x\}$. Suppose that the series $y(x)$
satisfies the equation $F(x,y(x))=0$, where $F(x,y)\in \mathbb{C}[x,y]\setminus
\mathbb{C}[x]$ gives an invariant algebraic curve  $F(x,y)=0$ of a system
\eqref{DS}.
 Then the series $y(x)$ solves the equation
\begin{equation}
\label{ODE_y}
P(x,y)y_x-Q(x,y)=0.
\end{equation}
\end{lemma}

Using algebraic closeness of the field of Puiseux series
$\mathbb{C}_{\infty}\{x\}$, it is straightforward to find the general
structure of irreducible invariant algebraic curves and their cofactors.

\begin{theorem}\label{T:Darboux_pols}
Let $F(x,y)=0$, where $F(x,y)\in \mathbb{C}[x,y]\setminus\mathbb{C}[x]$, be an
irreducible  invariant algebraic curve of a differential system~\eqref{DS}.
Then the polynomial $F(x,y)$ and the cofactor
$\lambda(x,y)\in\mathbb{C}[x,y]$ of the curve $F(x,y)=0$ take the form
\begin{equation}
\label{General_Fl}
F(x,y)=\left\{\mu(x)\prod_{j=1}^{N}\left\{y-y_{j,\infty}(x)\right\}\right\}_{+},
\end{equation}
\begin{equation}
\label{General_Fl_cof}
\begin{gathered}
\lambda(x,y)=\left\{\sum_{m=0}^{+\infty}\sum_{j=1}^{N}\frac{\{Q(x,y)-P(x,y)(y_{j,\infty})_{x}\}(y_{j,\infty})^m}{y^{m+1}}
+P(x,y)\sum_{m=0}^{+\infty}\sum_{l=1}^{L}\frac{\nu_lx_l^m}{x^{m+1}}\right\}_{+},
\end{gathered}
\end{equation}
where  $y_{1,\infty}(x)$, $\ldots$, $y_{N,\infty}(x)$ are pairwise distinct
Puiseux series in a neighborhood of the point $x=\infty$ that satisfy
equation \eqref{ODE_y}, $x_1$, $\ldots$, $x_L$ are pairwise distinct zeros of
the polynomial $\mu(x)\in\mathbb{C}[x]$ with multiplicities $\nu_1$,
$\ldots$, $\nu_L\in\mathbb{N}$ and $L\in\mathbb{N}\cup\{0\}$. Moreover, the
degree of $F(x,y)$ with respect to $y$ does not exceed the number of distinct
Puiseux series of the from \eqref{Puiseux_inf} satisfying equation
\eqref{ODE_y} whenever this number is finite. If $\mu(x)=\mu_0$, where
$\mu_0\in\mathbb{C}$, then we suppose that $L=0$ and the first series is
absent in the expression of the cofactor~$\lambda(x,y)$.
\end{theorem}

This theorem introduces an algebraic tool enabling one to find invariant
algebraic curves explicitly, for more details see article \cite{Demina18}.

The method of Puiseux series can be also used whenever one wishes to find
exponential invariants with non-polynomial arguments. The following theorem
giving necessary conditions for an exponential invariant related to an
invariant algebraic curve to exist is proved in article~\cite{Demina07}.

\begin{theorem}\label{T:Exp_fact}
Let the polynomial $f(x,y)\in\mathbb{C}[x,y]\setminus \mathbb{C}[x]$ give an
invariant algebraic curve $f(x,y)=0$ of a differential system~\eqref{DS}. The cofactor of the invariant
algebraic curve $f(x,y)=0$ we denote by $\lambda(x,y)\in\mathbb{C}[x,y]$.
Suppose that this system admits an exponential invariant
 $E = \exp(g/f)$ related to the algebraic curve $f(x,y)=0$, then
 for each non-constant Puiseux series $y_{j,\infty}(x)$ in a neighborhood of the point $x=\infty$
 that satisfies the equation $f(x,y)=0$ there exists a number $q\in\mathbb{Q}$ such that the Puiseux series
for the function $\lambda(x,y_{j,\infty}(x))/P(x,y_{j,\infty}(x))$ in a
neighborhood of the point $x=\infty$ is
  \begin{equation}
\begin{gathered}
 \label{Exp_fact_cond}
\frac{\lambda(x,y_{j,\infty}(x))}{P(x,y_{j,\infty}(x))}=\sum_{k=n}^{+\infty}b_kx^{-\frac{k}{n}},\quad b_n=q.
\end{gathered}
\end{equation}

\end{theorem}

Now our aim is to introduce a notion of local invariant curves that are given by
elements from one of the rings $\mathbb{C}_{x_0}\{x\}[y]$ and
$\mathbb{C}_{\infty}\{x\}[y]$.
 Let $x_0$ be a point of the extended complex plane
$\overline{\mathbb{C}}=\mathbb{C}\cup\{\infty\}$. Elements of the ring
$\mathbb{C}_{x_0}\{x\}[y]$ are referred to as Puiseux polynomials. Note
that the polynomials $P(x,y)$ and $Q(x,y)$ can be regarded as Puiseux
polynomials for any~$x_0\in \overline{\mathbb{C}}$. Again giving preference to the variable $y$, we shall
deal with formal Puiseux series solutions $y=y_{x_0}(x)$ of algebraic
first-order ordinary differential equation~\eqref{ODE_y}. We can do the same
analysis choosing the variable $x$ as depended.


\begin{definition}
The curve $F(x,y)=0$, where $F(x,y)\in \mathbb{C}_{x_0}\{x\}[y]$, is called a
local invariant curve of a differential system \eqref{DS}  whenever the following condition
$F_t|_{F=0}=(PF_x+QF_y)|_{F=0}=0$ is valid.

\end{definition}

 If the element $F(x,y)\in\mathbb{C}_{x_0}\{x\}[y]$ gives a local
invariant curve $F(x,y)=0$
 of differential
system \eqref{DS}, then there exists the Puiseux polynomial
$\lambda(x,y)\in\mathbb{C}_{x_0}\{x\}[y]$ such that the following equation
$\mathcal{X}F(x,y)=\lambda(x,y) F(x,y)$ is satisfied  \cite{Demina_Gine_Valls}. By analogy with the
algebraic case, the Puiseux polynomials $F(x,y)$ and $\lambda(x,y)$  will be
called a local invariant and the cofactor of the related local invariant and
local curve, respectively.  The theory of local invariants is introduced in
article \cite{Demina_Gine_Valls}, where all the statements given below are
proved.

\smallskip

In what follows we call local invariants and associated local invariant
curves given by the polynomials $F(x,y)=h_{x_0}(x)$ and $F(x,y)=y-y_{x_0}(x)$, $h_{x_0}(x)$, $y_{x_0}(x)\in\mathbb{C}_{x_0}\{x\}$ \textit{elementary}.
Any  local invariant $F(x,y)$ can be represented as a product of elementary
local invariants and the cofactor of the local invariant  $F(x,y)$ is the sum
of the cofactors of all the factors.

\begin{theorem}\label{T:inv_curve_prim}
Let $y_{x_0}(x)$ be a Puiseux series near the point $x=x_0$. The  element
$F(x,y)=y-y_{x_0}(x)$ of the ring $\mathbb{C}_{x_0}\{x\}[y]$   gives a local
invariant curve of differential system \eqref{DS} if and only if the series
$y=y_{x_0}(x)$ satisfies equation~\eqref{ODE_y}.
\end{theorem}

It is straightforward to obtain an explicit expression of the cofactor
similar to that presented in Theorem \ref{T:Darboux_pols}. We introduce the
operator of projection $\left\{S(x,y)\right\}_{+,\,y}$ yielding the
polynomial part with respect to $y$ of the element
$S(x,y)\in\mathbb{C}_{x_0,\infty} \{x\}$ $[\{y\}]$. The symbol
$\mathbb{C}_{x_0,\infty}\{x\}[\{y\}]$ denotes the ring of Puiseux series in
$y$ near the point $y=\infty$ with coefficients from the field
$\mathbb{C}_{x_0}\{x\}$. In other words, we obtain the relation
$\left\{S(x,y)\right\}_{+,\,y}\in\mathbb{C}_{x_0}\{x\}[y]$.

\begin{theorem}\label{T:coff_local2}
Let $F(x,y)=y-Y_{x_0}(x)$, $Y_{x_0}(x)\in\mathbb{C}_{x_0}\{x\}$ give an
elementary local invariant  curve $F(x,y)=0$ of differential system
\eqref{DS}, then its cofactor takes the form
\begin{equation}
 \label{Gen_coff1-2}
\lambda(x,y)=\left\{\sum_{m=0}^{+\infty}\frac{\{Q(x,y)-P(x,y)(Y_{x_0}(x))_x\}Y_{x_0}^m(x)}{y^{m+1}}\right\}_{+,\,y}.
\end{equation}
\end{theorem}

Similarly to the case of local invariant curves, we can introduce a concept
of local exponential invariants.


\begin{definition}
 A function $E(x,y)=\exp[g(x,y)/f(x,y)]$ with
relatively prime elements $g(x,y)$ and $f(x,y)$ of the ring
$\mathbb{C}_{x_0}\{x\}[y] $ is called \textit{a local exponential invariant}
of a differential system \eqref{DS} if
$E(x,y)$ satisfies the partial differential equation
$\mathcal{X}E(x,y)=\varrho(x,y)E(x,y)$, where $\varrho(x,y)$ is a Puiseux
polynomial.
\end{definition}

Again the Puiseux polynomial $\varrho(x,y)$ will be referred to as the
cofactor of the local exponential invariant $E(x,y)$. Note that exponential invariants belong to a differential extension of the field of fractions of the ring $\mathbb{C}_{x_0}\{x\}[y] $.

\begin{lemma}\label{L:exp_factor_inv_curve}
Let $E(x,y)=\exp[g(x,y)/f(x,y)]$, $f(x,y)\in\mathbb{C}_{x_0}\{x\}[y] \setminus \mathbb{C}_{x_0}\{x\}$  be a local
 exponential invariant  of a differential system \eqref{DS}, then  $f(x,y)=0$  is
 a local
invariant curve of the system  in question.
\end{lemma}

We shall say that local exponential invariants
\begin{equation}
\label{Gen_coff11}
\begin{gathered}
E(x,y)=\exp\left[g_l(x)y^l\right],\quad g_l(x)\in \mathbb{C}_{x_0}\{x\},\quad l\in\mathbb{N}\cup\{0\};\hfill \\
E(x,y)=\exp\left[ \frac{u(x)}{ \{y-y_{x_0}(x)\}^n}\right],\quad  y_{x_0}(x), \,\,u(x)\in \mathbb{C}_{x_0}\{x\},\quad n\in\mathbb{N}
\end{gathered}
\end{equation}
 are \textit{elementary
local exponential invariants}.  Since $E^{\alpha}(x,y)$ with
$\alpha\in\mathbb{C}$ is a local exponential invariant whenever so does
$E(x,y)$, the highest-order coefficients of the series $g_l(x)$ and $u(x)$
can be fixed in advance. Any  local exponential invariant $E(x,y)$ is as a product of elementary local exponential invariants and the
cofactor of the local invariant  $E(x,y)$ is the sum of the cofactors of all
the factors.

We conclude  that if system \eqref{DS} has a global invariant (algebraic or
exponential), then for any $x_0\in\mathbb{\overline{C}}$ this invariant
must be a product of local elementary  invariants in such a way that the
sum of their cofactors is a true, not Puiseux, polynomial. This observation
is used in subsequent sections.

\section{Invariant algebraic curves of  polynomial Li\'{e}nard differential systems}\label{S:Lienard_IAC}

With the aim to study the integrability properties of polynomial Li\'{e}nard
differential systems we need information on the existence  of
their invariant algebraic curves. The general structure of bivariate
polynomials giving irreducible invariant algebraic curves is derived in
articles \cite{Demina12, Demina18}.

\begin{theorem}\label{T1:Lienard_gen}
Let $F(x,y)=0$, $F(x,y)\in \mathbb{C}[x,y]\setminus\mathbb{C}$ be an
irreducible invariant algebraic curve of a Li\'{e}nard differential system
 \eqref{Lienard_gen} from family ($A$). Then the
polynomial $F(x,y)$ and the cofactor $\lambda(x,y)$ of the invariant algebraic curve
$F(x,y)=0$ take the form
\begin{equation}
\begin{gathered}
 \label{Lienard1_F}
F(x,y)=\left\{\prod_{j=1}^{N-k}\left\{y-y^{(1)}_{j,\infty}(x)\right\}\left\{y-y^{(2)}_{N,\infty}(x)\right\}^{k}\right\}_{+},
\end{gathered}
\end{equation}
\begin{equation}
\begin{gathered}
 \label{Lienard1_cof}
\lambda(x,y)=-Nf-(N-k)q_x-kp_x,
\end{gathered}
\end{equation}
where $k=0$ or $k=1$, $N\in\mathbb{N}$, and $y^{(1)}_{1,\infty}(x)$,
$\ldots$, $y^{(1)}_{N-1,\infty}(x)$, $y^{(2)}_{N,\infty}(x)$ are the series
\begin{equation}
\begin{gathered}
 \label{Lienard1_F_series}
(I):\,y^{(1)}_{j,\infty}(x)=q(x)+\sum_{l=0}^{+\infty}b_{m+1+l}^{(j)}x^{-l},\quad j=1,\ldots, N-k;\\
(II):\,y^{(2)}_{N,\infty}(x)=p(x)+\sum_{l=0}^{+\infty}b_{n-m+l}^{(N)}x^{-l}.\hfill
\end{gathered}
\end{equation}
 The coefficients of the series of type (II) and of the polynomials
\begin{equation}
\begin{gathered}
 \label{Lienard1_F_series_q}
q(x)=-\frac{f_0}{m+1}x^{m+1}+\sum_{l=1}^{m} q_{m+1-l} x^l\in\mathbb{C}[x],\\
 p(x)=-\frac{g_0}{f_0}x^{n-m}+\sum_{l=1}^{n-m-1} p_{n-m-l} x^l\in\mathbb{C}[x]
\end{gathered}
\end{equation}
 are uniquely determined.
The coefficients $b_{m+1}^{(j)}$, $j=1,\ldots, N-k$ are pairwise distinct.
All other coefficients
 $b_{m+1+l}^{(j)}$, $l\in\mathbb{N}$ are expressible via $b_{m+1}^{(j)}$, where $j=1,\ldots, N-k$.
 The corresponding product in  \eqref{Lienard1_F} is unit whenever $k=1$ and $N=1$.
\end{theorem}

\begin{theorem}\label{T:L_two_curves}
A Li\'{e}nard differential system \eqref{Lienard_gen} from family ($A$)  with  fixed
coefficients of the polynomials $f(x)$ and $g(x)$ has at most two distinct
irreducible invariant algebraic curves simultaneously. If the two distinct
irreducible invariant algebraic curves exist, then the first has $k=1$ in
representation \eqref{Lienard1_F} and the second is given by a first-degree polynomial
with respect to $y$ and takes the form $y-q(x)-z_0=0$.
\end{theorem}

\textit{Corollary.} Li\'{e}nard differential systems \eqref{Lienard_gen} from
family ($A$) are not integrable with a rational first integral.

Theorems \ref{T1:Lienard_gen} and \ref{T:L_two_curves} are proved in article
\cite{Demina12}.

The structure of polynomials producing invariant algebraic curves for
Li\'{e}nard differential  systems from family ($B$) is in strong correlation
with the properties  of the following quadratic equation
 \begin{equation}\label{eq:DP2_5_2}
p^2-\varrho p+(m+1)\varrho=0,
\end{equation}
 where we have introduced  notation
 \begin{equation}\label{eq:DP2_5_3}
\varrho=4(m+1)-\frac{f_0^2}{g_0}.
\end{equation}
The set of all positive rational numbers will be denoted as $\mathbb{Q}^+$.
Let $p_1$ and $p_2$ be the roots of equation \eqref{eq:DP2_5_2}.

\begin{theorem}\label{T:Lienard_degenerate}
Suppose $F(x,y)=0$, $F(x,y)\in \mathbb{C}[x,y]\setminus\mathbb{C}$ is an
irreducible  invariant algebraic curve of a Li\'{e}nard differential system
\eqref{Lienard_gen} from family ($B$) and equation \eqref{eq:DP2_5_2} has no
solutions in $\mathbb{Q}^+$. One of the following statements holds.

\begin{enumerate}

\item   If $p_1p_2\neq0$, then the polynomial
    $F(x,y)$ is of degree  at most two with respect to $y$ and
\begin{equation}
\begin{gathered}
 \label{F_L2_5_1}
F(x,y)=\left\{\left\{y-y^{(1)}_{\infty}(x)\right\}^{s_1}\left\{y-y^{(2)}_{\infty}(x)\right\}^{s_2}\right\}_{+},\hfill\\
\lambda(x,y)=-(s_1+s_2)f(x)-\left\{s_1\left(y^{(1)}_{\infty}\right)_x+s_2\left(y^{(2)}_{\infty}\right)_x\right\}_{+},\hfill\\
y^{(k)}_{+\infty}(x)=\sum_{l=0}^{\infty}b_l^{(k)}x^{m+1-l},\quad b_0^{(k)}=\frac{f_0}{p_k-2(m+1)},\quad k=1, 2,
\end{gathered}
\end{equation}
where $s_1$ and $s_2$ are either $0$ or $1$ independently, $s_1+s_2>0$.
The Puiseux series $y^{(k)}_{\infty}(x)$, $k=1$, $2$ are  Laurent series
and possess uniquely determined coefficients.

\item   If  $p_1p_2=0$, then $p_1=p_2=0$. The polynomial $F(x,y)$ and its cofactor
    $\lambda(x,y)$ take the form
\begin{equation}
\begin{gathered}
 \label{F_L2_5_4}
F(x,y)=y+\frac{f_0}{2(m+1)}x^{m+1}-\sum_{l=1}^{m+1}b_{l}x^{m+1-l},\hfill\\
\lambda(x,y)=-f(x)+\frac{f_0}{2}x^{m}-\sum_{l=1}^{m}(m+1-l)b_{l}x^{m-l},\hfill\\
\end{gathered}
\end{equation}
where the coefficients $b_1$, $\ldots $, $b_{m+1}$ are uniquely
determined. In addition, the following relation $4(m+1)g_0-f_0^2=0$ is
valid.

\end{enumerate}

\end{theorem}

This theorem is proved in  \cite{Demina18}. Note that
the Li\'{e}nard differential systems such that the related equation
\eqref{eq:DP2_5_2} has a positive rational solution  are also studied in~\cite{Demina18}. We do not reproduce the related results here.

The structure of polynomials giving  invariant algebraic curves and their
cofactors for Li\'{e}nard differential systems  satisfying the condition
$\deg g>2\deg f+1$ is derived in~\cite{Demina18}. The following
theorem is valid.

\begin{theorem}\label{T1:Lienard_2m+1}
Let $F(x,y)=0$, $F(x,y)\in \mathbb{C}[x,y]\setminus\mathbb{C}$ be an
irreducible  invariant algebraic curve of a Li\'{e}nard differential system~\eqref{Lienard_gen} from family ($C$). Then the polynomial $F(x,y)$ and the
cofactor $\lambda(x,y)$ of the invariant algebraic curve $F(x,y)=0$ take the
form
\begin{equation}
\begin{gathered}
 \label{Lienard2_F}
F(x,y)=\left\{\prod_{j=1}^{N_1}\left\{y-y^{(1)}_{j,\infty}(x)\right\}\prod_{j=1}^{N_2}\left\{y-y^{(2)}_{j,\infty}(x)\right\}\right\}_{+},
\end{gathered}
\end{equation}
\begin{equation}
\begin{gathered}
 \label{Lienard2_cof}
\lambda(x,y)=-(N_1+N_2)f-\left\{N_1h^{(1)}_x+N_2h^{(2)}_x\right\}_+
\end{gathered}
\end{equation}
where the Puiseux series $y^{(1,2)}_{j,\infty}(x)$ are given by the relations
\begin{equation}
\begin{gathered}
 \label{Lienard2_F_series}
y^{(1,2)}_{j,\infty}(x)=h^{(1,2)}(x)+\sum_{k=2(n+1)}^{+\infty}b^{(1,2)}_{k,\,j}x^{\frac{n+1}{2}-\frac{k}{2}},\quad
h^{(1,2)}(x)=\sum_{k=0}^{2n+1}b^{(1,2)}_{k}x^{\frac{n+1}{2}-\frac{k}{2}},
\end{gathered}
\end{equation}
 $N_1$, $N_2\in\mathbb{N}\cup\{0\}$, and $N_1+N_2\geq1$. The coefficients
$b_{2(n+1),\,j}^{(1,2)}$ with the same upper index are pairwise distinct and
all the coefficients
 $b_{l,\,j}^{(1,2)}$ with $l>2(n+1)$ are expressible via $b_{2(n+1),\,j}^{(1,2)}$.
 If $n$ is an odd number, then the corresponding Puiseux series are Laurent series and
 $b_{2l-1}^{(1,2)}=0$, $b_{2l-1,\,j}^{(1,2)}=0$, when  $l\in\mathbb{N}$.
 In addition, $N_k=1$ whenever $n$ is odd and $N_l=0$, where $k$, $l=1$, $2$ and $k\neq l$.
  If $n$ is an even number, then~$N_1=N_2$.

\end{theorem}

All the Puiseux series arising in Theorems \ref{T1:Lienard_gen},  \ref{T:Lienard_degenerate}, \ref{T1:Lienard_2m+1} solve the equation \eqref{Lienard_y_x} related to the Li\'{e}nard differential system  under consideration. The following theorem provides the necessary and sufficient conditions for a
Li\'{e}nard differential system \eqref{Lienard_gen} to have an invariant
algebraic curve.

\begin{theorem}\label{T:Darboux_pols_computation_Lienard}
The  polynomial $F(x,y)\in \mathbb{C}[x,y]\setminus\mathbb{C}$ of degree
$N>0$ with respect to $y$ gives an invariant algebraic curve $F(x,y)=0$ of a
Li\'{e}nard differential system \eqref{Lienard_gen}  if and only if there
exists $N$ Puiseux series $y_{1,\infty}(x)$, $\ldots$, $y_{N,\infty}(x)$ from
the field $\mathbb{C}_{\infty}\{x\}$ that solve equation~\eqref{Lienard_y_x}
and satisfy the conditions
\begin{equation}
\begin{gathered}
 \label{Existance_F_conditions_Lienard}
\left\{\sum_{j=1}^{N}y_{j,\infty}(x)\right\}_{-}=0.
\end{gathered}
\end{equation}
\end{theorem}

The theorems presented in this section essentially are based on the
method of Puiseux series.

\section{Integrability of a generic  Li\'{e}nard differential system}\label{S:Lienard}

We begin this section by investigating the existence of exponential
invariants with polynomial arguments.

\begin{lemma}\label{L:Lienard_exp_inv1}
If a Li\'{e}nard differential system \eqref{Lienard_gen} satisfying the
condition $\deg g>\deg f$ has an exponential invariant of the form
$E(x,y)=\exp[h(x,y)]$, $h(x,y)\in\mathbb{C}[x,y]$ with the cofactor
$\varrho(x,y)\in\mathbb{C}[x,y]$ of degree at most $\deg g-1$, then the
polynomial $\varrho(x,y)$ is divisible by $y$ in the ring $\mathbb{C}[x,y]$.
\end{lemma}

\begin{proof}
It is straightforward to show that the polynomial $h(x,y)$ satisfies the
following linear inhomogeneous partial differential equation
\begin{equation}
 \label{Lienard_exp_inv1_1}
yh_x-\{f(x)y+g(x)\}h_y=\varrho(x,y).
\end{equation}
Let us represent the functions $h(x,y)$ and $\varrho(x,y)$ as polynomials in
$y$ with polynomial in $x$ coefficients. Substituting the expressions
$h(x,y)=h_0(x)+h_1(x)y+\ldots$ and
$\varrho(x,y)=\varrho_0(x)+\varrho_1(x)y+\ldots$
 into equation
\eqref{Lienard_exp_inv1_1} and selecting the coefficients of $y^0$,  we get
the relation $g(x)h_1(x)=-\varrho_0(x)$. Since the degree of the polynomial
$\varrho_0(x)$ is at most $\deg g-1$, we conclude that $h_1(x)\equiv0$ and
$\varrho_0(x)\equiv 0$. As a result the polynomial $\varrho(x,y)$ is given by
the relation $\varrho(x,y)=(\varrho_1(x)+\ldots)y$. This completes the proof.

\end{proof}

As a direct consequence of this lemma, we  establish that Darboux or
Liouvillian integrable Li\'{e}nard differential systems necessarily have
invariant algebraic curves.

\begin{theorem}\label{T:Lienard_nonintegrability}
If a Li\'{e}nard differential system \eqref{Lienard_gen} satisfying the
condition $\deg g>\deg f$ has no invariant algebraic curves, then this system
is not integrable with a Darboux or Liouvillian first integral.
\end{theorem}

\begin{proof}
It is a simple result of the Darboux theory of integrability that a
differential system \eqref{DS} without invariant algebraic curves cannot have
rational first integrals.

Suppose that a Li\'{e}nard differential system without invariant algebraic
curves possesses a Darboux  first integral. Then this first integral is given
by an exponential invariant with a zero cofactor. Consequently, the argument
of the exponential function in the invariant provides a rational first
integral. This is a contradiction.

If a Li\'{e}nard differential system \eqref{Lienard_gen} without invariant
algebraic curves is Liouvillian integrable, then there exists a Darboux
integrating factor and this  integrating factor equals some exponential
invariant of the form $E(x,y)=\exp[h(x,y)]$, $h(x,y)\in\mathbb{C}[x,y]$. Calculating the divergence of
the vector field
\begin{equation}
 \label{Lienard_Inegrability_VF}
\mathcal{X}=y\frac{\partial}{\partial x}-(f(x)y+g(x))\frac{\partial}{\partial y},
\end{equation}
we get: $\text{div}\,\mathcal{X}=-f(x)$. The divergence is independent of $y$,
while by Lemma \ref{L:Lienard_exp_inv1} the cofactor of the exponential
invariant $E(x,y)=\exp[h(x,y)]$ is divisible by $y$ provided that $\deg
g>\deg f$. Consequently, condition \eqref{NDFI_gen_cond} in the autonomous
case is not valid.

\end{proof}

Thus, we conclude  that Liouvillian integrable Li\'{e}nard differential
systems~\eqref{Lienard_gen} with the restriction   $\deg f<\deg g$ must
have at least one invariant algebraic
 curve.

\begin{theorem}\label{T:Lienard_Intgerability_generic} A generic Li\'{e}nard differential system
\eqref{Lienard_gen} with fixed degrees of the polynomials $f(x)$ and $g(x)$
 is neither Darboux nor Liouvillian
integrable provided that the following restrictions $\deg g>\deg f$ and
($\deg f, \deg g$) $\neq (0,1)$ are valid.
\end{theorem}

\begin{proof}

Let us denote the set of all Li\'{e}nard differential systems with fixed
degrees of the polynomials $f(x)$ and $g(x)$ as $L_{m,n}$. Any particular
Li\'{e}nard differential system can be identified with a point in
$\mathbb{C}^{m+n}\times(\mathbb{C}\setminus\{0\})^2$. In view of Theorem
\ref{T:Lienard_nonintegrability}, we need to prove that the subset of
Li\'{e}nard differential systems~\eqref{Lienard_gen} without  invariant
algebraic curves is of full Lebesgue measure in the set  $L_{m,n}$. Note that
here we speak about finite invariant algebraic curves. Extending systems
\eqref{Lienard_gen}  from the complex plane $\mathbb{C}^2$ to the complex
projective plane $\mathbb{C}P^2$, we see that the infinite line is always
invariant for Li\'{e}nard differential systems.

We begin by considering systems from family ($A$). The subset of Li\'{e}nard
differential systems  such that the  related equations~\eqref{Lienard_y_x}
possess a family of formal Puiseux series solutions $y^{(1)}_{\infty}(x)$
has Lebesgue measure zero. Indeed, there always arises a compatibility
condition enabling this family of series to exist. We only need  to show that
the compatibility condition cannot be identically satisfied. For this aim we
track the appearance of the coefficient $g_{n-m}$ in the series
$y^{(1)}_{\infty}(x)$. We use the following representation
\begin{equation}
 \label{Lienard_Inegrability_Generic_PS}
y^{(1)}_{\infty}(x)=\sum_{l=0}^{+\infty}v_l(x)\left(g_{n-m}\right)^l,
\end{equation}
where $\{v_l(x)\}$ are elements of the field $\mathbb{C}_{\infty}\{x\}$. The
compatibly condition arises, when one tries to find the coefficient of $x^0$.
Substituting representation \eqref{Lienard_Inegrability_Generic_PS} into
equations \eqref{Lienard_y_x} and setting to zero the coefficients of
$g^l_{n-m}$ for $l=0$, $1$, we find the  ordinary differential
equations
\begin{equation}
 \label{Lienard_Inegrability_Generic_PS2}
v_0v_{0,x}+f(x)v_0+\hat{g}(x)=0,\quad v_0v_{1,x}+(f(x)+v_{0,x})v_1+x^m=0,
\end{equation}
where we use the  designation $\hat{g}(x)=g(x)-g_{n-m}x^m$. The dominant
behavior of the series $v_0(x)$  near the point $x=\infty$ is
\begin{equation}
 \label{Lienard_Inegrability_Generic_PS3}
v_0(x)=-\frac{f_0}{m+1}x^{m+1}+o(x^{m+1}),\quad x\rightarrow\infty.
\end{equation}
Analyzing the ordinary differential equation for the series $v_1(x)$, we see
that this equation should have a solution with the dominant behavior
$v_1(x)=e_0x^0$, $e_0\in\mathbb{C}\setminus\{0\}$, $x\rightarrow\infty$, but
it does not. In addition, the equations for $v_l(x)$, $l\geq 2$ have a zero
solution whenever $v_1(x)$ is zero.  Thus, the compatibility condition
enabling the series $y^{(1)}_{\infty}(x)$ to exist is given by a first-degree
polynomial with respect to $g_{n-m}$ and cannot hold identically.
Consequently, it is necessary to consider systems with invariant algebraic
curves given by bivariate polynomials not involving the series
$y^{(1)}_{\infty}(x)$ into the factorization in the ring
$\mathbb{C}_{\infty}\{x\}[y]$. In view of Theorem \ref{T1:Lienard_gen}, the
corresponding irreducible invariant algebraic curve is given by the
polynomial $F(x,y)=y-p(x)-z_1$, where $z_1\in\mathbb{C}$ and the polynomial
$p(x)$  of degree $n-m$ is described by
expression~\eqref{Lienard1_F_series_q}. Now we consider the subset of
Li\'{e}nard differential systems with the invariant algebraic curve
$y-p(x)-z_1=0$. Since  the related equation~\eqref{Lienard_y_x} possesses the
polynomial solution $y(x)=p(x)+z_1$, we find the polynomial $g(x)$. The
result is $g(x)=-(p+z_1)(p_x+f)$. We see that the dimension of the subset of Li\'{e}nard differential systems
under consideration is less than the dimension of $L_{m,n}$.

We turn to systems from family ($B$). It follows from Theorem
\ref{T:Lienard_degenerate} that if a generic Li\'{e}nard differential system
from family ($B$) possesses irreducible invariant algebraic curves, then
their generating polynomials are of degrees either $1$ or $2$ with respect to
$y$. Let us suppose that there exists the irreducible invariant algebraic
curve $y-q_l(x)=0$. In this expression $q_l(x)=\{y^{(l)}_{\infty}(x)\}_+$,
$l=1$, $2$ is a polynomial of degree $m+1$. By analogue with systems from
family ($A$), we can find the polynomial $g(x)$. We see from the relation
$g(x)=-q_l(f+q_{l,x})$ that the subset of Li\'{e}nard differential systems with the invariant
algebraic curve $y-q_l(x)=0$ is of dimension $2m+3$, while the dimension of
$L_{m,n}$ is $m+n+2=3(m+1)$. Hence the subset in question is of zero Lebesgue
measure in $L_{m,n}$ provided that $m>0$. Now we suppose that Li\'{e}nard
differential systems from family ($B$) have the hyperelliptic invariant
algebraic curve $(y+u(x))^2+w(x)=0$. By Theorem \ref{T:Lienard_degenerate}
this curve is unique. In addition, we see that  the polynomial $u(x)$ is of
degree $m+1$ and the polynomial $w(x)$ is of degree at most $2m+2$. Following
H. \.{Z}o{\l}\c{a}dec \cite{Zoladec01}, we substitute the bivariate
polynomial $F(x,y)=(y+u(x))^2+w(x)$ into the partial differential equation
\begin{equation}
 \label{Lienard_Inegrability_Generic_PDE}
yF_x-(f(x)y+g(x))F_y-\lambda(x,y)F=0.
\end{equation}
Further, we set to zero the coefficients at different powers of $y$. Since
the cofactor $\lambda(x,y)$ is independent of $y$, we express the polynomials
$f(x)$ and $g(x)$ via $u(x)$ and $v(x)$. The result~is
\begin{equation}
 \label{Lienard_Inegrability_Generic_PS4}
f(x)=u_x+\frac{uw_x}{2w},\quad g(x)=\frac{w_x}{2}+\frac{u^2w_x}{2w}.
\end{equation}
We see from these expressions that any zero of the polynomial $w(x)$ is also
a zero of the polynomial $u(x)$. The polynomial $u(x)$ is parameterized by at
most $m+2$ parameters and the polynomial $w(x)$ adds only one new parameter.
Thus, we conclude that the dimension of Li\'{e}nard differential systems from
family ($B$) with the hyperelliptic invariant algebraic curve
$(y+u(x))^2+w(x)=0$ is at most $m+3$. While the dimension of $L_{m,2m+1}$ is
$3m+3$.

Finally, we are left with systems from family ($C$). Let us track the
dependence of the Puiseux series $y^{(1)}_{\infty}(x)$ and
$y^{(2)}_{\infty}(x)$ on the coefficient $f_m$. We introduce the
representation
\begin{equation}
 \label{Lienard_Inegrability_Generic_PS_Cn}
y^{(k)}_{\infty}(x)=\sum_{l=0}^{+\infty}v_l^{(k)}(x)\left(f_{m}\right)^l,\quad k=1,2,
\end{equation}
where $\{v_l^{(k)}(x)\}$ are elements of the field
$\mathbb{C}_{\infty}\{x\}$. Let us introduce the following designation:
$\hat{f}(x)=f(x)-f_m$. Substituting these representations
 into equation~\eqref{Lienard_y_x} and
setting to zero the coefficients of $f^l_{m}$ with $l=0$, $1$, we find the
ordinary differential equations
\begin{equation}
 \label{Lienard_Inegrability_Generic_PS2_Cn}
v_0v_{0,x}+\hat{f}(x)v_0+g(x)=0,\quad v_0v_{1,x}+(\hat{f}(x)+v_{0,x})v_1+v_0=0.
\end{equation}
Note that we omit the upper index. Let us suppose that $n$ is even. Puiseux
series from the field $\mathbb{C}_{\infty}\{x\}$ that satisfy these equations
are of the form
\begin{equation}
 \label{Lienard_Inegrability_Generic_PS3_Cn}
 \begin{gathered}
v_0^{(k)}(x)=\sum_{j=0}^{+\infty}c_j^{(k)}x^{\frac{n+1-k}{2}},\quad k=1,2,\quad
c_0^{(1,2)}=\pm \frac{\sqrt{-2(n+1)g_0}}{n+1};\\
v_1^{(k)}(x)=\sum_{j=0}^{+\infty}e_j^{(k)}x^{\frac{2-k}{2}},\quad k=1,2,\quad
e_0^{(1,2)}=- \frac{2}{n+1}.\hfill
\end{gathered}
\end{equation}
Theorem \ref{T:Darboux_pols_computation_Lienard} gives the following
necessary condition  for an invariant algebraic curve $F(x,y)=0$ to exist:
$N_1b_{n+3}^{(1)}+N_2b_{n+3}^{(2)}=0$. Recall that $b_{n+3}^{(k)}$ is a
coefficient of the family of series $y^{(k)}_{\infty}(x)$ and $N_k$ is the
number of times the family of series $y^{(k)}_{\infty}(x)$ enters the
factorization of $F(x,y)$ in the ring $\mathbb{C}_{\infty}\{x\}[y]$. Using
Theorem \ref{T1:Lienard_2m+1}, we get $N_1=N_2$. Calculating the first five
coefficients of the series $v_1^{(k)}(x)$, we see that
$e_{4}^{(1)}+e_{4}^{(2)}$ is not identically zero for any non-zero
polynomial $f(x)$ of degree $m<(n-1)/2$. Hence the expression
$N_1(b_{n+3}^{(1)}+b_{n+3}^{(2)})$ is a polynomial with respect to $f_m$
possessing a non-zero coefficient of $f_m$ in the generic case. Now we assume
that $n$ is odd. The Puiseux series $v_0^{(k)}(x)$ and $v_1^{(k)}(x)$ solving
equations \eqref{Lienard_Inegrability_Generic_PS2_Cn} are of the form
\eqref{Lienard_Inegrability_Generic_PS3_Cn}, where the coefficients
$\left\{c_j^{(k)}\right\}$ and $\left\{e_j^{(k)}\right\}$ with odd lower
indices equal zero. Further, we again consider the necessary condition
$N_1b_{n+3}^{(1)}+N_2b_{n+3}^{(2)}=0$, but now the numbers $N_1$ and $N_2$
may not be equal. Calculating the first three non-trivial coefficients of the
series $v_1^{(k)}(x)$ , we see that the condition
$N_1e_{4}^{(1)}+N_2e_{4}^{(2)}$ is not identically zero for any numbers
$N_1$, $N_2\geq0$, $N_1+N_2>0$ and any non-zero polynomial $f(x)$ of degree
$m<(n-1)/2$. Consequently, the expression $N_1b_{n+3}^{(1)}+N_2b_{n+3}^{(2)}$,
regarded as a polynomial with respect to $f_m$, possesses a non-zero
coefficient of $f_m$ in the generic case. We conclude that a generic
Li\'{e}nard differential system from family ($C$) has no finite invariant
algebraic curves.

\end{proof}

It is straightforward to see that if $(\deg f,\deg g)=(0,1)$, then the
associated Li\'{e}nard differential systems are linear. They always have
invariant lines and are Darboux integrable.

\section{Integrability of  Li\'{e}nard differential systems from family ($A$)}\label{S:Lienard_A}

We have proved in the previous section that if a polynomial Li\'{e}nard differential system is Liouvillian integrable, then it has at least one invariant algebraic curve. Let us investigate the existence of exponential invariants with
non-polynomial arguments. Here and in what follows we use the designations of
Theorem \ref{T1:Lienard_gen}. In particular, the polynomials $q(x)$ and
$p(x)$ give the initial parts of the
    Puiseux series near the point $x=\infty$ that solve equation
    \eqref{Lienard_y_x} whenever $\deg f<\deg g<2\deg f+1$. These polynomials are presented in expression \eqref{Lienard1_F_series_q}.

\begin{lemma}\label{L:Lienard_exp_inv2}
Li\'{e}nard  differential  systems  \eqref{Lienard_gen} from family ($A$)
 do not have exponential invariants of the form
$E(x,y)=\exp\left\{h(x,y)/r(x,y)\right\}$, where $h(x,y)\in\mathbb{C}[x,y]$
and $r(x,y)\in\mathbb{C}[x,y]\setminus\mathbb{C}$ are coprime polynomials.
\end{lemma}

\begin{proof}

The proof is by contradiction. Let $E(x,y)=\exp\left\{h(x,y)/r(x,y)\right\}$
 be an exponential
invariant of a Li\'{e}nard differential system \eqref{Lienard_gen} from
family ($A$). Since the polynomial $r(x,y)$ is not a constant, we see that
$r(x,y)=0$ is an invariant algebraic curve of the differential system under
consideration \cite{Christopher}.  According to the results of Theorem
\ref{T1:Lienard_gen} we can represent the polynomial $r(x,y)$ in the form
$r(x,y)=F_1^{n_1}(x,y)F_2^{n_2}(x,y)$, where $n_1$, $n_2\in\mathbb{N}_0$,
$n_1+n_2>0$, and the polynomials $F_1(x,y)$ and $F_2(x,y)$ produce invariant
algebraic curves of the related differential system. These polynomials are
given by expression \eqref{Lienard1_F} with $N=1$, $k=0$ (the invariant
algebraic curve $F_1(x,y)=0$) and by expression \eqref{Lienard1_F} with
$N\in\mathbb{N}$, $k=1$ (the invariant algebraic curve $F_2(x,y)=0$).
Relation \eqref{Lienard1_cof} yields the explicit expressions of the
cofactors
\begin{equation}
 \label{Lienard_exp_inv1_cof1}
 \begin{gathered}
\lambda_1(x,y)=-f(x)-q_x(x)=o(x^m),\quad x\rightarrow \infty,\hfill \\
\lambda_2(x,y)=-Nf(x)-(N-1)q_x(x)-p_x(x)
=-f_0x^m+o(x^m),\quad x\rightarrow \infty.
\end{gathered}
\end{equation}
related to these invariant algebraic curves, respectively. The cofactor of
the invariant algebraic curve $r(x,y)=0$ equals
$\lambda(x,y)=n_1\lambda_1(x,y)+n_2\lambda_2(x,y)$.

If the Li\'{e}nard differential system under consideration possesses only one
irreducible invariant algebraic curve, for example $F_1(x,y)=0$, then we
assume that $n_2=0$ and vice versa.

Here and in what follows, $O$  denotes the zero element of the field
$\mathbb{C}_{\infty}\{x\}$. Let us suppose that $y^{(2)}_{\infty}(x)$ is a
Puiseux series of type ($II$) such that the following condition
$\displaystyle F_2\left(x,y^{(2)}_{\infty}(x)\right)=O$ is valid.
Substituting $y=y^{(2)}_{\infty}(x)$ into the function $ \lambda(x,y)/y$, we can
expand it into a Puiseux series near infinity. Supposing that $n_2>0$, we
find the dominant behavior near the point $x=\infty$ of this series. The
result is
\begin{equation}
 \label{Lienard_exp_inv1_3}
\frac{ \lambda\left(x,y^{(2)}_{\infty}(x)\right)}{y^{(2)}_{\infty}(x)}=\frac{n_2f_0^2}{g_0}x^{2m-n}+o\left(x^{2m-n}\right), \quad x\rightarrow\infty.
\end{equation}
The inequality $\deg f<\deg g<2\deg f+1$ yields $2m-n>-1$. Using
Theorem~\ref{T:Exp_fact}, we come to a contradiction.  Consequently, we
should set $n_2=0$. The polynomial $r(x,y)$ now takes the form
$r(x,y)=F_1^{n_1}(x,y)$ with $F_1(x,y)=y-q(x)-z_0$ and $z_0\in\mathbb{C}$.

Equation $\mathcal{X}E=\varrho(x,y) E$ produces the following linear
inhomogeneous  partial differential equation
\begin{equation} \label{Lienard_exp_inv1_3a}
yh_x-(g(x)y+f(x))h_y=n_1\lambda_{1}(x,y)h+\varrho(x,y) F^{n_1}_{1}(x,y), \, \varrho(x,y)\in\mathbb{C}[x,y]
\end{equation}
satisfied by the polynomial $h(x,y)$. Let us consider the truncated Puiseux
series $y(x)=q(x)+z_0$ that solves equation \eqref{Lienard_y_x} whenever
invariant algebraic curve $F_{1}(x,y)=0$ exists. Recall that $y(x)$ is a zero
of the polynomial $F_{1}(x,y)$. Considering the restriction
$H(x)=h(x,y)|_{y=y(x)}$, we obtain the ordinary differential equation
\begin{equation} \label{Lienard_exp_inv1_4}
\left(q(x)+z_0\right)\frac{d H}{d x}=n_1\lambda_{1}(x,y) H,\quad \lambda_1(x,y)=-f(x)-q_x(x).
\end{equation}
Since the polynomials $h(x,y)$ and $F_{1}(x,y)$ are relatively prime, we
conclude that $H(x)\not\equiv 0$. Indeed, assuming the converse and using the
B\'{e}zout's theorem, we see that the bivariate polynomials $h(x,y)$ and
$r(x,y)$ have a common factor.

It follows from the relations $\deg q=m+1$ and $\deg \lambda_1\leq m-1$  that
equation \eqref{Lienard_exp_inv1_4} does not have non-zero polynomial
solutions. This fact contradicts the existence of exponential invariants
$E(x,y)=\exp\left\{h(x,y)/r(x,y)\right\}$ with a non-constant polynomial
$r(x,y)$.

\end{proof}

Now let us establish that Li\'{e}nard differential systems
\eqref{Lienard_gen} do not have Darboux first integrals whenever $\deg f<\deg
g<2\deg f+1$.

\begin{theorem}\label{T:Lienard_Integrability1}
Li\'{e}nard differential systems  \eqref{Lienard_gen}  from family ($A$) are
not Darboux integrable.
\end{theorem}
\begin{proof}
Using Theorem \ref{T:L_two_curves}, we see that a Li\'{e}nard differential
system \eqref{Lienard_gen} with the restriction $\deg f<\deg g<2\deg f+1$ has
at most two distinct irreducible invariant algebraic curves simultaneously.
As in the proof of Lemma \ref{L:Lienard_exp_inv2}, we denote these  curves as
$F_1(x,y)=0$ and $F_2(x,y)=0$. Without loss of generality, we can  suppose
that the generating polynomials of these algebraic curves are given by
expression \eqref{Lienard1_F} with $N=1$, $k=0$ and $N\in\mathbb{N}$, $k=1$,
respectively. If there exists a first integral being a Darboux function, then
this first integral should be of the form
\begin{equation}
 \label{Lienard_Inegrability_Gen1}
I(x,y)=F_1^{d_1}(x,y)F_2^{d_2}(x,y),\quad d_1,d_2\in\mathbb{C},\quad |d_1|+|d_2|>0.
\end{equation}
Indeed, Lemma \ref{L:Lienard_exp_inv2} forbids the existence of exponential
factors given by invariants $E(x,y)=\exp\left\{h(x,y)/r(x,y)\right\}$, where
$h(x,y)\in\mathbb{C}[x,y]$ and $r(x,y)\in\mathbb{C}[x,y]\setminus\mathbb{C}$.
By Lemma \ref{L:Lienard_exp_inv1} exponential invariants
$E(x,y)=\exp[h(x,y)]$ that could arise in an expression of the first integral
have divisible by $y$ cofactors, while invariant algebraic curves
$F_1(x,y)=0$ and $F_2(x,y)=0$ have cofactors  that are independent of~$y$.
Recall that there are no rational first integrals by the corollary to Theorem
\ref{T:L_two_curves}. Consequently, first integrals that are Darboux
functions do not have exponential factors.

Let us note that if the Li\'{e}nard differential system in question has only
one irreducible invariant algebraic curve, for example $F_1(x,y)=0$, then we
suppose that $d_2=0$ and vice versa.

The cofactors $\lambda_1(x,y)$ and $\lambda_2(x,y)$ of invariant algebraic
curves $F_1(x,y)=0$ and $F_2(x,y)=0$ are given in relation
\eqref{Lienard_exp_inv1_cof1}. First integral
\eqref{Lienard_Inegrability_Gen1} exists provided that the following
condition $d_1\lambda_1(x,y)+d_2\lambda_2(x,y)=0$ is satisfied. This
condition can be rewritten as
\begin{equation}
 \label{Lienard_Inegrability_Gen2}
d_1[f(x)+q_x(x)]+d_2[Nf(x)+(N-1)q_x(x)+p_x(x)]=0.
\end{equation}
The highest-degree term in expression \eqref{Lienard_Inegrability_Gen1} is
$d_2f_0x^m $. Since $f_0\neq0$, we get $d_2=0$. Consequently, the first
integral can be chosen as a rational function with $d_1=1$. Again we recall
the corollary to Theorem \ref{T:L_two_curves}, which excludes the existence
of rational first integrals.

\end{proof}

Further, our aim is to study the existence of non-autonomous Darboux first
integrals with a time-dependent exponential factor.

\begin{lemma}\label{L:Lienard_Integrability_time1}
A Li\'{e}nard differential system  \eqref{Lienard_gen} from family ($A$) has
a non-autonomous Darboux first integral  with a time-dependent
 exponential factor \eqref{FI_t_gen} if and only if  $\deg g=\deg f +1$ and the following condition
 \begin{equation}
 \label{Lienard_Inegrability_1g}
g(x)=\omega\left(\int_{0}^xf(s)ds-\omega x-z_0\right),\quad \omega,z_0\in\mathbb{C},\quad \omega\neq0
\end{equation}
is valid. A first integral  takes the form
 \begin{equation}
 \label{Lienard_Inegrability_1}
I(x,y,t)=\left(y+\int_{0}^xf(s)ds-\omega x-z_0\right)\exp(\omega t).
\end{equation}
There are no other independent non-autonomous Darboux first integrals with a
time-dependent exponential factor.
\end{lemma}

\begin{proof}
If the polynomial $g(x)$ is given by relation
\eqref{Lienard_Inegrability_1g}, then it is straightforward to verify that
function  \eqref{Lienard_Inegrability_1} is a non-autonomous first integral
of the related differential system.

Let us establish the necessity of condition \eqref{Lienard_Inegrability_1g}.
We repeat the reasoning given in the proof of Theorem
\ref{T:Lienard_Integrability1}. By Theorems~\ref{T:L23_Non_aut_FI} and
\ref{T:L_two_curves}  a non-autonomous first integral \eqref{FI_t_gen} reads
as
\begin{equation}
 \label{Lienard_Inegrability_t_Gen1}
 \begin{gathered}
I(x,y,t)=F_1^{d_1}(x,y)F_2^{d_2}(x,y)\exp(\omega t),\, d_1,d_2,\omega\in\mathbb{C},\, |d_1|+|d_2|>0,\, \omega\neq0.
\end{gathered}
\end{equation}
Now we need to study the following condition
\begin{equation}
 \label{Lienard_Inegrability_t_Gen2}
d_1[f(x)+q_x(x)]+d_2[Nf(x)+(N-1)q_x(x)+p_x(x)]-\omega =0.
\end{equation}
Similarly  to the case of Theorem \ref{T:Lienard_Integrability1}, we obtain
$d_2=0$. Further, it is without loss of generality to set $d_1=1$.
Considering relation  \eqref{Lienard_Inegrability_t_Gen2} as an ordinary
differential equation for the polynomial $q(x)$ and performing the
integration, we find its solution
\begin{equation}
 \label{Lienard_Inegrability_t_Gen2_2}
q(x)=-\int_{0}^xf(s)ds+\omega x
\end{equation}
satisfying the condition $q(0)=0$ presented in \eqref{Lienard1_F_series_q}.
Finally, we note that the invariant algebraic curve $F_1(x,y)=0$ given by the
polynomial $F_1(x,y)=y-q(x)-z_0$ exists whenever the series of type ($I$) defined
in \eqref{Lienard1_F_series} terminates at the zero term. In this case
equation \eqref{Lienard_y_x} has a polynomial solution $y(x)=q(x)+z_0$.
Substituting our results into equation \eqref{Lienard_y_x}, we find expression
\eqref{Lienard_Inegrability_1g}. Consequently, we obtain the following
equality $\deg g= \deg f +1 $. The absence  of other independent
non-autonomous Darboux first integrals \eqref{FI_t_gen} follows from the
previous considerations and the uniqueness of the invariant algebraic curve
$F_1(x,y)=0$.

\end{proof}

It follows  from Lemma \ref{L:Lienard_Integrability_time1} that  ordinary
differential equation \eqref{Lienard_y_x} related to a Li\'{e}nard
differential system from family ($A$) possessing a non-autonomous Darboux
first integral
 has a polynomial solution
$y(x)=q(x)+z_0$.

We recall that integrating factors and Jacobi last multipliers with a
constant ratio belong to the same equivalence class. We do not distinguish between
them.

\begin{theorem}\label{T:Lienard_Integrability2}
A Li\'{e}nard differential system  \eqref{Lienard_gen} from family ($A$) is
Liouvillian integrable if and only if the following assertions are
 valid:
\begin{enumerate}

\item the system under consideration possesses two distinct irreducible
    invariant algebraic curves $F_1(x,y)=0$ and $F_2(x,y)=0$, where the
    polynomial $F_1(x,y)=y-q(x)-z_0$ is given by
    expression~\eqref{Lienard1_F} with $N=1$, $k=0$ and the polynomial
    $F_2(x,y)\in\mathbb{C}[x,y]$ takes the form \eqref{Lienard1_F} with
    $N\in\mathbb{N}$, $k=1$;

\item the polynomials $q(x)$ and $p(x)$ giving initial parts of the
    Puiseux series near the point $x=\infty$ that solve equation
    \eqref{Lienard_y_x} identically satisfy
     the condition
\begin{equation}
 \label{Lienard_Inegrability_C1_cond}
(n-m)[f(x)+q_x(x)]+(m+1)p_x(x)=0.
\end{equation}

\end{enumerate}
The related  Li\'{e}nard differential system has the unique Darboux
integrating factor
   \begin{equation}
 \label{Lienard_Inegrability_Int_fact}
M(x,y)=\frac{(y-q(x)-z_0)^{\frac{N(m+1)-(n+1)}{m+1}}}{F_2(x,y)}.
\end{equation}

\end{theorem}

\begin{proof}
Using Theorem \ref{T:Liouville}, we see that a Liouvillian integrable
 differential system  should have a Darboux integrating factor.
Arguing as in the proof of Theorem~\ref{T:Lienard_Integrability1}, we
conclude that such an integrating factor does not involve exponential
invariants and is of the form
\begin{equation}
 \label{Lienard_Inegrability_IF1}
M(x,y)=F_1^{d_1}(x,y)F_2^{d_2}(x,y),\quad d_1,d_2\in\mathbb{C},\quad |d_1|+|d_2|>0.
\end{equation}
In this expression the polynomials $F_1(x,y)$ and $F_2(x,y)$ define invariant
algebraic curves. The polynomial $F_1(x,y)$ is given by expression
\eqref{Lienard1_F} with $N=1$ and $k=0$. The polynomial $F_2(x,y)$ is
produced by the same expression with $N\in\mathbb{N}$ and $k=1$. Again we
note that if the Li\'{e}nard differential system in question has only one
irreducible invariant algebraic curve, for example $F_1(x,y)=0$, then we set
$d_2=0$ and vice versa. Calculating the divergence of vector field
\eqref{Lienard_Inegrability_VF} yields $\text{div}\,\mathcal{X}=-f(x)$. Hence
the necessary and sufficient condition for Darboux integrating
factor~\eqref{Lienard_Inegrability_IF1} to exist takes the form
$d_1\lambda_1(x,y)+d_2\lambda_2(x,y)=f(x)$, where $\lambda_1(x,y)$ and
$\lambda_2(x,y)$ are the cofactors of invariant algebraic curves $F_1(x,y)=0$
and $F_2(x,y)=0$, respectively. These cofactors are given in expression
\eqref{Lienard_exp_inv1_cof1}. The condition enabling the existence of the
Darboux integrating factor explicitly reads as
\begin{equation}
 \label{Lienard_Inegrability_IF2}
d_1[f(x)+q_x(x)]+d_2[Nf(x)+(N-1)q_x(x)+p_x(x)]=-f(x).
\end{equation}
Balancing the highest-degree terms in this expression, we get $d_2=-1$.
Further, we recall that the invariant algebraic curve $F_1(x,y)=0$ is given
by the polynomial $F_1(x,y)=y-q(x)-z_0$. If $d_1=0$, then the curve
$F_1(x,y)=0$ either does not exist or does not enter  the explicit expression
\eqref{Lienard_Inegrability_IF1} of the integrating factor.

Now let us consider the representation of the cofactors $\lambda_1(x,y)$ and
$\lambda_2(x,y)$ in the ring $\mathbb{C}_{\infty}\{x\}[y]$. They are of the
form
\begin{equation}
 \label{Lienard_Inegrability_C1_n1}
 \begin{gathered}
\lambda_1(x,y)=-f(x)-q_x(x),\,
\lambda_2(x,y)=-\sum_{j=1}^{N-1}\left[f(x)+\left(y^{(1)}_{j,\infty}\right)_x\right]-\left[f(x)+\left(y^{(2)}_{N,\infty}\right)_x\right].
\end{gathered}
\end{equation}
Substituting  these expressions into the condition
$d_1\lambda_1(x,y)+d_2\lambda_2(x,y)=f(x)$ with $d_2=-1$, we find
\begin{equation}
 \label{Lienard_Inegrability_C1_n2}
 \begin{gathered}
d_1[f(x)+q_x(x)]-\sum_{j=1}^{N-1}\left[f(x)+\left(y^{(1)}_{j,\infty}\right)_x\right]=\left(y^{(2)}_{N,\infty}\right)_x.
\end{gathered}
\end{equation}
The Puiseux series $y(x)=y^{(1)}_{j,\infty}(x)$ and the polynomial
$y(x)=q(x)+z_0$ solve equation \eqref{Lienard_y_x}. Thus, we obtain
\begin{equation}
 \label{Lienard_Inegrability_C1_n3}
 \begin{gathered}
f(x)+q_x(x)=-\frac{g(x)}{q(x)+z_0},\quad f(x)+\left(y^{(1)}_{j,\infty}\right)_x=-\frac{g(x)}{y^{(1)}_{j,\infty}}.
\end{gathered}
\end{equation}
Again we suppose that if there is no polynomial solution $y(x)=q(x)+z_0$,
then $d_1=0$. Using expressions \eqref{Lienard_Inegrability_C1_n3}, we
rewrite condition \eqref{Lienard_Inegrability_C1_n2} in the form
\begin{equation}
 \label{Lienard_Inegrability_C1_n4}
 \begin{gathered}
g(x)\left[\sum_{j=1}^{N-1}\frac{1}{y^{(1)}_{j,\infty}}-\frac{d_1}{q(x)+z_0}\right]=\left(y^{(2)}_{N,\infty}\right)_x.
\end{gathered}
\end{equation}
Let us find the highest-order terms in a neighborhood of the point
$x=\infty$. Using the asymptotic formulae
\begin{equation}
 \label{Lienard_Inegrability_C1_n5}
 \begin{gathered}
g(x)=g_0x^n+o(x^n),\quad q(x)=-\frac{f_0}{m+1}x^{m+1}+o(x^{m+1}),\quad x\rightarrow\infty\hfill\\
y^{(1)}_{j,\infty}(x)=-\frac{f_0}{m+1}x^{m+1}+o(x^{m+1}),\quad y^{(2)}_{N,\infty}(x)=-\frac{g_0}{f_0}x^{n-m}+o(x^{n-m}),\quad x\rightarrow\infty,
\end{gathered}
\end{equation}
we collect the coefficients of $x^{n-m-1}$ in expression
\eqref{Lienard_Inegrability_C1_n4}. This yields
\begin{equation}
 \label{Lienard_Inegrability_C1_n6}
d_1=N-1-\frac{n-m}{m+1}.
\end{equation}
Using inequalities $m<n<2m+1$, we see that the parameter $d_1$ cannot be
zero. Substituting relation  \eqref{Lienard_Inegrability_C1_n6} and $d_2=-1$
into expressions  \eqref{Lienard_Inegrability_IF1} and
\eqref{Lienard_Inegrability_IF2}, we find condition
\eqref{Lienard_Inegrability_C1_cond} and the Darboux integrating factor as
given in  \eqref{Lienard_Inegrability_IF1}.

\end{proof}

\textit{Remark.} Condition  \eqref{Lienard_Inegrability_C1_cond} is
identically satisfied whenever $\deg g=\deg f +1$. Indeed, if $n=m+1$, then
we get
\begin{equation}
 \label{Lienard_Inegrability_C1_n7}
f(x)+q_x(x)=\frac{(m+1)g_0}{f_0},\quad p(x)=-\frac{g_0}{f_0}x.
\end{equation}
Consequently, Li\'{e}nard
differential systems satisfying the restriction $\deg g=\deg f +1$ and possessing two distinct irreducible invariant algebraic curves are always Liouvillian integrable.

We see from Theorem \ref{T:Lienard_Integrability2} that equation
\eqref{Lienard_y_x} related to a Liouvillian integrable Li\'{e}nard
differential system from family ($A$) has a polynomial
 solution given by the expression $y(x)=q(x)+z_0$. In addition, using
condition \eqref{Lienard_Inegrability_C1_cond} and equation
\eqref{Lienard_y_x}, we can represent the polynomials $f(x)$ and $g(x)$ in
the integrable cases as
\begin{equation}
 \label{Lienard_Inegrability_C1_n8}
f(x)=-q_x(x)-\frac{m+1}{n-m}p_x(x),\quad g(x)=\frac{m+1}{n-m}[q(x)+z_0]p_x(x).
\end{equation}
Recall that the polynomial $q(x)$ is of degree $m+1$ and the polynomial
$p(x)$ is of degree $n-m$. It follows from Theorem
\ref{T:Lienard_Integrability2} that Li\'{e}nard differential systems
\eqref{Lienard_gen} from family ($A$) are  Liouvillian integrable  if and
only if  the polynomials $f(x)$ and $g(x)$ are given by
expressions~\eqref{Lienard_Inegrability_C1_n8} and the systems possess an
invariant algebraic curve with the generating polynomial taking the form
\eqref{Lienard1_F}, where
    $N\in\mathbb{N}$ and $k=1$.

Now let us study the integrability of  Li\'{e}nard differential systems such
that the related equation \eqref{Lienard_y_x} has two distinct polynomial
solutions. The following theorem is valid.

\begin{theorem}\label{T:Lienard_IntegrabilityA_partial}
A Li\'{e}nard differential system  \eqref{Lienard_gen} from family ($A$) with
two distinct invariant algebraic curves given by first-degree polynomials
with respect to $y$ is Liouvillian integrable if and only if the system is of
the form
\begin{equation}
 \label{Lienard_Inegrability_Apartial_1}
x_t=y,\quad y_t=\left[k\beta v^{k-1}(x)+(k+l)v^{l-1}(x)\right]v_xy-k\left[\beta v^k(x)+v^l(x)\right]v^{l-1}(x)v_x,
\end{equation}
where $\beta\in\mathbb{C}\setminus\{0\}$, $v(x)$ is a polynomial of degree
$(n-m)/l$, $k$ and $l$ are relatively prime natural numbers such that the
 restriction  $(m+1)l=(n-m)k$
 holds. The related  Li\'{e}nard differential system has the unique
Darboux integrating factor
   \begin{equation}
 \label{Lienard_Inegrability_Int_Fact_A_p1}
M(x,y)=\{y-\beta v^k(x)-v^l(x)\}^{-\frac{l}{k}}\left\{y-v^l(x)\right\}^{-1}
\end{equation}
and the invariant algebraic curves $y-\beta v^k(x)-v^l(x)=0$, $y-v^l(x)=0$. A
Liouvillian first integral is of the form
 \begin{equation}
 \begin{gathered}
 \label{Lienard_Inegrability_FI_A_p1}
I(x,y)=\frac{k\beta^{\frac{l}{k}}}{k-l}\{y-\beta v^k(x)-v^l(x)\}^{\frac{k-l}{k}}+
\sum_{j=0}^{m}\exp\left[-\frac{\pi l(2j+1)i}{k}\right]\\
\times\ln\left\{\{y-\beta v^k(x)-v^l(x)\}^{\frac{1}{m+1}}
-\exp\left[\frac{\pi(2j+1)i}{m+1}\right][\beta v^k(x)]^{\frac{1}{m+1}}\right\}.
\end{gathered}
\end{equation}
\end{theorem}

\begin{proof}
We begin the proof by noting that we use novel designations for the
polynomial solutions of equation \eqref{Lienard_y_x}. We set
$\tilde{q}(x)=q(x)+z_0$ and $\tilde{p}(x)=p(x)+z_1$.
Condition~\eqref{Lienard_Inegrability_C1_cond} and equation
\eqref{Lienard_y_x} provide explicit expressions
\eqref{Lienard_Inegrability_C1_n8} for the polynomials $f(x)$ and $g(x)$.
Substituting the relation $y(x)=\tilde{p}(x)$ and expressions
\eqref{Lienard_Inegrability_C1_n8} into equation~\eqref{Lienard_y_x}, we
integrate the resulting equation with respect to the polynomial
$\tilde{q}(x)$. This way, we get
\begin{equation}
 \label{Lienard_Inegrability_C1_n9}
\tilde{q}(x)=\beta \tilde{p}^{\frac{m+1}{n-m}}(x)+\tilde{p}(x),
\end{equation}
where $\beta\in\mathbb{C}$ is a constant of integration. Recalling the fact
that $\tilde{q}(x)$ is a polynomial of degree $m+1$ and $\tilde{p}(x)$ is a
polynomial of degree $n-m$, we find that $\beta\neq0$. Introducing relatively
prime natural numbers $k$ and $l$ according to the rule $(m+1)l=(n-m)k$, we
represent the polynomials $\tilde{p}(x)$ and $\tilde{q}(x)$ as
\begin{equation}
 \label{Lienard_Inegrability_C1_p_q}
\tilde{p}(x)=v^l(x),\quad \tilde{q}(x)=\beta v^k(x) + v^l(x).
\end{equation}
In this expression $v(x)$ is an arbitrary polynomial of degree $(n-m)/l$.
  Substituting $N=1$ into expression~\eqref{Lienard_Inegrability_Int_fact}, we find the integrating factor as
given in relation \eqref{Lienard_Inegrability_Int_Fact_A_p1}. Finally, we
calculate  the line integral
\begin{equation}
 \label{Lienard_InegrabilityA_first_integrals}
 \begin{gathered}
I(x,y)=\int_{(x_0,y_0)}^{(x,y)}M(x,y)\left[ydy+\left\{f(x)y+g(x)\right\}dx\right]
\end{gathered}
\end{equation}
and obtain first integral  \eqref{Lienard_Inegrability_FI_A_p1}, where $i$ is
the imaginary unit.
\end{proof}
\textit{Remark.}  The family of systems \eqref{Lienard_Inegrability_Apartial_1} can be transformed to the following simple form
\begin{equation}
 \label{Lienard_Inegrability_Apartial_1_Sundman}
s_{\tau}=z,\quad z_{\tau}=\left[k\beta s^{k-1}+(k+l)s^{l-1}\right]z
-k\left[\beta s^k+s^l\right]s^{l-1}
\end{equation}
via the generalized Sundman transformation $s(\tau)=v(x)$, $z(\tau)=y$, $d\tau=v_x(x)dt$. Substituting $v(x)=s$, $y=z$ into \eqref{Lienard_Inegrability_FI_A_p1}, we find a Liouvillian first integral for systems \eqref{Lienard_Inegrability_Apartial_1_Sundman}.

 Let us demonstrate that for any fixed
 degrees of the polynomials $f(x)$ and $g(x)$ there exist Li\'{e}nard
differential systems \eqref{Lienard_gen} from family ($A$) that have
Liouvillian first integrals. With this aim we  suppose that the polynomial
$\tilde{p}(x)$ takes the following form $\tilde{p}(x)=x^{n-m}$. We find the
polynomial $\tilde{q}(x)$ using expression
\eqref{Lienard_Inegrability_C1_n9}. The result is $\tilde{q}(x)=\beta
x^{m+1}+x^{n-m}$. Thus, we conclude that the following family of Li\'{e}nard
differential systems
\begin{equation}
 \label{Lienard_Inegrability_C1_n10}
 \begin{gathered}
x_t=y,\quad y_t=\left[(m+1)\beta x^m+(n+1)x^{n-m-1}\right]y
-(m+1)\left(\beta x^n+x^{2n-2m-1}\right)
\end{gathered}
\end{equation}
 is Liouvillian integrable. The related Darboux integrating factor reads as
\begin{equation}
 \label{Lienard_Inegrability_C1_n10_DIF}
 \begin{gathered}
M(x,y)=\left(y-x^{n-m}\right)^{-1}\left(y-\beta x^{m+1}-x^{n-m}\right)^{\frac{m-n}{m+1}}.
\end{gathered}
\end{equation}
In addition, if the numbers $n$
and $m$ are the following $n=l(k+1)-1$ and $m=lk-1$, where $l$,
$k\in\mathbb{N}$, $k>1$, then relations \eqref{Lienard_Inegrability_C1_n8}
and \eqref{Lienard_Inegrability_C1_n9} produce  Liouvillian integrable
systems
\begin{equation}
 \label{Lienard_Inegrability_C1_n11}
 \begin{gathered}
x_t=y,\quad y_t=\left[k\beta \tilde{p}^{k-1}(x)+k+1\right]\tilde{p}_x(x)y
-k\left(\beta\tilde{p}^{k-1}(x)+1\right)\tilde{p}(x)\tilde{p}_x(x),
\end{gathered}
\end{equation}
where $\tilde{p}(x)$ is a polynomial of degree $n-m=l$. Using expression
\eqref{Lienard_Inegrability_Int_Fact_A_p1}, we see that
systems~\eqref{Lienard_Inegrability_C1_n11} possess the  Darboux integrating
factor
\begin{equation}
 \label{Lienard_Inegrability_C1_n12}
 \begin{gathered}
M(x,y)=\{y-\beta\tilde{p}^k(x)-\tilde{p}(x)\}^{-\frac1{k}}\{y-\tilde{p}(x)\}^{-1}.
\end{gathered}
\end{equation}
A related Liouvillian first integral is fairly simple in the case $k=2$. Let
us explicitly present another expression of a first integral:
\begin{equation}
 \label{Lienard_Inegrability_C1_n13}
 \begin{gathered}
I(x,y)=\text{arctanh}\,\left\{\sqrt{\frac{\beta \tilde{p}^{\,2}(x)}{F(x,y)}}\right\}+\sqrt{\beta F(x,y)},\quad
F(x,y)=\beta\tilde{p}^{\,2}(x)+\tilde{p}(x)-y.
\end{gathered}
\end{equation}
Finally, let us note that there exist Liouvillian integrable Li\'{e}nard
differential systems from family ($A$) such that the polynomial $F_2(x,y)$ in
expression \eqref{Lienard_Inegrability_Int_fact} is of degree with respect to
$y$ greater than $1$. In fact, the degree with respect to $y$ of the
polynomial $F_2(x,y)$ can be an arbitrary natural number. This fact was
established in article \cite{Demina13}.

Our next step is to investigate the existence of non-autonomous
Darboux--Jacobi last
 multipliers. The cases $\deg g=\deg f +1$ and $\deg g\neq \deg f +1$ will be considered
 separately.

\begin{lemma}\label{L:Lienard_Integrability_t2}
A Li\'{e}nard differential system \eqref{Lienard_gen} satisfying
 the condition $\deg f+1<\deg g<2\deg f+1$ has a non-autonomous Darboux--Jacobi last
 multiplier of the form \eqref{JLM_gen} if and only if the following assertions are
 valid:
\begin{enumerate}

\item the system under consideration possesses two distinct irreducible
    invariant algebraic curves $F_1(x,y)=0$ and $F_2(x,y)=0$, where the
    polynomial $F_1(x,y)=y-q(x)-z_0$ is given by
    expression~\eqref{Lienard1_F} with $N=1$, $k=0$ and the polynomial
    $F_2(x,y)\in\mathbb{C}[x,y]$ reads as \eqref{Lienard1_F} with
    $N\in\mathbb{N}$, $k=1$;

\item the polynomials $q(x)$ and $p(x)$ giving initial parts of the
    Puiseux series near the point $x=\infty$ that solve equation
    \eqref{Lienard_y_x} identically satisfy the condition
\begin{equation}
 \label{Lienard_Inegrability_C1_cond_time}
(n-m)[f(x)+q_x(x)]+(m+1)[p_x(x)+\omega]=0,
\end{equation}
where $\omega\in\mathbb{C}\setminus\{0\}$ is a constant.

\end{enumerate}
 The non-autonomous Darboux--Jacobi last
 multiplier is unique and takes the form
   \begin{equation}
 \label{Lienard_Inegrability_Int_fact_time}
M(x,y,t)=\frac{(y-q(x)-z_0)^{\frac{N(m+1)-(n+1)}{m+1}}}{F_2(x,y)}\exp(\omega t).
\end{equation}
\end{lemma}

\begin{proof}
We use Theorem \ref{T:L23_Non_aut_JLM} and repeat the proof of the previous
theorem. The only difference is in condition
\eqref{Lienard_Inegrability_IF2}. Now this condition reads as
\begin{equation}
 \label{Lienard_Inegrability_t_IF2}
d_1[f(x)+q_x(x)]+d_2[Nf(x)+(N-1)q_x(x)+p_x(x)]-\omega=-f(x),
\end{equation}
where $\omega\neq0$. The related non-autonomous Darboux--Jacobi last
 multiplier is given by the expression
\begin{equation}
 \label{Lienard_Inegrability_new1_time_add2}
M(x,y,t)=F_1^{d_1}(x,y)F_2^{d_2}(x,y)\exp[\omega t].
\end{equation}
Similarly to the case of Theorem \ref{T:Lienard_Integrability2}, we find the
values of $d_1$ and $d_2$. They are $d_1=N-1-(n-m)/(m+1)$ and $d_2=-1$. Both
of them are non-zero. Consequently, a Li\'{e}nard differential system
\eqref{Lienard_gen} with a non-autonomous Darboux--Jacobi last
 multiplier has two
distinct irreducible invariant algebraic curves $F_1(x,y)=0$ and
$F_2(x,y)=0$. Substituting explicit values of the parameters $d_1$ and $d_2$
into condition \eqref{Lienard_Inegrability_t_IF2} yields relation
\eqref{Lienard_Inegrability_C1_cond_time}.

If there exist two distinct non-autonomous Darboux-Jacobi last
 multipliers~\eqref{JLM_gen}, then their ratio is a Darboux first integral
 either autonomous or non-autonomous of the form \eqref{FI_t_gen}. By
 Theorem \ref{T:Lienard_Integrability1} and Lemma
 \ref{L:Lienard_Integrability_time1} such a situation is impossible.

\end{proof}

\begin{lemma}\label{L:Lienard_Integrability_t2_add}

  A Li\'{e}nard differential system \eqref{Lienard_gen} satisfying
 the condition $\,$ $\deg g = \deg f+1$ has a non-autonomous Darboux--Jacobi last
 multiplier of the form \eqref{JLM_gen} if and only if  the system under study possesses the irreducible
    invariant algebraic curve $F_2(x,y)=0$ given by relation \eqref{Lienard1_F} with
    $N\in\mathbb{N}$ and $k=1$.
A non-autonomous Darboux--Jacobi last
 multiplier is of the from
   \begin{equation}
 \label{Lienard_Inegrability_Int_fact_time_add}
M(x,y,t)=\frac{\exp\left[\frac{g_0}{f_0}\left\{1-(N-1)(m+1) \right\}t\right]}{F_2(x,y)}.
\end{equation}
There are no other  non-autonomous Darboux--Jacobi last
 multipliers whenever the invariant algebraic curve $F_2(x,y)=0$ is unique.
 If, in addition, there exists the invariant algebraic curve $F_1(x,y)=0$
 presented in
    expression~\eqref{Lienard1_F} with $N=1$ and $k=0$, then the system under
    consideration has a family of non-autonomous Darboux--Jacobi last
 multipliers
 \begin{equation}
 \label{Lienard_Inegrability_Int_fact_time_addn}
M(x,y,t)=\frac{F_1^{d_1}(x,y)\exp\left[\frac{g_0}{f_0}\left\{1+\{d_1-(N-1)\}(m+1) \right\}t\right]}{F_2(x,y)},\, d_1\in\mathbb{C}.
\end{equation}

\end{lemma}

\begin{proof}
Similarly to the case of Lemma \ref{L:Lienard_Integrability_t2}, we see that
a non-autonomous Darboux--Jacobi last
 multiplier of the form \eqref{JLM_gen} exists if and only if the system
 under consideration has at least one invariant algebraic curve and
 condition~\eqref{Lienard_Inegrability_t_IF2} is valid. Again we suppose that $d_j=0$ whenever the invariant algebraic curve $F_j(x,y)=0$
 does not exist. The explicit expression of a non-autonomous Darboux--Jacobi last
 multiplier is given by  \eqref{Lienard_Inegrability_new1_time_add2}. Balancing the highest-degree terms in condition \eqref{Lienard_Inegrability_t_IF2} yields $d_2=-1$. This fact proves the
 existence of the invariant algebraic curve $F_2(x,y)=0$.   Further, we find that the polynomials
 $f(x)+q_x(x)$ and $p_x(x)$ are constants, see relation
 \eqref{Lienard_Inegrability_C1_n7}. Substituting equalities $d_2=-1$, $f(x)+q_x(x)=(m+1)g_0/f_0$ and
 $p_x(x)=-g_0/f_0$ into condition \eqref{Lienard_Inegrability_t_IF2}, we
 obtain
 \begin{equation}
 \label{Lienard_Inegrability_Int_fact_time_add_pq}
(m+1)g_0d_1+\left\{1-(N-1)(m+1)\right\}g_0-f_0\omega=0.
\end{equation}
Setting $d_1=0$, we find the value of $\omega$ as given in relation
\eqref{Lienard_Inegrability_Int_fact_time_add}. If the Li\'{e}nard
differential system in question does not have the invariant algebraic curve
$F_1(x,y)=0$, then non-autonomous Darboux--Jacobi last
 multiplier \eqref{Lienard_Inegrability_Int_fact_time_add} is unique. In the
 converse case, non-autonomous Darboux--Jacobi last
 multiplier  \eqref{Lienard_Inegrability_Int_fact_time_add} also exists.
 In addition, we recall that the system has time-dependent Darboux first
 integral \eqref{Lienard_Inegrability_1}. It is straightforward to show that
 the product of a first integral and a non-constant Jacobi last
 multiplier is another Jacobi last
 multiplier.
\end{proof}

It is demonstrated in  \cite{Demina17}
 that Li\'{e}nard differential
systems   with non-autonomous Darboux--Jacobi last
 multipliers  \eqref{Lienard_Inegrability_Int_fact_time_add} indeed exist. In addition, let us
note that the famous Duffing--van der  Pol oscillators  belong to family
($A$). The classification of Liouvillian integrable  Duffing--van der Pol
oscillators is performed in \cite{Demina07}.

\section{Integrability of  Li\'{e}nard differential systems from family ($B$)}\label{S:Lienard_B}

Let us consider families  of Li\'{e}nard differential systems with fixed
degrees of the polynomials $f(x)$ and $g(x)$ such that the following relation
$\deg g=2\deg f +1$ is satisfied. If we do not impose  restrictions on the
highest-degree coefficients $f_0$ and $g_0$ of the polynomials $f(x)$ and
$g(x)$, then it has been established in \cite{Demina18} that the
Fuchs indices of the Puiseux series near the point $x=\infty$ that solve related
equations \eqref{Lienard_y_x} depend on the parameters $f_0$ and $g_0$.
Consequently, performing the classification of irreducible invariant algebraic
curves and Darboux or Liouvillian first integrals is a very difficult problem
whenever only the degrees of the polynomials $f(x)$ and $g(x)$  are fixed.
The method of Puiseux series can deal with each case of a positive rational
Fuchs index separately. Interestingly, such a degeneracy leads to a variety
of distinct integrable cases arising in Li\'{e}nard differential systems from
family ($B$).

This section is mainly devoted to the non-resonant case. We say that a
Li\'{e}nard differential system from family ($B$) is resonant near infinity
if equation~\eqref{eq:DP2_5_2} possesses a solution in $\mathbb{Q}^+$. For
convenience, we introduce the parameter $\delta$ according to the rule
\begin{equation}\label{eq:DP2_5_sig}
g_0=\frac{f^2_0-\delta^2}{4(m+1)},\quad \delta\in\mathbb{C},\quad \delta\neq\pm f_0.
\end{equation}
Using this normalization, we solve equation  \eqref{eq:DP2_5_2}. As a result
we find the Fuchs indices. They take the form
\begin{equation}\label{Lienard_degenerate_Int1}
p_1=\frac{2(m+1)\delta}{\delta-f_0},\quad p_2=\frac{2(m+1)\delta}{\delta+f_0}.
\end{equation}
There are no positive rational Fuchs indices if and only if the condition
$\delta/f_0\not\in\mathbb{Q}\setminus\{0\}$ is valid. In addition, we assume
that the following inequality $ \deg f>0$ holds. The integrability problem
for Li\'{e}nard  differential systems~\eqref{Lienard_gen} under restrictions
$\deg f=0$ and $\deg g=1$ is simple, see the end of Section \ref{S:Lienard}.

In this section we use the designations of Theorem
\ref{T:Lienard_degenerate}. In particular, the Puiseux series near the point
$x=\infty$ that solve equation \eqref{Lienard_y_x} are denoted as
$y^{(1)}_{\infty}(x)$ and $y^{(2)}_{\infty}(x)$. Introducing the variable
$\delta$ instead of the parameter $g_0$, we see that these series have the
following dominant behavior
\begin{equation}\label{Lienard_degenerate_Puiseux_series_dominant}
\begin{gathered}
y^{(1)}_{\infty}(x)=\frac{\delta-f_0}{2(m+1)}x^{m+1}+o(x^{m+1}),\quad x\rightarrow\infty;\\
y^{(2)}_{\infty}(x)=-\frac{\delta+f_0}{2(m+1)}x^{m+1}+o(x^{m+1}),\quad x\rightarrow\infty.
\end{gathered}
\end{equation}
Let us denote the polynomial parts of the series $y^{(1)}_{\infty}(x)$ and
$y^{(2)}_{\infty}(x)$ as $q_1(x)$ and $q_2(x)$, respectively. Thus, we have
the equalities
\begin{equation}\label{Lienard_degenerate_Puiseux_series_Polynomial_parts}
\begin{gathered}
q_1(x)=\left\{y^{(1)}_{\infty}(x)\right\}_+,\quad q_2(x)=\left\{y^{(2)}_{\infty}(x)\right\}_+.
\end{gathered}
\end{equation}
If $\delta=0$, then the series $y^{(1)}_{\infty}(x)$ and
$y^{(2)}_{\infty}(x)$ coincide. Let us omit the  indices and set
$q(x)=\left\{y_{\infty}(x)\right\}_+$.

We  begin by investigating the existence of exponential invariants related to
invariant algebraic curves.

\begin{lemma}\label{L:Lienard_exp_degenerate}
Let $h(x,y)\in\mathbb{C}[x,y]$ and
$r(x,y)\in\mathbb{C}[x,y]\setminus\mathbb{C}$ be  relatively prime
polynomials. A Li\'{e}nard  differential  system \eqref{Lienard_gen}
satisfying
 the conditions $\deg g=2\deg f+1$ and $\delta/f_0\not\in\mathbb{Q}\setminus\{0\}$ has exponential invariants
$E(x,y)=\exp\left\{h(x,y)/r(x,y)\right\}$ if and only if the following
statements are valid:

\begin{enumerate}

\item $\delta=0$;

\item there exists the invariant algebraic curve $F(x,y)=0 $ with the
    generating polynomial $F(x,y)=y-q(x)$;

\item the ordinary differential equation
\begin{equation}\label{Lienard_degenerate_Int2}
q(x)u_x(x)+[f(x)+q_x(x)]u(x)=0
\end{equation}
has a non-zero polynomial solution $u(x)$.

\end{enumerate}
The related exponential invariant is of the form
\begin{equation}\label{Lienard_degenerate_exp_inv_explicit}
E(x,y)=\exp\left[\frac{u(x)}{y-q(x)}\right]
\end{equation}
and possesses the cofactor $\varrho(x,y)=u_x(x)$.

\end{lemma}

\begin{proof}

It follows from Theorem \ref{T:Lienard_degenerate} that any Li\'{e}nard
differential  system \eqref{Lienard_gen} satisfying
 the conditions $\deg g=2\deg f+1$ and
 $\delta/f_0\not\in\mathbb{Q}\setminus\{0\}$ has at most two distinct irreducible invariant algebraic curves simultaneously.
  The degrees with respect to $y$ of polynomials producing
 irreducible invariant algebraic curves are either $1$ or $2$.
 If there exists an irreducible invariant algebraic curves of degree $2$ with respect to $y$,
 then it is unique. Further, there can arise at most two distinct irreducible invariant algebraic curves of degree $1$ with respect to $y$.
Since $E(x,y)=\exp\left\{h(x,y)/r(x,y)\right\}$ is an exponential invariant,
we conclude that $r(x,y)=0$ is an invariant algebraic curve of the related
system. Let us represent exponential invariants in the from
\begin{equation}\label{Lienard_degenerate_exp_inv_explicit_gen}
E(x,y)=\exp\left[\frac{h(x,y)}{F^{n_1}_1(x,y)F^{n_2}_2(x,y)}\right],\quad n_1,n_2\in\mathbb{N}_0,\quad n_1+n_2>0,
\end{equation}
where $F_1(x,y)=0$ and $F_2(x,y)=0$ are irreducible invariant algebraic
curves. Without loss of generality we set $n_k=0$ whenever the invariant algebraic curve $F_k(x,y)=0$ does not exist. Here $k=1$ or
 $k=2$.

 \textit{Case 1.} Let us suppose that there exists only one irreducible
 invariant algebraic curve $F_1(x,y)=0$ of degree $2$ with respect to $y$.
 Using relation \eqref{F_L2_5_1}, we find the cofactor $\lambda_1(x,y)$. The
 result is
\begin{equation}\label{Lienard_degenerate_exp_inv_cof1}
\lambda(x,y)=-2f(x)-\left\{\left(y^{(1)}_{\infty}\right)_x+\left(y^{(2)}_{\infty}\right)_x\right\}_+
\end{equation}
The dominant behavior of the cofactor near the point $x=\infty$ is
$\lambda(x,y)=-f_0x^{m}+o(x^{m})$. Now let us take one of the Puiseux series,
for example,  the series $y^{(1)}_{\infty}(x)$. We find the asymptotic
relation
\begin{equation}\label{Lienard_degenerate_exp_inv_cof1_add}
\frac{\lambda\left(x,y^{(1)}_{\infty}(x)\right)}{y^{(1)}_{\infty}(x)}=
-\frac{2(m+1)f_0}{(\delta-f_0)x}+o\left(\frac1{x}\right),\quad x\rightarrow\infty
\end{equation}
Since $\delta/f_0\not\in\mathbb{Q}\setminus\{0\}$, we conclude that the
conditions of Theorem \ref{T:Exp_fact} are not satisfied whenever
$\delta\neq0$. The case $\delta=0$ will be considered separately. Exponential
invariants  do
not exist provided that $\delta\neq 0$.

\textit{Case 2.} Let us suppose that the differential system under
consideration has two distinct  irreducible
 invariant algebraic curves $F_1(x,y)=0$ and $F_2(x,y)=0$. The polynomials producing the curves are of degree $1$ with respect to $y$
 and have the cofactors
\begin{equation}\label{Lienard_degenerate_exp_inv_cof2}
\lambda_m(x,y)=-f(x)-\left\{\left(y^{(m)}_{\infty}\right)_x\right\}_+,\quad m=1,2.
\end{equation}
The cofactor of the invariant algebraic curve
$F^{n_1}_1(x,y)F^{n_2}_2(x,y)=0$ is given by the polynomial
$\lambda(x,y)=n_1\lambda_1(x,y)+n_2\lambda_2(x,y)$ and has the following
dominant behavior near the point $x=\infty$:
\begin{equation}\label{Lienard_degenerate_exp_inv_cof2_add}
\lambda(x,y)=\frac12\left[n_2(\delta-f_0)-n_1(\delta+f_0)\right]x^{m}+o(x^{m}),\quad x\rightarrow \infty.
\end{equation}
Suppose that one of the numbers $n_1$ and $n_2$ is non-zero, for example,
$n_1$. Further, we consider the following asymptotic relation
\begin{equation}\label{Lienard_degenerate_exp_inv_cof2_addn}
\frac{\lambda\left(x,y^{(1)}_{\infty}(x)\right)}{y^{(1)}_{\infty}(x)}=(m+1)\left(n_2-\frac{(\delta+f_0)n_1}{\delta-f_0}\right)
\frac{1}{x}+o\left(\frac1{x}\right),\quad x\rightarrow\infty.
\end{equation}
Using Theorem \ref{T:Exp_fact} and expression
$\delta/f_0\not\in\mathbb{Q}\setminus\{0\}$, we see that there are no
exponential invariants \eqref{Lienard_degenerate_exp_inv_explicit_gen}
whenever $\delta\neq0$.

\textit{Case 3.} Let us suppose that the differential system in question has
only one  irreducible
 invariant algebraic curve: either $F_1(x,y)=0$ or $F_2(x,y)=0$. If
 $\delta\neq0$, then arguing as in the previous case, we prove the non-existence
 of exponential invariants.

  \textit{Case 4.} Now let us suppose that $\delta=0$. Recall that the Puiseux series $y^{(1)}_{\infty}(x)$ and
$y^{(2)}_{\infty}(x)$ merge in this case. The unique Puiseux series centered
at the point $x=\infty$ that satisfies equation \eqref{Lienard_y_x} is
denoted as $y_{\infty}(x)$. In what follows we use the local theory of
 invariants considered in Section \ref{S:Local}. Let a Li\'{e}nard
differential  system \eqref{Lienard_gen} satisfying
 the conditions $\deg g=2\deg f+1$ and
 $\delta=0$ have the irreducible invariant algebraic curve $F(x,y)=0$, where
 $F(x,y)=y-q(x)$. If in addition the system possesses an exponential invariant
 $E(x,y)=\exp\{h(x,y)/F^{k}(x,y)\}$ for some $k\in\mathbb{N}$, then it is without loss of generality to suppose that the degree of the
 polynomial $h(x,y)$ with respect to $y$ is at most $k-1$. There
 exists a finite
number of local elementary  exponential invariants
\begin{equation} \label{Lienard_degenerate_exp_inv1_loc1}
\begin{gathered}
E_j(x,y)=\exp\left[\frac{u_j(x)}{\{y-y_{\infty}(x)\}^{k_j}}\right],\quad
 u_j(x)\in \mathbb{C}_{\infty}\{x\},\\
  k_j\in\mathbb{N},\quad j=1,\ldots, K,\quad K\in\mathbb{N}
\end{gathered}
\end{equation}
  such that the exponential invariant
$E(x,y)$ equals the product $\displaystyle \prod_{j=1}^KE_j^{}(x,y)$. We
assume that the following inequalities  $1\leq k_1<k_2<\ldots<k_K\leq k$ are
valid. Let us denote the cofactor of the local elementary exponential
invariant $E_j(x,y)$ by $\varrho_j(x,y)\in\mathbb{C}_{\infty}\{x\}[y]$. We
see that the following expression $\displaystyle \sum_{j=1}^K\varrho_j(x,y)$
equals the cofactor $\varrho(x,y)$ of the invariant $E(x,y)$. Note that the
cofactor $\varrho(x,y)$ is an element of the ring~$\mathbb{C}[x,y]$.

Substituting the explicit representation of $E_j(x,y)$ into the partial
differential equation $\mathcal{X}E_j(x,y)=\varrho_j(x,y)E_j(x,y)$, we find
\begin{equation} \label{Lienard_degenerate_exp_inv1_2n}
\begin{gathered}
yu_{j,x}=k_j\lambda_j(x,y)u_j(x)+\varrho_j(x,y)\left\{y-y_{\infty}(x)\right\}^{k_j}.
\end{gathered}
\end{equation}
In this expression $\lambda_j(x,y)\in\mathbb{C}_{\infty}\{x\}[y]$  is the
cofactor of the local elementary invariant $f(x,y)=y-y_{\infty}(x)$.  Using
Theorem \ref{T:coff_local2}, we obtain the cofactor $\lambda_j(x,y)$. The
result is
\begin{equation} \label{Lienard_degenerate_exp_inv1_cof_local}
\begin{gathered}
\lambda_j(x,y)=-f(x)-\left\{y_{\infty}(x)\right\}_x.
\end{gathered}
\end{equation}
 Analyzing expression
\eqref{Lienard_degenerate_exp_inv1_2n}, we see that $j=1$, $k_1=k=1$,
$\varrho_1(x,y)=u_{1,x}(x)$, and the series $u_1(x)$ satisfies the following
ordinary differential equation
\begin{equation} \label{Lienard_degenerate_exp_inv1_3n}
\begin{gathered}
y_{\infty}(x)u_{1,x}(x)+\left(f(x)+\left\{y_{\infty}(x)\right\}_x\right)u_1(x)=0.
\end{gathered}
\end{equation}
Since there exists the invariant algebraic curve $y-q(x)=0$, we conclude that
the Puiseux series $y_{\infty}(x)$ is in fact the polynomial $q(x)$:
$y_{\infty}(x)=q(x)$. Consequently, the exponential invariant
$E(x,y)=\exp\{h(x,y)/F(x,y)\}$ exists if and only if $h(x,y)=u_1(x)$ and
$u_1(x)$ is a polynomial. Omitting the index, we get expressions
\eqref{Lienard_degenerate_Int2} and
\eqref{Lienard_degenerate_exp_inv_explicit}.

\end{proof}

\textit{Remark.} Balancing the highest-degree terms in equation
\eqref{Lienard_degenerate_Int2}, it can be shown that the polynomial $u(x)$
is of degree $m+1$.

Our next step is to derive the necessary and sufficient conditions of Darboux
integrability for non-resonant  Li\'{e}nard  differential  systems from
family ($B$). The case $\delta=0$ will be considered separately.

\begin{theorem}\label{T:Lienard_degenerate_Darboux1}
A Li\'{e}nard  differential  system \eqref{Lienard_gen} satisfying
 the conditions $\deg g=2\deg f+1$ and $\delta/f_0\not\in\mathbb{Q}$
is Darboux integrable if and only if   the system is of the form
\begin{equation} \label{Lienard_degenerate_Darboux_explicit}
x_t=y,\quad y_t=\frac{2f_0}{f_0-\delta}q_{1,\,x}y-\frac{f_0+\delta}{f_0-\delta}q_{1,\,x}q_1,
\end{equation}
where $q_1(x)$ is a polynomial of degree $m+1$ with the highest-degree
coefficient $(\delta-f_0)/(2\{m+1\})$. A related Darboux
    first integral reads as
\begin{equation}\label{Lienard_degenerate_Darboux_FI1}
I(x,y)=\left[y-q_1(x)\right]^{\delta-f_0}\left[y-\frac{(f_0+\delta)}{(f_0-\delta)}q_1(x)\right]^{\delta+f_0}.
\end{equation}
\end{theorem}

\begin{proof}

It is straightforward to verify that expression
\eqref{Lienard_degenerate_Darboux_FI1} gives a Darboux first integral
whenever all other conditions of the theorem are valid.

Let us prove the converse statement. It follows from Lemmas
\ref{L:Lienard_exp_inv1} and \ref{L:Lienard_exp_degenerate} that Darboux
first integrals of Li\'{e}nard differential  systems \eqref{Lienard_gen}
satisfying
 the conditions $\deg g=2\deg f+1$ and $\delta/f_0\not\in\mathbb{Q}$ do not have
exponential factors. By Theorem \ref{T:Lienard_degenerate} the systems under
consideration have at most two distinct irreducible invariant algebraic
curves simultaneously.

First, we consider a Li\'{e}nard  differential  system \eqref{Lienard_gen} that satisfies
 the conditions $\deg g=2\deg f+1$, $\delta/f_0\not\in\mathbb{Q}$ and possesses only one irreducible invariant
algebraic curve. If there exists a Darboux first integral, then it can be
chosen as a bivariate polynomial $I(x,y)=F(x,y)$ producing the invariant
algebraic curve $F(x,y)=0$. Using Theorem \ref{T:Lienard_degenerate}, we see
that the generating polynomial $F(x,y)$ is of degree at most $2$ with respect
to $y$. The Darboux first integral $I(x,y)=F(x,y)$ exists if an only if the
cofactor $\lambda(x,y)$ of the invariant algebraic curve $F(x,y)=0$ is
identically zero. We need to consider three possibilities. Let us write down
the related cofactors and their dominant behavior near the point $x=\infty$.
With the help of relations \eqref{Lienard_degenerate_Puiseux_series_dominant}
and \eqref{Lienard_degenerate_Puiseux_series_Polynomial_parts}, we get
\begin{equation} \label{Lienard_degenerate_Darboux_FI2}
\begin{gathered}
N=1:\, \lambda(x,y)=-f(x)-q_{1,x},\, \lambda(x,y)=-\frac12(\delta+f_0)x^m+o(x^m),\quad x\rightarrow\infty;\hfill\\
N=1:\, \lambda(x,y)=-f(x)-q_{2,x},\, \lambda(x,y)=\frac12(\delta-f_0)x^m+o(x^m),\quad x\rightarrow\infty;\hfill\\
N=2:\, \lambda(x,y)=-2f(x)-q_{1,x}-q_{2,x},\, \lambda(x,y)=-f_0x^m+o(x^m),\quad x\rightarrow\infty.\hfill\\
\end{gathered}
\end{equation}
In these expressions $N$ denotes the degree of the polynomial $F(x,y)$ with
respect to $y$. We conclude that the cofactors are not identically zero.
Thus, there are no Darboux first integrals.

Now, we consider a Li\'{e}nard  differential  system \eqref{Lienard_gen} that satisfies
 the conditions $\deg g=2\deg f+1$, $\delta/f_0\not\in\mathbb{Q}$ and possesses  two distinct irreducible invariant
algebraic curves $F_1(x,y)=0$ and $F_2(x,y)=0$. By Theorem
\ref{T:Lienard_degenerate} the polynomials $F_1(x,y)$ and $F_2(x,y)$ are of
the form $F_1(x,y)=y-q_1(x)$ and $F_2(x,y)=y-q_2(x)$. If there exists a  Darboux
first integral, then it can be represented in the form
\begin{equation} \label{Lienard_degenerate_Darboux_FI3}
\begin{gathered}
I(x,y)=F_1^{d_1}(x,y)F_2^{d_2}(x,y),\quad d_1,d_2\in\mathbb{C},\quad |d_1|+|d_2|>0
\end{gathered}
\end{equation}
This first integral exists if and only if the following condition
$d_1\lambda_1(x,y)+d_2\lambda_2(x,y)=0$ is satisfied. Finding the dominant
behavior near the point $x=\infty$ of the cofactors
\begin{equation} \label{Lienard_degenerate_Darboux_FI3_dom}
\begin{gathered}
\lambda_1(x,y)=-\frac12(\delta+f_0)x^m+o(x^m),\quad \lambda_2(x,y)=\frac12(\delta-f_0)x^m+o(x^m),
\end{gathered}
\end{equation}
we obtain $d_1=\delta-f_0$, $d_2=\delta+f_0$ and the expression
\begin{equation}\label{Lienard_degenerate_Darboux_condition1}
2\delta f(x)+(\delta-f_0)q_{1,x}(x)+(\delta+f_0)q_{2,x}(x)=0.
\end{equation}
Recall that one of the parameters $d_1$ or $d_2$ can be chosen arbitrary. By
Theorem \ref{T:Lienard_degenerate_Darboux1}  equation~\eqref{Lienard_y_x}
related to the Li\'{e}nard  differential  system under consideration
possesses two distinct polynomial solutions $y(x)=q_1(x)$ and $y(x)=q_2(x)$.
Thus, we obtain the following relations
\begin{equation} \label{Lienard_degenerate_Darboux_explicit_add2}
f(x)+q_{1,x}(x)=-\frac{g(x)}{q_{1}(x)},\quad f(x)+q_{2,x}(x)=-\frac{g(x)}{q_{2}(x)}.
\end{equation}
Substituting these relations and the values of $d_1$ and $d_2$ into the
condition on the cofactors $d_1[f(x)+q_{1,x}(x)]+d_2[f(x)+q_{2,x}(x)]=0$
yields the expression $q_2(x)=-d_2q_1(x)/d_1$. Finally, we use relations
\eqref{Lienard_degenerate_Darboux_condition1} and
\eqref{Lienard_degenerate_Darboux_explicit_add2} in order to derive the
explicit representations of the  polynomials $f(x)$ and $g(x)$.
\end{proof}

\textit{Corollary.} The family of first-order ordinary differential equations
\begin{equation} \label{Lienard_degenerate_Darboux_explicit_add}
yy_x-\frac{2f_0}{f_0-\delta}q_{1,\,x}y+\frac{f_0+\delta}{f_0-\delta}q_{1,\,x}q_1=0
\end{equation}
associated with systems \eqref{Lienard_degenerate_Darboux_explicit} has two
distinct polynomial solutions of the form $y(x)=q_1(x)$ and
$y(x)=(f_0+\delta)q_1(x)/(f_0-\delta)$.

\textit{Remark 1.} Suppose we are in assumptions of Theorem
\ref{T:Lienard_degenerate_Darboux1} with the exception of the condition
$\delta/f_0\not\in\mathbb{Q}\setminus\{0\}$. Then function~\eqref{Lienard_degenerate_Darboux_FI1} is still a Darboux first integral of
the related  Li\'{e}nard  differential  system. In addition, the corollary to
Theorem \ref{T:Lienard_degenerate_Darboux1} is also valid. However, there may
exist  other resonant Li\'{e}nard  differential systems from family ($B$)
with Darboux first integrals.

\textit{Remark 2.} Expression \eqref{Lienard_degenerate_Darboux_explicit}
 provides a set of systems with
rational first integrals~\eqref{Lienard_degenerate_Darboux_FI1} provided that
$f_0/\delta$ is a rational number.

Next, let us study the Darboux integrability in the case $\delta=0$.

\begin{theorem}\label{T:Lienard_degenerate_Darboux2}
A Li\'{e}nard  differential  system \eqref{Lienard_gen} satisfying
 the conditions $\deg g=2\deg f+1$ and $\delta=0$
is Darboux integrable if and only if the system can be represented in the
form
\begin{equation}\label{Lienard_degenerate_Darboux_del0_f_g}
x_t=y,\quad y_t=2q_x(x)y-q(x)q_x(x),
\end{equation}
where $q(x)$ is a polynomial of degree $m+1$. A related Darboux
    first integral reads~as
\begin{equation}\label{Lienard_degenerate_Darboux_FI1_del0}
I(x,y)=\left[y-q(x)\right]\exp\left[-\frac{q(x)}{y-q(x)}\right].
\end{equation}
\end{theorem}

\begin{proof}
By direct computations we verify that expression
\eqref{Lienard_degenerate_Darboux_FI1_del0} gives a Darboux first integral of
a Li\'{e}nard  differential  system \eqref{Lienard_gen} with the polynomials
$f(x)$ and $g(x)$ satisfying
relations~\eqref{Lienard_degenerate_Darboux_del0_f_g} and $\delta=0$. Thus,
we have established sufficiency of conditions presented in the theorem.

Let us prove their necessity. Suppose that a Li\'{e}nard  differential system
from family ($B$) with $\delta=0$ possesses a Darboux first integral. A
Darboux integrable differential system~\eqref{DS} has at least one invariant
algebraic curve. It follows from Theorem \ref{T:Lienard_degenerate}  that the
Li\'{e}nard differential system in question possesses at most one irreducible
invariant algebraic curve. This curve is of the form $y-q(x)=0$ and has the
cofactor $\lambda(x,y)=-f(x)-q_x(x)$. Thus, we have established the existence
of the invariant algebraic curve $y-q(x)=0$. According to
Lemma~\ref{L:Lienard_exp_degenerate}, the Li\'{e}nard differential system under consideration may have exponential invariants associated with the invariant
algebraic curve $y-q(x)=0$. These invariants take the form
\eqref{Lienard_degenerate_exp_inv_explicit} and possess the cofactor
$\varrho(x,y)=u_x(x)$, where the polynomial $u(x) $ satisfies equation
\eqref{Lienard_degenerate_Int2}. Using Lemma \ref{L:Lienard_exp_inv1}, we
conclude that exponential invariants with a polynomial argument cannot enter
explicit expressions of Darboux first integrals. Consequently, a Darboux
first integral can be represented in the form
\begin{equation}\label{Lienard_degenerate_Darboux_FI1_del0_test}
I(x,y)=\left[y-q(x)\right]^{d}\exp\left[\frac{u(x)}{y-q(x)}\right],\quad d\in\mathbb{C},
\end{equation}
where we suppose that $u(x)\equiv0$ whenever exponential invariants
\eqref{Lienard_degenerate_exp_inv_explicit} do not exist.

If $d=0$, then the related Li\'{e}nard differential system has an invariant
algebraic curve $u(x)=0$ independent of $y$. It is impossible due to Theorem
\ref{T:Lienard_degenerate}. Thus, it is without loss of generality to set
$d=1$. The cofactors of all the invariants identically satisfy the relation
$\lambda(x,y)+\varrho(x,y)=0$ provided that first integral
\eqref{Lienard_degenerate_Darboux_FI1_del0_test} exists. As a result, we get
the expression $f(x)=u_x(x)-q_x(x)$. Substituting this expression into
equation \eqref{Lienard_degenerate_Int2} yields $u(x)=-q(x)$. This relation
proves existence of exponential invariants
\eqref{Lienard_degenerate_exp_inv_explicit}. Since $y=q(x)$ is a polynomial
solution of equation \eqref{Lienard_y_x} and $f(x)=-2q_x(x)$, we find the
polynomial $g(x)$. The result is $g(x)=q(x)q_x(x)$.

\end{proof}

Interestingly, Li\'{e}nard  differential systems
\eqref{Lienard_degenerate_Darboux_explicit} and
\eqref{Lienard_degenerate_Darboux_del0_f_g} are those characterized by the
so-called Chiellini integrability condition
\begin{equation} \label{Lienard_Chiellini}
\frac{d}{dx}\left[\frac{g(x)}{f(x)}\right]=\alpha f(x),\quad \alpha\in\mathbb{C}\setminus\{0\}.
\end{equation}
This condition was originally introduced by A. Chiellini \cite{Chiellini01}.
  Chiellini integrable Li\'{e}nard  differential systems can be transformed to linear systems $s_{\tau}=z$, $z_{\tau}=-z-\alpha s$
via the generalized Sundman transformation $s(\tau)=\int f(x) dx$, $z(\tau)=y$, $d\tau=f(x)dt$, see \cite{Berkovich01}. Some other properties of Chiellini integrable Li\'{e}nard  differential systems
are presented in~\cite{Choudhury_Lienard}.

Our next step is to study the existence of non-autonomous Darboux first
integrals with a time-dependent exponential factor~\eqref{FI_t_gen}.

\begin{lemma}\label{L:Lienard_degenerate_Darboux_time1}
A Li\'{e}nard  differential  system \eqref{Lienard_gen} satisfying
 the conditions $\deg g=2\deg f+1$, $\deg f\neq0$, and $\delta/f_0\not\in\mathbb{Q}$
possesses a non-autonomous Darboux first integral \eqref{FI_t_gen} if and
only if the system reads as
\begin{equation}\label{Lienard_degenerate_Darboux_system_time}
\begin{gathered}
x_t=y,\quad y_t=-\left\{f_0(x-x_0)^m+\frac{(m+2)\omega}{2(m+1)\delta}\right\}y-
\frac{f_0^2-\delta^2}{4(m+1)}(x-x_0)^{2m+1}\\
-\frac{f_0\omega}{2(m+1)\delta}(x-x_0)^{m+1}-\frac{\omega^2}{4(m+1)\delta^2}(x-x_0)^{},
\end{gathered}
\end{equation}
where $x_0\in\mathbb{C}$ and $\omega\in\mathbb{C}\setminus\{0\}$. A related
non-autonomous Darboux
    first integral takes the form
\begin{equation}\label{Lienard_degenerate_Darboux_FI1_time}
\begin{gathered}
I(x,y,t)=\left[y+\frac{f_0-\delta}{2(m+1)}(x-x_0)^{m+1}+\frac{\omega}{2(m+1)\delta}(x-x_0)\right]^{\delta-f_0}\\
\times\left[y+\frac{\delta+f_0}{2(m+1)}(x-x_0)^{m+1}+\frac{\omega}{2(m+1)\delta}(x-x_0)\right]^{\delta+f_0}\exp(\omega t).
\end{gathered}
\end{equation}

\end{lemma}

\begin{proof}
It is straightforward to derive that expression
\eqref{Lienard_degenerate_Darboux_FI1_time} is a non-autonomous Darboux first
integral of systems \eqref{Lienard_degenerate_Darboux_system_time}. We only
need to prove the converse statement. Supposing that a non-resonant
Li\'{e}nard  differential  system from family ($B$) possesses a
non-autonomous Darboux first integral \eqref{FI_t_gen}, we use the arguments
given in the proof of Theorem \ref{T:Lienard_degenerate_Darboux1} to represent
this first integral as
\begin{equation}\label{Lienard_degenerate_Darboux_FI1_time_a}
I(x,y,t)=\left[y-q_1(x)\right]^{\delta-f_0}\left[y-q_2(x)\right]^{\delta+f_0}\exp(\omega t) ,\quad \omega\neq0,
\end{equation}
where $q_1(x)$ and $q_2(x)$ are distinct polynomial  solutions of equation
\eqref{Lienard_y_x}. In addition, we get the following condition
\begin{equation}\label{Lienard_degenerate_Darboux_condition1_time}
2\delta f(x)+(\delta-f_0)q_{1,\,x}(x)+(\delta+f_0)q_{2,\,x}(x)-\omega=0.
\end{equation}
Further, we express the polynomial $f(x)$ from this condition. Substituting
$y(x)=q_1(x)$ into equation \eqref{Lienard_y_x}, we obtain the polynomial
$g(x)$. Let us introduce the polynomial $v(x)$ according to the rule
$q_2(x)-q_1(x)=v(x)$. Requiring that the function $q_2(x)=q_1(x)+v(x)$ is a
solution of equation \eqref{Lienard_y_x}, we get the relation
\begin{equation}\label{Lienard_degenerate_Darboux_condition1_time_eq_v}
[(\delta-f_0)v(x)+2\delta q_1(x)]v_x(x)+\omega v(x)=0
\end{equation}
Since $q_1(x)$ is a polynomial, we obtain the ordinary differential equation
$\beta(x-x_0)v_x=v$, where $\beta$, $x_0\in\mathbb{C}$. Integrating this
equation yields $v(x)=v_0(x-x_0)^{1/\beta}$ with $v_0\in\mathbb{C}$ being a
constant of integration. It follows from expression $q_2(x)-q_1(x)=v(x)$ that
$v(x)$ is a polynomial of degree $m+1$. Thus, we obtain $\beta=1/(m+1)$. As a
result the polynomials
 $q_1(x)$ and $q_2(x)$ can be represented in the form
 \begin{equation}\label{Lienard_degenerate_Darboux_condition1_time_eq_q_12}
 \begin{gathered}
q_1(x)=\frac{\delta-f_0}{2(m+1)}(x-x_0)^{m+1}-\frac{\omega}{2(m+1)\delta}(x-x_0),\\
q_2(x)=-\frac{\delta+f_0}{2(m+1)}(x-x_0)^{m+1}-\frac{\omega}{2(m+1)\delta}(x-x_0).
\end{gathered}
\end{equation}
Finally,   we find the polynomials $f(x)$ and $g(x)$ from
condition \eqref{Lienard_degenerate_Darboux_condition1_time} and equation
\eqref{Lienard_y_x} recalling the fact that, for example, $y(x)=q_1(x)$ is a
solution of the latter. The uniqueness of independent non-autonomous Darboux
    first integral \eqref{Lienard_degenerate_Darboux_FI1_time} follows from the uniqueness of the polynomials $q_1(x)$, $q_2(x)$
and the dominant behavior of the cofactors given by expression~\eqref{Lienard_degenerate_Darboux_FI3_dom}.

\end{proof}

\textit{Remark.} Suppose we are in assumptions of Lemma
\ref{L:Lienard_degenerate_Darboux_time1} with the exception of the condition
$\delta/f_0\not\in\mathbb{Q}\setminus\{0\}$. Then function~\eqref{Lienard_degenerate_Darboux_FI1_time} is still a non-autonomous Darboux first integral of
the related  Li\'{e}nard  differential  system.  However, there may
exist  other resonant Li\'{e}nard  differential systems from family ($B$)
with non-autonomous Darboux first integrals of the form~\eqref{FI_t_gen}.

Let us note that if $\delta=\pm mf_0/(m+2)$, then any system \eqref{Lienard_degenerate_Darboux_system_time} has not only a non-autonomous Darboux first integral \eqref{Lienard_degenerate_Darboux_FI1_time}, but also an independent Darboux first integral~\eqref{Lienard_degenerate_Darboux_FI1}, where $q_1(x)$ is given by the relation
 \begin{equation}\label{Lienard_degenerate_Darboux_time_GS_q1}
 \begin{gathered}
q(x)=-\frac{f_{0} x^{m +1}}{\left(m +1\right) \left(m +2\right)}-\frac{\left(m +2\right) \omega  x}{2 \left(m +1\right) f_{0} m}.
\end{gathered}
\end{equation}
Let $I_1(x,y)$ be Darboux first integral~\eqref{Lienard_degenerate_Darboux_FI1} and $I_2(x,y,t)$ be non-autonomous Darboux first integral~\eqref{Lienard_degenerate_Darboux_FI1_time}. Eliminating $y$ from the relations $I_1(x,y)=C_1$ and $I_2(x,y,t)=C_2$, we can find the general solution of a system \eqref{Lienard_degenerate_Darboux_system_time} under the condition $\delta=\pm mf_0/(m+2)$. Note that such a system is resonant near $x=\infty$. The general solution in the case $m=2$ previously appeared in \cite{Ruiz01}, see also \cite{DS2021}.

\begin{lemma}\label{L:Lienard_degenerate_Darboux_time2}
 A Li\'{e}nard  differential  system \eqref{Lienard_gen} satisfying
 the conditions $\deg g=2\deg f+1$, $\deg f\neq0$, and $\delta=0$
possesses a non-autonomous Darboux first integral~\eqref{FI_t_gen} if and
only if the system is of the form
\begin{equation}\label{Lienard_degenerate_Darboux_del0_f_g_time}
\begin{gathered}
x_t=y,\quad y_t=-\left\{f_0(x-x_0)^m+\frac{(m+2)\omega}{m+1}\right\}y
-\frac{(x-x_0)}{4(m+1)}\left\{f_0(x-x_0)^{m}+2\omega\right\}^2,
\end{gathered}
\end{equation}
where $x_0\in\mathbb{C}$ and $\omega\in\mathbb{C}\setminus\{0\}$. A related
non-autonomous Darboux
    first integral reads as
\begin{equation}\label{Lienard_degenerate_Darboux_FI1_del0_time}
\begin{gathered}
I(x,y,t)=\exp\left[\frac{f_0(x-x_0)^{m+1}}{2(m+1)y+f_0(x-x_0)^{m+1}+2\omega(x-x_0)}\right]\\
\times\left[y+\frac{f_0(x-x_0)^{m+1}}{2(m+1)}+\frac{\omega(x-x_0)}{m+1}\right]\exp(\omega t).
\end{gathered}
\end{equation}

\end{lemma}

\begin{proof}
By direct computations we verify that expression
\eqref{Lienard_degenerate_Darboux_FI1_del0_time} is a time-dependent first
integral of system \eqref{Lienard_degenerate_Darboux_del0_f_g_time}.

Let us prove the converse statement. We suppose that   a Li\'{e}nard
differential system from family ($B$) satisfies the restriction $\delta=0$
and possesses a non-autonomous Darboux first integral~\eqref{FI_t_gen}.
Repeating the arguments used in the proof of Theorem
\ref{T:Lienard_degenerate_Darboux2}, we represent such a first integral in
the form
\begin{equation}\label{Lienard_degenerate_Darboux_FI1_del0_time_new}
I(x,y,t)=\left[y-q(x)\right]^d\exp\left[\frac{u(x)}{y-q(x)}\right]\exp(\omega t),\quad
d\in\mathbb{C},\quad\omega\in\mathbb{C}\setminus\{0\}.
\end{equation}
In addition, we note that  the related system possesses the invariant
algebraic curve $y-q(x)=0$ with the cofactor $\lambda(x,y)=-f(x)-q_x(x)$. By
Lemma \ref{L:Lienard_exp_degenerate} the cofactor $\varrho(x,y)$ of the
exponential invariant $E(x,y)=\exp[u(x)/(y-q(x))]$ reads as
$\varrho(x,y)=u_x(x)$. Condition~\eqref{JLM_gen_cond} relating these
cofactors and the parameter $\omega$ takes the form
\begin{equation}\label{Lienard_degenerate_cond_cof_time1}
d[f(x)+q_x(x)]-u_x(x)-\omega=0,\quad d\in\mathbb{C}.
\end{equation}
Let us begin with the case $d=0$.  We conclude from condition
\eqref{Lienard_degenerate_cond_cof_time1} that $u(x)$ is a first-degree
polynomial. Using the remark to Lemma \ref{L:Lienard_exp_degenerate}, we find
the value of $m$. The result is $m=0$. As it was mentioned at the beginning
of this section, we do not consider Li\'{e}nard differential systems with the
restriction $m=0$ ($\deg f=0$).

We turn to the case $d\neq 0$. Without loss of generality, we  set $d=1$.
Now let us suppose that the restriction $u(x)\equiv0$ is valid. The related system may not have exponential invariants. Finding
the dominant behavior of the cofactor $\lambda(x,y)=-f(x)-q_x(x)$ near the
point $x=\infty$, we obtain
\begin{equation}\label{Lienard_degenerate_cof_dom_time1}
\lambda(x,y)=-\frac{f_0}{2}x^m+o(x^m),\quad x\rightarrow \infty.
\end{equation}
Recalling the inequality $m>0$, we see that condition
\eqref{Lienard_degenerate_cond_cof_time1} is not satisfied. Thus, we conclude
that the exponential invariant exists. Further, we eliminate from relations
\eqref{Lienard_degenerate_Int2} and \eqref{Lienard_degenerate_cond_cof_time1}
the polynomial $f(x)$. As a result we get the following expression
\begin{equation}\label{Lienard_degenerate_DI_q}
q(x)=-u(x)-\omega\frac{u(x)}{u_x(x)}.
\end{equation}
Consequently, the ratio $u(x)/u_x(x)$ is a polynomial. It is straightforward
to see that this polynomial is of the first degree and can be represented as
$\beta(x-x_0)$, where $\beta$, $x_0\in\mathbb{C}$. Integrating the ordinary
differential equation $\{\beta(x-x_0)\}u_x(x)=u(x)$, we obtain
$u(x)=u_0(x-x_0)^{1/\beta}$, where $u_0\in\mathbb{C}$ is a constant of
integration. We recall that that $u(x)$ is a polynomial of degree $m+1$. As a
result we get $\beta=1/(m+1)$. Substituting the relation
$u(x)=u_0(x-x_0)^{m+1}$ into expressions \eqref{Lienard_degenerate_DI_q} and
\eqref{Lienard_degenerate_cond_cof_time1}, we find the polynomials $q(x)$ and
$f(x)$. In addition, we choose the parametrization $u_0=f_0/(2\{m+1\})$. The
polynomial $g(x)$ we find recalling the fact that $y(x)=q(x)$ is the
polynomial solution of the related equation \eqref{Lienard_y_x}.

Thus, we see that if the Li\'{e}nard  differential  system in question has a
non-autonomous Darboux first integral~\eqref{FI_t_gen}, then there exist the invariant
algebraic curve $y-q(x)=0$ and  the exponential invariant $E(x,y)=\exp[\alpha
u(x)/(y-q(x))]$, where $\alpha\in\mathbb{C}$, the polynomial $q(x)$ is given
by expression \eqref{Lienard_degenerate_DI_q}, and the polynomial $u(x)$ is
$u(x)=f_0(x-x_0)^{m+1}/(2\{m+1\})$.


\end{proof}

Below we shall prove that Li\'{e}nard  differential  systems \eqref{Lienard_degenerate_Darboux_system_time} and \eqref{Lienard_degenerate_Darboux_del0_f_g_time}
are Liouvillian integrable.

Let us study the Liouvillian integrability of  Li\'{e}nard differential
systems from family ($B$). We begin with some partial case characterized by
Darboux integrating factors of a special form. Note that  in
Theorem \ref{T:Lienard_degenerate_Liouville_polynomial} we use novel
designations for polynomial solutions of equations \eqref{Lienard_y_x} related to Li\'{e}nard differential
systems. We need novel
designations because polynomials $p_1(x)$ and $p_2(x)$ may have coinsiding
dominant terms. After considering this special case, we shall turn to
non-resonant systems.

\begin{theorem}\label{T:Lienard_degenerate_Liouville_polynomial}
A Li\'{e}nard  differential  system \eqref{Lienard_gen} from family ($B$)
possesses the Darboux integrating factor
\begin{equation}\label{Lienard_degenerate_Liouville_IF_polynomial}
M(x,y)=[y-p_1(x)]^{d_1}[y-p_2(x)]^{d_2}, \quad d_1,d_2\in\mathbb{C}\setminus\{0\},
\end{equation}
where $p_1(x)$ and $p_2(x)$ are distinct polynomials, if and only if  one of
the following assertions is valid.

\begin{enumerate}

\item The system is of the form
    \eqref{Lienard_degenerate_Darboux_explicit} and the polynomials
    $p_1(x)$ and $p_2(x)$ are linearly dependent:
    $p_2(x)=(f_0+\delta)/(f_0-\delta)p_1(x)$. In this case the parameters
    $d_1$ and $d_2$ can be chosen as $d_1=d_2=-1$ and the following
    relations $p_1(x)=q_1(x)$, $p_2(x)=q_2(x)$ are valid. In fact, there
    exists a family of Darboux integrating
    factors~\eqref{Lienard_degenerate_Liouville_IF_polynomial} that are
    products of the integrating factor
    $M_0(x,y)=[y-q_1(x)]^{-1}[y-q_2(x)]^{-1}$ and the Darboux first
    integrals $I^{\varkappa}(x,y)$, where the function $I(x,y)$ is given
    by expression \eqref{Lienard_degenerate_Darboux_FI1} and
    $\varkappa\in\mathbb{C}$.

\item The system reads as
\begin{equation}\label{Lienard_degenerate_Liouville_system_polynomial}
\begin{gathered}
x_t=y,\quad y_t=\left[\beta(l+k)u^{k-1}+\frac{\{(2d_1+1)l+k\}l}{k-l}u^{l-1}\right]u_xy\\
-\left[l\beta^2u^{2k-1}+\frac{\{(2d_1+1)l+k\}l\beta}{k-l}u^{k+l-1}+\frac{(ld_1+k)(d_1+1)l^2}{(k-l)^2}u^{2l-1}\right]u_x,
\end{gathered}
\end{equation}
where $k$ and $l$ are relatively prime both non-unit natural numbers,
$u(x)$ is a polynomial of degree $(m+1)/\max\{k,l\}$ and
$\beta\in\mathbb{C}\setminus\{0\}$. The polynomials $p_1(x)$ and $p_2(x)$
can be represented as
\begin{equation}\label{Lienard_degenerate_Liouville_system_polynomial_q}
\begin{gathered}
p_1(x)=\beta u^k(x)+\frac{(d_1+1)l}{k-l}u^l(x),\quad p_2(x)=\beta u^k(x)+\frac{(ld_1+k)}{k-l}u^l(x)
\end{gathered}
\end{equation}
and the parameter $d_2$ is given by the relation $d_2=-(d_1+1+k/l)$.

\end{enumerate}
\end{theorem}

\begin{proof}
Expression \eqref{Lienard_degenerate_Liouville_IF_polynomial} gives a Darboux
integrating factor of a Li\'{e}nard  differential  system if and only if the
system possesses the invariant algebraic curves $y-p_1(x)=0$ and $y-p_2(x)=0$
such that the following condition
$d_1\lambda_1(x,y)+d_2\lambda_2(x,y)-f(x)=0$ is identically satisfied. It is
straightforward to find the cofactor $\lambda_j(x,y)$ of  the invariant
algebraic curve $y-p_j(x)=0$. The result is
$\lambda_j(x,y)=-f(x)-p_{j,\,x}(x)$, $j=1$, $2$. Thus, we arrive at the
condition
\begin{equation}\label{Lienard_degenerate_Liouville_polynomial_cond1}
\begin{gathered}
(d_1+d_2+1)f(x)+d_1p_{1,\,x}(x)+d_2p_{2,\,x}(x)=0.
\end{gathered}
\end{equation}
If the following restriction $d_2=-1-d_1$ is valid, then integrating equation
\eqref{Lienard_degenerate_Liouville_polynomial_cond1} with respect to the
polynomial $p_1(x)$, we obtain $p_1(x)=(d_1+1)p_2(x)/d_1+\beta$, where
$\beta\in\mathbb{C} $ is a constant of integration. Recalling the fact that
$y=p_1(x)$ and $y=p_2(x)$ are polynomial solutions of equation
\eqref{Lienard_y_x}, we find the polynomials $f(x)$ and $g(x)$. The
polynomial $f(x)$ can be represented in the form
\begin{equation}\label{Lienard_degenerate_Liouville_polynomial_cond1_fg}
\begin{gathered}
f(x)=-\frac{(2d_1+1)p_2(x)+d_1(d_1+1)\beta}{d_1(p_2(x)+d_1\beta)}p_{2,\,x}(x).
\end{gathered}
\end{equation}
Analyzing this expression, we conclude that the function on the right-hand
side is not a polynomial whenever $p_2(x)$ is non-constant.

Let us consider the case $d_2\neq-1-d_1$. We find the polynomials $f(x)$ and
$g(x)$ from condition \eqref{Lienard_degenerate_Liouville_polynomial_cond1}
and equation \eqref{Lienard_y_x}, where we set $y(x)=p_1(x)$. Substituting
the resulting expressions into equation \eqref{Lienard_y_x} and recalling the
fact that $y(x)=p_2(x)$ is a solution of the latter, we obtain the equation
\begin{equation}\label{Lienard_degenerate_Liouville_polynomial_cond1_q1_2}
\begin{gathered}
\{(d_2+1)p_1(x)+d_1p_2(x)\}p_{1,\,x}(x)-\{d_2p_1(x)+(d_1+1)p_2(x)\}p_{2,\,x}(x)=0.
\end{gathered}
\end{equation}
Introducing the polynomial $v(x)$ according to the rule $v(x)=p_2(x)-p_1(x)$,
we substitute the relation $p_2(x)=p_1(x)+v(x)$ into equation
\eqref{Lienard_degenerate_Liouville_polynomial_cond1_q1_2}. Integrating the
result with respect to the polynomial $p_1(x)$, we obtain
\begin{equation}\label{Lienard_degenerate_Liouville_polynomial_cond1_q1}
\begin{gathered}
d_2=-2-d_1:\quad p_1(x)=-(d_1+1)v(x)\ln v(x)+\beta v(x);\hfill\\
d_2\neq-2-d_1:\quad p_1(x)=\beta v^{-d_2-d_1-1}(x)-\frac{d_1+1}{d_2+d_1+2}v(x),
\end{gathered}
\end{equation}
where $\beta\in\mathbb{C}$ is a constant of integration. Analyzing the first possibility,
we need to set $d_1=-1$. As a result, we get Darboux integrable family
\eqref{Lienard_degenerate_Darboux_explicit} of Li\'{e}nard  differential
systems. The parameter $\beta$ can be derived with the help of the dominant
behavior of the polynomials $p_1(x)$ and $p_2(x)$, which now coincide with
$q_1(x)$ and $q_2(x)$, respectively. Recalling the fact that the product of
an integrating factor and a first integral is again an integrating factor, we
obtain the family of integrating factors $M_0(x,y)I^{\varkappa}(x,y)$, where
$M_0(x,y)=[y-q_1(x)]^{-1}[y-q_2(x)]^{-1}$, $\varkappa\in\mathbb{C}$, and the
Darboux first
    integral $I(x,y)$ is given
    by expression \eqref{Lienard_degenerate_Darboux_FI1}.

Now we turn to the case $d_2\neq-2-d_1$. If $\beta=0$, then we get the
equality $p_2(x)=-(d_2+1)p_1(x)/(d_1+1)$. In addition, we obtain the
following representations of the polynomials $f(x)$ and $g(x)$:
$f(x)=(d_2-d_1)p_{1,\,x}(x)/(d_1+1)$ and
$g(x)=-(d_2+1)p_1(x)p_{1,\,x}(x)/(d_1+1)$. Considering these expressions, we
again arrive at Darboux  integrable systems
\eqref{Lienard_degenerate_Darboux_explicit}.

Thus, it is without loss of generality to set $\beta\neq0$. Recalling the
restriction $d_2\neq-1-d_1$, we introduce relatively prime both non-unit
natural numbers $k$ and $l$ satisfying the condition $d_2+d_1+1=-k/l$.
Analyzing expression
\eqref{Lienard_degenerate_Liouville_polynomial_cond1_q1}, we conclude that
there exists a polynomial $u(x)$ such that the following relation
$v(x)=u^l(x)$ holds. In this way we express the polynomials $p_1(x)$ and
$p_2(x)$ via the polynomial $u(x)$. The result is given in expression
\eqref{Lienard_degenerate_Liouville_system_polynomial_q}. By construction the
degree of the polynomial $u(x)$ equals $(m+1)/\max\{k,l\}$.

\end{proof}

Using integrating factor \eqref{Lienard_degenerate_Liouville_IF_polynomial},
we find the following expression of a Liouvillian first integral
\begin{equation}\label{Lienard_degenerate_Liouvillian_FI}
\begin{gathered}
I(x,y)=\frac{p_2(x)B\left(\frac{y-p_1(x)}{p_2(x)-p_1(x)};1+d_1,-d_1-\frac{k}{l}\right)}{\{p_2(x)-p_1(x)\}^{\frac{k}{l}}}-
\frac{B\left(\frac{y-p_1(x)}{p_2(x)-p_1(x)};1+d_1,1-d_1-\frac{k}{l}\right)}{\{p_2(x)-p_1(x)\}^{\frac{k}{l}-1}}
\end{gathered}
\end{equation}
of systems \eqref{Lienard_degenerate_Liouville_system_polynomial}. The
polynomials $p_1(x)$ and $p_2(x)$ are given by relation
\eqref{Lienard_degenerate_Liouville_system_polynomial_q}. Symbol
$B(s;\alpha,\delta)$ denotes the incomplete beta function
\begin{equation}\label{Lienard_degenerate_Liouvillian_incomplete_beta}
B(s;\alpha,\delta)=\int_0^s z^{\alpha-1}(1-z)^{\beta-1}dz.
\end{equation}
The family of systems \eqref{Lienard_degenerate_Liouville_system_polynomial} can be transformed to the following simple form
\begin{equation}
 \label{Lienard_Inegrability_Bpartial_0_Sundman}
 \begin{gathered}
s_{\tau}=z,\quad z_{\tau}=\left[\beta(l+k)s^{k-1}+\frac{\{(2d_1+1)l+k\}l}{k-l}s^{l-1}\right]z\\
-\left[l\beta^2s^{2k-1}+\frac{\{(2d_1+1)l+k\}l\beta}{k-l}s^{k+l-1}+\frac{(ld_1+k)(d_1+1)l^2}{(k-l)^2}s^{2l-1}\right]
\end{gathered}
\end{equation}
via the generalized Sundman transformation $s(\tau)=u(x)$, $z(\tau)=y$, $d\tau=u_x(x)dt$. Substituting $u(x)=s$, $y=z$ into \eqref{Lienard_degenerate_Liouvillian_FI} and \eqref{Lienard_degenerate_Liouville_system_polynomial_q}, we find a Liouvillian first integral for systems~\eqref{Lienard_Inegrability_Bpartial_0_Sundman}.

The careful examination of expression
\eqref{Lienard_degenerate_Liouville_system_polynomial_q} shows that systems
\eqref{Lienard_degenerate_Liouville_system_polynomial} are resonant near
infinity whenever the following inequality $k>l$ is valid. Indeed, two
distinct polynomial solutions of equation \eqref{Lienard_y_x} have the
coinciding dominant behavior near the point $x=\infty$ only in a resonant
case. Let us find the necessary and sufficient conditions of the Liouvillian
integrability in the non-resonant case.

\begin{theorem}\label{T:Lienard_degenerate_Liouville1}
A Li\'{e}nard  differential  system \eqref{Lienard_gen} satisfying
 the conditions $\deg g=2\deg f+1$ and $\delta/f_0\not\in\mathbb{Q}$
is Liouvillian integrable if and only if the system is either Darboux
integrable and reads as \eqref{Lienard_degenerate_Darboux_explicit} or takes
the form \eqref{Lienard_degenerate_Liouville_system_polynomial}, where $k<l$
and the following normalization
\begin{equation}\label{Lienard_degenerate_Liouvillian_normalization}
d_1=\frac{(k-l)f_0}{2l\delta }-\frac{l+k}{2l},\quad u_0=\left(-\frac{\delta}{m+1}\right)^{\frac1{l}}
\end{equation}
is introduced. By $u_0$ we denote the highest-degree coefficient of
the polynomial $u(x)$. In the case of systems
\eqref{Lienard_degenerate_Liouville_system_polynomial} the Darboux
integrating factor is the following
\begin{equation}\label{Lienard_degenerate_Darboux_IF}
M(x,y)=\left\{y-q_1(x)\right\}^{\frac{(k-l)f_0}{2l\delta }-\frac{l+k}{2l}}\left\{y-q_2(x)
\right\}^{\frac{(l-k)f_0}{2l\delta }-\frac{l+k}{2l}}
\end{equation}
with the polynomials $q_j(x)\equiv p_j(x)$, $j=1$, $2$ given by expression
\eqref{Lienard_degenerate_Liouville_system_polynomial_q}.
\end{theorem}

\begin{proof}
By direct computations we verify that systems
\eqref{Lienard_degenerate_Darboux_explicit} and
\eqref{Lienard_degenerate_Liouville_system_polynomial} are Liouvillian
integrable provided that all other conditions of the theorem are satisfied.
This observation proves the sufficiency of these conditions.

Let us prove their necessity. It follows from Lemmas \ref{L:Lienard_exp_inv1}
and \ref{L:Lienard_exp_degenerate} that exponential invariants cannot arise
in a Darboux integrating factor. Consequently, a Darboux integrating factor
is constructed from generating polynomials of invariant algebraic curves. It
is straightforward to see that if there are no invariant algebraic curves,
then Darboux integrating factors do not exist. In view of Theorem
\ref{T:Lienard_degenerate} we need to consider three distinct cases.

\textit{Case 1.} Let us suppose that a Liouvillian integrable Li\'{e}nard  differential  system \eqref{Lienard_gen} satisfying
 the conditions $\deg g=2\deg f+1$ and $\delta/f_0\not\in\mathbb{Q}$ has only
one irreducible invariant algebraic curve with a generating polynomial of the
first degree with respect to $y$. This curve reads as $y-q_k(x)=0$, $k=1$ or
$k=2$, and has the cofactor $\lambda(x,y)=-f(x)-q_{k,x}(x)$. A Darboux
integrating factor can be represented in the form
\begin{equation}\label{Lienard_degenerate_Liouville_IF_1}
M(x,y)=[y-q_k(x)]^{d_k},\quad d_k\in\mathbb{C}\setminus\{0\}.
\end{equation}
This integrating factor exists if and only if the following condition
\begin{equation}\label{Lienard_degenerate_Liouville_cond_1}
d_k\{f(x)+q_{k,x}(x)\}+f(x)=0
\end{equation}
is identically valid. Balancing the terms at $x^m$ in this relation, we
obtain
\begin{equation}\label{Lienard_degenerate_Liouville_cond_2}
k=1:\quad d_1=-\frac{2f_0}{\delta+f_0};\quad k=2:\quad d_2=\frac{2f_0}{\delta-f_0}.
\end{equation}
Expressing $f(x)$ and $g(x)$ from relation
\eqref{Lienard_degenerate_Liouville_cond_1} and equation \eqref{Lienard_y_x}
with $y(x)=q_k(x)$, we see that our Li\'{e}nard  differential  system is of
the form \eqref{Lienard_degenerate_Darboux_explicit} and possesses two
distinct irreducible invariant algebraic curves. It is a contradiction.

\textit{Case 2.} Let us suppose that  a Liouvillian integrable Li\'{e}nard  differential  system \eqref{Lienard_gen} satisfying
 the conditions $\deg g=2\deg f+1$ and $\delta/f_0\not\in\mathbb{Q}$ has an
irreducible invariant algebraic curve with a generating polynomial of the
second degree with respect to $y$. This curve is given by the expression
$\{[y-y^{(1)}_{\infty}(x)][y-y^{(2)}_{\infty}(x)]\}_{+}=0$. Its cofactor
reads as $\lambda(x,y)=-2f(x)-q_{1,x}(x)-q_{2,x}(x)$. A Darboux integrating factor
can be represented in the form
\begin{equation}\label{Lienard_degenerate_Liouville_IF_2}
M(x,y)=F^{d}(x,y),\, F(x,y)=\{[y-y^{(1)}_{\infty}(x)][y-y^{(2)}_{\infty}(x)]\}_{+},\, d\in\mathbb{C}\setminus\{0\}
\end{equation}
and exists if and only if the following condition
\begin{equation}\label{Lienard_degenerate_Liouville_cond_3}
d\{2f(x)+q_{1,x}(x)+q_{2,x}(x)\}+f(x)=0
\end{equation}
is identically satisfied. Considering the coefficients of $x^m$ in this
condition yields the value of $d$: $d=-1$. Further, we represent the
polynomial $F(x,y)$ in the form $F(x,y)=y^2+v(x)y+w(x)$, where $v(x)$,
$w(x)\in\mathbb{C}[x]$. Substituting expression $M(x,y)=F^{-1}(x,y)$ into the
partial differential equation
\begin{equation}\label{Lienard_degenerate_Liouville_PDE_1}
yM_x-[f(x)y+g(x)]M_y-f(x)M=0,
\end{equation}
we get rid of the denominator. Setting to zero the coefficients of different powers of $y$, we get
$w(x)=\beta v(x)^2$, where $\beta\in\mathbb{C}\setminus\{0\}$. This equality
contradicts irreducibility of the polynomial~$F(x,y)$.

\textit{Case 3.} Now we assume that a Liouvillian integrable Li\'{e}nard  differential  system \eqref{Lienard_gen} satisfying
 the conditions $\deg g=2\deg f+1$ and $\delta/f_0\not\in\mathbb{Q}$ has two distinct
irreducible invariant algebraic curves  $y-q_1(x)=0$ and $y-q_2(x)=0$. Their
cofactors are the following $\lambda_1(x,y)=-f(x)-q_{1,x}(x)$ and
$\lambda_2(x,y)=-f(x)-q_{2,x}(x)$, respectively. Condition
\eqref{JLM_gen_cond} enabling the existence of a Darboux integrating factor
\begin{equation}\label{Lienard_degenerate_Liouville_IF_3}
M(x,y)=[y-q_1(x)]^{d_1}[y-q_2(x)]^{d_2},\quad d_1,d_2\in\mathbb{C},\quad |d_1|+|d_2|>0
\end{equation}
is of the form
\begin{equation}\label{Lienard_degenerate_Liouville_cond_4}
d_1\{f(x)+q_{1,x}(x)\}+d_2\{f(x)+q_{2,x}(x)\}+f(x)=0.
\end{equation}
All the Li\'{e}nard  differential  systems from family ($B$) with integrating
factor of the form~\eqref{Lienard_degenerate_Liouville_IF_3} have been
identified in Theorem \ref{T:Lienard_degenerate_Liouville_polynomial}. We
need to extract non-resonant systems from those given by expression
\eqref{Lienard_degenerate_Liouville_system_polynomial}. Thus, the polynomials
$q_1(x)$ and $q_2(x)$ necessarily have distinct dominant terms. This fact
yields the inequality $k<l$. Finally, we need to introduce the normalization
adopted at the beginning of this section. Using relations
\eqref{Lienard_degenerate_Puiseux_series_dominant} and
\eqref{Lienard_degenerate_Puiseux_series_Polynomial_parts}, we obtain
expression \eqref{Lienard_degenerate_Liouvillian_normalization}.

\end{proof}

\textit{Corollary.} Li\'{e}nard  differential systems
\eqref{Lienard_degenerate_Darboux_system_time} with a non-autonomous Darboux
first integral \eqref{Lienard_degenerate_Darboux_FI1_time}  are Liouvillian
integrable. These systems have the Darboux integrating factor
\begin{equation}\label{Lienard_degenerate_Liouville_IF_time_partial}
\begin{gathered}
M(x,y)=\frac{\left[y+\frac{\delta+f_0}{2(m+1)}(x-x_0)^{m+1}+\frac{\omega}{2(m+1)\delta}(x-x_0)\right]^{\frac{mf_0-(m+2)\delta}{2(m+1)\delta}}}
{\left[y+\frac{f_0-\delta}{2(m+1)}(x-x_0)^{m+1}+\frac{\omega}{2(m+1)\delta}(x-x_0)\right]^\frac{mf_0+(m+2)\delta}{2(m+1)\delta}}.
\end{gathered}
\end{equation}

\begin{proof} We establish the validity of the statement substituting relations
\begin{equation}\label{Lienard_degenerate_Liouville_IF_time_parameters}
\begin{gathered}
k=1,\, l=m+1,\, u(x)=u_0(x-x_0),\, \beta=-\frac{\omega}{2(m+1)\delta u_0},\, u_0=\left\{-\frac{\delta}{m+1}\right\}^{\frac1{m+1}}
\end{gathered}
\end{equation}
into expressions \eqref{Lienard_degenerate_Liouville_system_polynomial} and
\eqref{Lienard_degenerate_Darboux_IF}. In addition, we recall that the
parameter $d_1$ reads
as~\eqref{Lienard_degenerate_Liouvillian_normalization}.
\end{proof}

\begin{theorem}\label{T:Lienard_degenerate_Liouville2}
A Li\'{e}nard  differential  system \eqref{Lienard_gen} satisfying
 the conditions $\deg g=2\deg f+1$ and $\delta=0$
is Liouvillian integrable if and only if the system has the irreducible
invariant algebraic curve
    $y-q(x)=0$ and an exponential invariant $E(x,y)=\exp[u(x)/(y-q(x))]$
    such that one of the following assertions is valid.

\begin{enumerate}

\item The system is Darboux integrable and takes the form
    \eqref{Lienard_degenerate_Darboux_del0_f_g}. A related Darboux
    integrating factor reads as
\begin{equation}\label{Lienard_degenerate_Liouville_IF_main2}
M(x,y)=\frac{1}{[y-q(x)]^{2}}
\end{equation}
The polynomial $u(x)$  arising in the exponential invariant is
$u(x)=\alpha q(z)$, $ \alpha\in\mathbb{C}\setminus\{0\}$.

\item The system is of the form
    \begin{equation}\label{Lienard_degenerate_Liouville_f_g_Case2}
    \begin{gathered}
   x_t=y,\quad     y_t=-\left[\frac{2l^2}{l-k}v^{l-1}-(l+k)\beta v^{k-1}\right]v_xy\\
   +   \left[\frac{2l^2\beta}{l-k}v^{l+k-1}-\frac{l^3}{(l-k)^2}v^{2l-1}-l\beta^2 v^{2k-1}\right]v_x,
\end{gathered}
\end{equation}
where $\beta\in\mathbb{C}\setminus\{0\}$, $v(x)$ is a non-constant
polynomial, $k$ and $l$ are relatively prime natural numbers satisfying
the inequality $k<l$.  The associated Darboux
    integrating factor reads~as
\begin{equation}\label{Lienard_degenerate_Liouville_IF_main3}
M(x,y)=[y-q(x)]^{-\frac{l+k}{l}}\exp\left[\frac{v^l(x)}{y-q(x)}\right],
\quad q(x)=-\frac{l}{l-k}v^l+\beta v^k.
\end{equation}
In addition, the following relation $m+1=l\deg v$ is valid.

\end{enumerate}

\end{theorem}

\begin{proof}
It is straightforward to verify that expressions
\eqref{Lienard_degenerate_Liouville_IF_main2} and
\eqref{Lienard_degenerate_Liouville_IF_main3} are Darboux integrating factors
of Li\'{e}nard  differential  systems \eqref{Lienard_gen} satisfying
 the restrictions $\deg g=2\deg f+1$ and $\delta=0$ whenever all other
 conditions of the theorem are satisfied.

 Let us prove the converse statement. Suppose we consider a Liouvillian
 integrable Li\'{e}nard  differential  system \eqref{Lienard_gen} such that
 $\deg g=2\deg f+1$ and $\delta=0$. By Theorems \ref{T:Liouville}, \ref{T:Lienard_degenerate} and Lemmas
 \ref{L:Lienard_exp_inv1}, \ref{L:Lienard_exp_degenerate} the system has the
irreducible invariant algebraic curve $y-q(x)=0$ and a Darboux integrating
factor that can be represented in the form
\begin{equation}\label{Lienard_degenerate_Liouville_IF_main4}
M(x,y)=[y-q(x)]^{d}\exp\left[\frac{u(x)}{y-q(x)}\right],\quad d\in\mathbb{C},
\end{equation}
where we suppose that $u(x)\equiv0$
 whenever the exponential invariant $E(x,y)=\exp[u(x)/(y-q(x))]$ related to the invariant
 algebraic curve $y-q(x)=0$ either does not
 exist or is not involved into an explicit expression of the integrating factor.
 Condition \eqref{JLM_gen_cond} with $\omega=0$
 now takes the form
\begin{equation}\label{Lienard_degenerate_Liouville_cond 1}
d[f(x)+q_x(x)]-u_x(x)+f(x)=0.
\end{equation}
Further, we shall consider several distinct cases separately.

\textit{Case 1.} Let us begin with the case $u(x)\equiv0$. Substituting the
asymptotic relations $f(x)=f_0x^m+o(x^m)$ and
$q(x)=-f_0x^{m+1}/(2\{m+1\})+o(x^{m+1})$, $x\rightarrow\infty$ into condition
\eqref{Lienard_degenerate_Liouville_cond 1} and setting to zero the
coefficient of $x^m$, we obtain the equalities $d=-2$ and $f(x)=-2q_x(x)$.
Recalling the fact that $y=q(x)$ is the  polynomial solution of equation
\eqref{Lienard_y_x}, we find the polynomial $g(x)$ as given
in~\eqref{Lienard_degenerate_Darboux_del0_f_g}. Using
Theorem~\ref{T:Lienard_degenerate_Darboux2}, we conclude that the related
Li\'{e}nard differential  system possesses a Darboux first integral
\eqref{Lienard_degenerate_Darboux_FI1_del0} and the exponential invariants
$E(x,y)=\exp[\alpha q(x)/(y-q(x))]$, $ \alpha\in\mathbb{C}\setminus\{0\}$.
Note that integrating the differential form
\begin{equation}\label{Lienard_degenerate_Liouville_differential_form}
\frac{ydy+(qq_x-2q_xy)dx}{\{y-q(x)\}^2}.
\end{equation}
yields a first integral in the form $\ln I(x,y)$ with the function $I(x,y)$
given by expression~\eqref{Lienard_degenerate_Darboux_FI1_del0}.

\textit{Case 2.} Now let us suppose that the polynomial $u(x)$ is not
identically zero. We see that the system under consideration has the
exponential invariant $E(x,y)=\exp[u(x)/(y-q(x))]$ with the polynomial $u(x)$
satisfying equation~\eqref{Lienard_degenerate_Int2}. Expressing $f(x)$ from
the latter equation and substituting the result into condition
\eqref{Lienard_degenerate_Liouville_cond 1} yields the relation
\begin{equation}\label{Lienard_degenerate_Liouville_cond 2}
(d+1)qu_x+q_xu+uu_x=0.
\end{equation}
This relation viewed as an ordinary differential equation with respect to the
polynomial $q(x)$ can be integrated. Thus, we find the expressions
\begin{equation}\label{Lienard_degenerate_Liouville_cond 3}
\begin{gathered}
d\neq-2:\quad q(x)=-\frac1{d+2}u+\beta u^{-(d+1)};\quad
d=-2:\quad q(x)=(\beta - \ln u)u,\hfill
\end{gathered}
\end{equation}
where $\beta\in\mathbb{C}$ is a constant of integration. If $d=-2$, then
$q(x)$ is not a polynomial. Further, we set $d\neq-2$.  The case $\beta=0$
again leads to a Darboux integrable family of Li\'{e}nard  differential
systems given in Theorem \ref{T:Lienard_degenerate_Darboux2}. Thus, we
suppose that the constant $\beta$ is non-zero. Recalling the fact that $q(x)$
and $u(x)$ are polynomials with the dominant behavior
$q(x)=-f_0x^{m+1}/(2\{m+1\})+o(x^{m+1})$ and $u(x)=u_0x^{m+1}+o(x^{m+1})$,
$u_0\in\mathbb{C}\setminus\{0\}$ near the point $x=\infty$, we find two
possibilities
\begin{equation}\label{Lienard_degenerate_Liouville_cond 4}
\begin{gathered}
d=-1:\quad u(x)=\beta-q(x);\quad d=-\frac{l+k}{l}:\quad u(x)=v^l(x).\hfill
\end{gathered}
\end{equation}
In this expressions $l$ and $k$ are relatively prime natural numbers
satisfying the restriction $k<l$ and $v(x)$ is a non-constant polynomial.
Analyzing the possibility $d=-1$, we substitute the equality $u(x)=\beta-q(x)
$ into relation  \eqref{Lienard_degenerate_Int2} and find the polynomial
$f(x)$. The result is
\begin{equation}\label{Lienard_degenerate_Liouville_cond 5}
f(x)=\frac{(2q(x)-\beta)q_x(x)}{\beta-q(x)}.
\end{equation}
Let $x_0$ be a zero of the polynomial  $\beta-q(x) $. Considering the
behavior near $x_0$ of the rational function on the right-hand side of
expression \eqref{Lienard_degenerate_Liouville_cond 5}, we see that $f(x)$ is
not a polynomial whenever $\beta\neq0$.

Finally, we suppose that the following relations $d=-(l+k)/l$ and
$u(x)=v^l(x)$ are valid. We use expressions
\eqref{Lienard_degenerate_Liouville_cond 3},
\eqref{Lienard_degenerate_Liouville_cond 1}, and \eqref{Lienard_y_x} to find
the polynomials $f(x)$, $g(x)$, and $q(x)$ as given in relations
\eqref{Lienard_degenerate_Liouville_f_g_Case2} and
\eqref{Lienard_degenerate_Liouville_IF_main3}. In addition, we verify that
equation  \eqref{Lienard_degenerate_Int2} is identically satisfied. The proof
is completed.

\end{proof}

\textit{Corollary 1.} Li\'{e}nard  differential systems
\eqref{Lienard_degenerate_Darboux_del0_f_g_time} possessing a non-autonomous
Darboux first integral \eqref{Lienard_degenerate_Darboux_FI1_del0_time}  are
Liouvillian integrable with the Darboux integrating factor
\begin{equation}\label{Lienard_degenerate_Liouville_IF_main3_partial}
\begin{gathered}
M(x,y)=\exp\left[\frac{mf_0(x-x_0)^{m+1}}{(m+1)\{2(m+1)y+f_0(x-x_0)^{m+1}+2\omega(x-x_0)\}}\right]\\
\times\left[y+\frac{f_0(x-x_0)^{m+1}}{2(m+1)}+\frac{\omega(x-x_0)}{m+1}\right]^{-\frac{m+2}{m+1}}.
\end{gathered}
\end{equation}

\begin{proof} We prove the validity of the statement substituting relations
\begin{equation}\label{Lienard_degenerate_Liouville_IF_main3_parameters}
\begin{gathered}
k=1,\, l=m+1,\, v(x)=v_0(x-x_0),\, \beta=-\frac{\omega}{(m+1)v_0},\, v_0=\left\{\frac{mf_0}{2(m+1)^2}\right\}^{\frac1{m+1}}
\end{gathered}
\end{equation}
into expressions \eqref{Lienard_degenerate_Liouville_f_g_Case2} and
\eqref{Lienard_degenerate_Liouville_IF_main3}.
\end{proof}

\textit{Corollary 2.} If the following inequality $k>l$ holds, then systems
\eqref{Lienard_degenerate_Liouville_f_g_Case2} are also Liouvillian
integrable Li\'{e}nard differential systems from family ($B$). But these
systems are resonant near infinity. The related Darboux integrating factor
again is given by expression \eqref{Lienard_degenerate_Liouville_IF_main3}.

A Liouvillian first integral produced by integrating factor
\eqref{Lienard_degenerate_Liouville_IF_main3} reads as
\begin{equation}\label{Lienard_degenerate_Liouville_FI}
I(x,y)=v^{l-k}(x)\gamma\left(-\frac{l-k}{l},\frac{v^l(x)}{q(x)-y}\right)-
\frac{q(x)}{v^{k}(x)}\gamma\left(\frac{k}{l},\frac{v^l(x)}{q(x)-y}\right),
\end{equation}
where the polynomial $q(x)$ is given in expression
\eqref{Lienard_degenerate_Liouville_f_g_Case2} and  $\gamma(\delta,s)$ is the
lower incomplete Gamma function
\begin{equation}\label{lower_incomplete_Gamma}
\gamma(\delta,s)=\int_0^st^{\delta-1}\exp(-t)dt.
\end{equation}
Note that we need to consider the analytic continuation of this integral for
complex or real non-positive values of~$s$. If $k=1$ and $l=2$, then we obtain another representation of a Liouvillian
first integral
\begin{equation}\label{Lienard_degenerate_Liouville_FI_a}
\begin{gathered}
I(x,y)=2\sqrt{\beta v(x)-2v^2(x)-y}\exp\left[\frac{v^2(x)}{y+2v^2(x)-\beta v(x)}\right]\\
-\sqrt{\pi}\beta\erfc\left[\frac{v(x)}{\sqrt{\beta v(x)-2v^2(x)-y}}\right],
\end{gathered}
\end{equation}
where $\erfc(s)$ is the complementary error function
\begin{equation}\label{error_function}
\erfc(s)=\frac{2}{\sqrt{\pi}}\int_s^{\infty}\exp\left(-t^2\right)dt.
\end{equation}
The family of systems \eqref{Lienard_degenerate_Liouville_f_g_Case2} can be transformed to the following simple form
\begin{equation}
 \label{Lienard_Inegrability_Bpartial_1_Sundman}
s_{\tau}=z,\, z_{\tau}=-\left[\frac{2l^2}{l-k}s^{l-1}-(l+k)\beta s^{k-1}\right]z\\
   +   \frac{2l^2\beta}{l-k}s^{l+k-1}-\frac{l^3}{(l-k)^2}s^{2l-1}-l\beta^2 s^{2k-1}
\end{equation}
via the generalized Sundman transformation $s(\tau)=v(x)$, $z(\tau)=y$, $d\tau=v_x(x)dt$. Substituting $v(x)=s$, $y=z$ into \eqref{Lienard_degenerate_Liouville_FI}, we find a Liouvillian first integral for systems \eqref{Lienard_Inegrability_Bpartial_1_Sundman}.

It seems that Liouvillian integrable families of Li\'{e}nard  differential
systems given in Theorems
\ref{T:Lienard_degenerate_Liouville_polynomial}, \ref{T:Lienard_degenerate_Liouville1}, and
\ref{T:Lienard_degenerate_Liouville2} are new with the exception of systems
\eqref{Lienard_degenerate_Darboux_del0_f_g_time}. The latter are presented
by T. Stachowiak \cite{Stachowiak}.

Now let us investigate the existence of Jacobi last multipliers with a
time-dependent exponential factor.

\begin{lemma}\label{L:Lienard__degenerate Liouville_t1}

 A Li\'{e}nard  differential  system \eqref{Lienard_gen} satisfying
 the conditions $\deg g=2\deg f+1$ and $\delta/f_0\not\in\mathbb{Q}$ has a non-autonomous Darboux--Jacobi last
 multiplier of the form \eqref{JLM_gen} if and only if one the following assertions
 is valid.
\begin{enumerate}

\item The system under consideration possesses one irreducible invariant
    algebraic curve $y-q_k(x)=0$, where $k=1$ or $k=2$, such that the
    polynomials $f(x)$ and $q_k(x)$ identically satisfy the condition
\begin{equation}
 \label{Lienard_degenerate_Inegrability_time1}
 \begin{gathered}
k=1:\, (f_0-\delta)f(x)+2f_0q_{1,x}(x)+(f_0+\delta)\omega=0,\, \omega\in\mathbb{C}\setminus\{0\},\\
k=2:\, (f_0+\delta)f(x)+2f_0q_{2,x}(x)+(f_0-\delta)\omega=0,\, \omega\in\mathbb{C}\setminus\{0\}.
\end{gathered}
\end{equation}
A related Darboux--Jacobi last multiplier reads as
\begin{equation}
 \label{Lienard_degenerate_Inegrability_time2}
 \begin{gathered}
k=1:\quad M(x,y,t)=[y-q_1(x)]^{-\frac{2f_0}{\delta+f_0}}\exp(\omega t),\\
k=2:\quad M(x,y,t)=[y-q_2(x)]^{\frac{2f_0}{\delta-f_0}}\exp(\omega t).\hfill
\end{gathered}
\end{equation}

\item The system under consideration possesses one irreducible invariant
    algebraic curve $F(x,y)=0$ with $\displaystyle
    F(x,y)=\left\{\left[y-y^{(1)}_{\infty}(x)\right]\left[y-y^{(2)}_{\infty}(x)\right]\right\}_+$
    such that the polynomials $f(x)$, $\displaystyle
   q_1(x)=\left\{y^{(1)}_{\infty}(x)\right\}_+$, and $\displaystyle
   q_2(x)=\left\{y^{(2)}_{\infty}(x)\right\}_+$ identically satisfy the
    condition
\begin{equation}
 \label{Lienard_degenerate_Inegrability_time3}
 \begin{gathered}
f(x)+q_{1,x}(x)+q_{2,x}(x)+\omega=0,\, \omega\in\mathbb{C}\setminus\{0\}.
\end{gathered}
\end{equation}
A related Darboux--Jacobi last multiplier reads as
\begin{equation}
 \label{Lienard_degenerate_Inegrability_time4}
 \begin{gathered}
M(x,y,t)=\frac{\exp(\omega t)}{F(x,y)},\quad F(x,y)=\left\{\left[y-y^{(1)}_{\infty}(x)\right]\left[y-y^{(2)}_{\infty}(x)\right]\right\}_+.
\end{gathered}
\end{equation}

\item The system under consideration possesses two distinct irreducible
    invariant algebraic curves $y-q_1(x)=0$ and $y-q_2(x)=0$ such that
    the polynomials $f(x)$, $\displaystyle q_1(x)$, and $\displaystyle
   q_2(x)$ identically satisfy the condition
\begin{equation}
 \label{Lienard_degenerate_Inegrability_time5}
 \begin{gathered}
\,[(2d_2+1)\delta-f_0]f(x)+[(\delta-f_0)d_2-2f_0]q_{1,x}(x)+(\delta+f_0)d_2q_{2,x}(x)\\
-(\delta+f_0)\omega=0,\quad
 d_2\in\mathbb{C},\quad\omega\in\mathbb{C}\setminus\{0\}.
\end{gathered}
\end{equation}
A related Darboux--Jacobi last multiplier reads as
\begin{equation}
 \label{Lienard_degenerate_Inegrability_time6}
 \begin{gathered}
M(x,y,t)=[y-q_1(x)]^{\frac{(\delta-f_0)d_2-2f_0}{\delta+f_0}}[y-q_2(x)]^{d_2}\exp(\omega t).
\end{gathered}
\end{equation}

\end{enumerate}

\end{lemma}

This lemma is proved similarly to Theorem
\ref{T:Lienard_degenerate_Liouville1}.

\begin{lemma}\label{L:Lienard__degenerate Liouville_t2}

 A Li\'{e}nard  differential  system \eqref{Lienard_gen} satisfying
 the conditions $\deg g=2\deg f+1$ and $\delta=0$ has a non-autonomous Darboux--Jacobi last
 multiplier of the form \eqref{JLM_gen} if and only if one the following assertions
 is valid.
\begin{enumerate}

\item The system under consideration possesses the irreducible invariant
    algebraic curve $y-q(x)=0$ such that the polynomials $f(x)$ and
    $g(x)$ can be represented~as
\begin{equation}
 \label{Lienard_degenerate_Inegrability_time_del0_1}
 \begin{gathered}
f(x)=-2q_x(x)-\omega,\quad g(x)=q(x)(q_x(x)+\omega),\quad \omega\in\mathbb{C}\setminus\{0\}.
\end{gathered}
\end{equation}
A related Darboux--Jacobi last multiplier reads as
\begin{equation}
 \label{Lienard_degenerate_Inegrability_time_del0_2}
 \begin{gathered}
 M(x,y,t)=\frac{\exp(\omega t)}{[y-q(x)]^{2}}.
\end{gathered}
\end{equation}

\item The system under consideration possesses the irreducible invariant
    algebraic curve $y-q(x)=0$ and the exponential invariant
    $E(x,y)=\exp[u(x)/(y-q(x))]$ such that the polynomials $f(x)$,
    $q(x)$, and $u(x)$ identically satisfy the condition
\begin{equation}
 \label{Lienard_degenerate_Inegrability_time_del0_3}
 \begin{gathered}
(d+1)f(x)+dq_x(x)=u_x(x)+\omega,\quad d\in\mathbb{C},\quad \omega\in\mathbb{C}\setminus\{0\}.
\end{gathered}
\end{equation}
A related Darboux--Jacobi last multiplier reads as
\begin{equation}
 \label{Lienard_degenerate_Inegrability_time_del0_4}
 \begin{gathered}
 M(x,y,t)=[y-q(x)]^{d}\exp\left[\frac{u(x)}{y-q(x)}\right]\exp(\omega t).
\end{gathered}
\end{equation}

\end{enumerate}

\end{lemma}

The proof of this lemma is analogous to the proof of Theorem
\ref{T:Lienard_degenerate_Liouville2}.

Concluding this section let us note that autonomous and non-autonomous
Darboux first integrals of non-resonant Li\'{e}nard  differential  systems
\eqref{Lienard_gen}
 satisfying
 the conditions $\deg f=1$, $\deg g=3$ and $\deg f=2$, $\deg g=5$ are
classified in \cite{DS2019, DS2021}.

\section{Integrability of  Li\'{e}nard differential systems from family ($C$)}\label{S:Lienard_C}

We start investigating integrability properties of Li\'{e}nard differential
systems from family~($C$) by proving the absence of exponential invariants
related to invariant algebraic curves. We use designations of Theorem
\ref{T1:Lienard_2m+1}. In particular, by $h^{(1)}(x)$ and $h^{(2)}(x)$ we
denote the initial parts of the Puiseux series $y^{(1)}_{\infty}(x)$ and
$y^{(2)}_{\infty}(x)$, accordingly. These initial parts involve monomials
with exponents exceeding $-(n+1)/2$. Recall that we use the designations
$\deg g=n$ and $\deg f =m$. Thus, the following inequality $n>2m+1$ is valid.

\begin{lemma}\label{L:Lienard_exp_inv3}
Suppose a Li\'{e}nard  differential  system~\eqref{Lienard_gen} from family
($C$) is not integrable with a rational first integral. Then
 this system
 does not have exponential invariants of the form
$E(x,y)=\exp\left\{h(x,y)/r(x,y)\right\}$, where $h(x,y)\in\mathbb{C}[x,y]$
and $r(x,y)\in\mathbb{C}[x,y]\setminus\mathbb{C}$ are relatively prime
polynomials.
\end{lemma}

\begin{proof}
We shall use the local theory presented in Section~\ref{S:Local}. If a
Li\'{e}nard differential system~\eqref{Lienard_gen} has an exponential
invariant $E(x,y)=\exp\left\{h(x,y)/r(x,y)\right\}$ with
$r(x,y)\in\mathbb{C}[x,y]\setminus\mathbb{C}$, then
$r(x,y)\in\mathbb{C}[x,y]\setminus\mathbb{C}[x]$ and it is without loss of
generality to assume that the degree of the polynomial $h(x,y)$ with respect
to $y$ is less than the degree of the polynomial $r(x,y)$ with respect to
$y$. There exists a finite number of local elementary exponential invariants
\begin{equation} \label{Lienard_exp_inv1_loc1}
\begin{gathered}
E_j(x,y)=\exp\left[\frac{u_j(x)}{\{y-Y_{j,\infty}(x)\}^{n_j}}\right],\quad
 u_j(x),\, Y_{j,\infty}(x)\in \mathbb{C}_{\infty}\{x\},\\
  n_j\in\mathbb{N},\quad j=1,\ldots, K,\quad K\in\mathbb{N}
\end{gathered}
\end{equation}
such that the exponential invariant $E(x,y)$ equals the product
$E_1^{}(x,y)\times \ldots \times E_K^{}(x,y)$. It follows from Theorem
\ref{T:inv_curve_prim} and Lemma \ref{L:exp_factor_inv_curve} that each
series $Y_{j,\infty}(x)$ satisfies equation \eqref{Lienard_y_x}. In fact, the
series $Y_{j,\infty}(x)$  coincides with one of the series
$y^{(1,2)}_{j,\infty}(x)$
 presented in Theorem~\ref{T1:Lienard_2m+1}.  Let
us denote the cofactor of the local elementary exponential invariant
$E_j(x,y)$ by $\varrho_j(x,y)\in\mathbb{C}_{\infty}\{x\}[y]$. We see that the
following expression $\varrho_1(x,y)+\ldots+\varrho_K(x,y)$ equals the
cofactor $\varrho(x,y)$ of the invariant $E(x,y)$. Recall that the cofactor
$\varrho(x,y)$ is an element of the ring~$\mathbb{C}[x,y]$.

Substituting the explicit representation of $E_j(x,y)$ into the partial
differential equation $\mathcal{X}E_j(x,y)=\varrho_j(x,y)E_j(x,y)$, we get
\begin{equation} \label{Lienard_exp_inv1_2n}
\begin{gathered}
yu_{j,x}=n_j\lambda_j(x,y)u_j(x)+\varrho_j(x,y)\left\{y-Y_{j,\infty}(x)\right\}^{n_j}.
\end{gathered}
\end{equation}
In this expression $\lambda_j(x,y)\in\mathbb{C}_{\infty}\{x\}[y]$  is the
cofactor of the local elementary invariant $F_j(x,y)=y-Y_{j,\infty}(x)$.
Using Theorem \ref{T:coff_local2}, we find the cofactor $\lambda_j(x,y)$. The
result is
\begin{equation} \label{Lienard_exp_inv1_cof_local}
\begin{gathered}
\lambda_j(x,y)=-f(x)-\{Y_{j,\infty}(x)\}_x.
\end{gathered}
\end{equation}
 Analyzing expression
\eqref{Lienard_exp_inv1_2n}, we see that $n_j=1$,
$\varrho_j(x,y)=u_{j,x}(x)$, and the series $u_j(x)$ satisfies the following
ordinary differential equation
\begin{equation} \label{Lienard_exp_inv1_3n}
\begin{gathered}
Y_{j,\infty}(x)u_{j,x}(x)+(f(x)+[Y_{j,\infty}(x)]_x)u_j(x)=0.
\end{gathered}
\end{equation}
Using the dominant behavior of the series $Y_{j,\infty}(x)$ given by the
monomial $b_{0,j}x^{(n+1)/2}$, we find the dominant monomial of the series
$u_{j}(x)$. Thus, we get the relation
$u_{j}(x)=\sigma_{j}x^{-(n+1)/2}+o(x^{-(n+1)/2})$, where
$\sigma_j\in\mathbb{C}$ and $x\rightarrow \infty$. Hence the cofactor
$\varrho_j(x,y)=u_{j,x}(x)$ does not have monomials with non-negative
exponents. Consequently, the polynomial $\displaystyle
\varrho(x,y)=\varrho_1(x,y)+\ldots+\varrho_K(x,y)$ should be identically
zero. We conclude that the argument $h(x,y)/r(x,y)$ of the exponential
invariant $E(x,y)=\exp\left\{h(x,y)/r(x,y)\right\}$ is a rational first
integral of the differential system in question. It is a contradiction.

\end{proof}
\textit{Remark.} It will be shown below that Li\'{e}nard differential systems
\eqref{Lienard_gen} satisfying the condition $\deg g>2\deg f+1$ do not have
rational first integrals. Thus, this condition can be removed from the
statement of Lemma \ref{L:Lienard_exp_inv3}.

Now our aim is to prove  that Li\'{e}nard differential systems
\eqref{Lienard_gen} do not have Darboux first integrals provided that $\deg
g>2\deg f+1$.

\begin{theorem}\label{T:Lienard_Integrability1_C}
Li\'{e}nard differential systems  \eqref{Lienard_gen} from family ($C$) are
not Darboux integrable.
\end{theorem}
\begin{proof}
Let us suppose that a Li\'{e}nard differential system~\eqref{Lienard_gen}
with $\deg g>2\deg f+1$ has a Darboux first integral. In view of Theorem
\ref{T:Darboux_rat} the system possesses a rational integrating factor.
Consequently, there exists $K\in\mathbb{N}$ pairwise distinct irreducible
invariant algebraic curves $F_1(x,y)=0$, $\ldots$, $F_K(x,y)=0$ and $K$
non-zero integer numbers $d_1$, $\ldots$, $d_K$ such that the following
condition
\begin{equation}
 \label{Lienard_Inegrability_new1}
\sum_{j=1}^{K}d_j\lambda_j(x,y)=f(x),\quad
 d_1,\ldots, d_{K}\in\mathbb{Z}
\end{equation}
is valid. In this expression $\lambda_j(x,y)$ is the cofactor of the
 invariant algebraic curve $F_j(x,y)=0$. Suppose the family of Puiseux series $y^{(l)}_{\infty}(x)$
 arises $N_{l,j}$ times in the factorization of the polynomial $F_j(x,y)$ in the ring
 $\mathbb{C}_{\infty}\{x\}[y]$.
 Here $l=1$, $2$ and we use the designations of  Theorem
\ref{T1:Lienard_2m+1}. The cofactor $\lambda_j(x,y)$ reads~as
\begin{equation}
 \label{Lienard_Inegrability_coff_n1}
 \begin{gathered}
\lambda_j(x,y)=-(N_{1,j}+N_{2,j})f(x)-\left\{N_{1,j}h^{(1)}_x(x)+N_{2,j}h^{(2)}_x(x)\right\}_+,\\
 N_{1,j},N_{2,j}\in\mathbb{N}_0,\quad N_{1,j}+N_{2,j}>0,
\end{gathered}
\end{equation}
where $h^{(l)}(x)$ is the initial part of Puiseux series
$y^{(l)}_{\infty}(x)$ introduced in Theorem~\ref{T1:Lienard_2m+1}.
Substituting expression \eqref{Lienard_Inegrability_coff_n1} into relation
\eqref{Lienard_Inegrability_new1} yields
\begin{equation}
 \label{Lienard_Inegrability_new2} \begin{gathered}
\sum_{j=1}^{K}d_jN_{1,j}\left(f(x)+\left\{h^{(1)}_x(x)\right\}_+\right)
+\sum_{j=1}^{K}d_jN_{2,j}\left(f(x)+\left\{h^{(2)}_x(x)\right\}_+\right)
=-f(x).
\end{gathered}
\end{equation}
 Let us suppose that $n$ is odd. The dominant behavior of the truncated series $h^{(1)}(x)$ and $h^{(2)}(x)$
is $b_0x^{(n+1)/2}$ and $-b_0x^{(n+1)/2}$, respectively.  Setting to zero the
coefficient of $x^{(n-1)/2}$ in expression \eqref{Lienard_Inegrability_new2},
we obtain
\begin{equation}
 \label{Lienard_Inegrability_new3} \begin{gathered}
\sum_{j=1}^{K}d_jN_{1,j}=\sum_{j=1}^{K}d_jN_{2,j}.
\end{gathered}
\end{equation}
If $n$ is even, then by Theorem \ref{T1:Lienard_2m+1} we get
$N_{1,j}=N_{2,j}$. Consequently, relation~\eqref{Lienard_Inegrability_new3}
is identically satisfied. As a result, condition
\eqref{Lienard_Inegrability_new2} takes the form
\begin{equation}
 \label{Lienard_Inegrability_new4} \begin{gathered}
B\left(2f(x)+\left\{h^{(1)}_x(x)+h^{(2)}_x(x)\right\}_+\right)
=-f(x),\quad B=\sum_{j=1}^{K}d_jN_{1,j}.
\end{gathered}
\end{equation}
Obviously $B$ is a non-zero integer number. Now let us consider the equation
\begin{equation}
 \label{Lienard_Inegrability_ODE_alt}
yy_x+\varepsilon f(x)y+g(x)=0.
\end{equation}
Puiseux series near the point $x=\infty$ that satisfy this equation coincide
with the Puiseux series $y_{j,\infty}^{(1,2)}(x)$ presented in Theorem
\ref{T1:Lienard_2m+1} provided that $\varepsilon=1$. Analogously to the case
$\varepsilon=1$, we find two families  of Puiseux series near the point
$x=\infty$ solving equation \eqref{Lienard_Inegrability_ODE_alt}. We denote
these families as $Y_{\infty}^{(1)}(x)$ and $Y_{\infty}^{(2)}(x)$. It is
straightforward to see that these series can be represented as
\begin{equation}
 \label{Lienard_Inegrability_PS_alt}
Y_{\infty}^{(1,2)}(x)=\sum_{k=0}^{\infty}\varepsilon^kv^{(1,2)}_k(x),
\end{equation}
where the coefficients $v^{(1,2)}_k(x)$ are Puiseux series near the point
$x=\infty$. Substituting expression \eqref{Lienard_Inegrability_PS_alt} into
equation \eqref{Lienard_Inegrability_ODE_alt} and setting to zero the
coefficients of different powers of $\varepsilon$, we find ordinary
differential equations for the series $v^{(1,2)}_k(x)$. Two first equations
take the form
\begin{equation}
 \label{Lienard_Inegrability_ODE_for_coeff}
v_0v_{0,x}+g(x)=0,\quad v_0v_{1,x}+v_{0,x}v_1+f(x)v_0=0,
\end{equation}
where the upper index is omitted. Thus, we get
\begin{equation}
 \label{Lienard_Inegrability_coeff_expl}
v^{(2)}_0(x)=-v^{(1)}_0(x),\quad v^{(1,2)}_1(x)=-\frac{2f_0}{2m+n+3}x^{m+1}+o(x^{m+1}),\, x\rightarrow\infty.
\end{equation}
Analyzing other ordinary differential equations for the series $v^{(1,2)}_k(x)$
with $k\geq 2$, we find $v^{(1,2)}_k(x)=o(x^{m+1})$,  $x\rightarrow\infty$,
$k\geq 2$. We recall that the Puiseux series $y_{j,\infty}^{(1,2)}(x)$ with
the same upper index have coinciding initial parts involving monomials with
exponents exceeding $-(n+1)/2$. These initial parts are given by $h^{(1)}(x)$
and $h^{(2)}(x)$. Substituting our results into condition
\eqref{Lienard_Inegrability_new4} and considering the coefficients of the
leading term $x^m$, we come to the equation
\begin{equation}
 \label{Lienard_Inegrability_new5} \begin{gathered}
B\left(2f_0-\frac{4(m+1)f_0}{2m+n+3}\right)=-f_0.
\end{gathered}
\end{equation}
Solving this equation, we find the expression
\begin{equation}
 \label{Lienard_Inegrability_new6} \begin{gathered}
B=-\frac{m+1}{n+1}-\frac12.
\end{gathered}
\end{equation}
Inequality $n>2m+1$ shows that $B$ is not an integer. It is a contradiction.
\end{proof}

\textit{Corollary 1.} A Li\'{e}nard differential system~\eqref{Lienard_gen}
from family ($C$) has at most one invariant algebraic curve $F(x,y)=0$ with
the property $N_1=N_2$.
\begin{proof} Recall that the variables $N_1$ and $N_2$ give the number of distinct Puiseux
series $y^{(1)}_{\infty}(x)$ and $y^{(2)}_{\infty}(x)$ from the field
$\mathbb{C}_{\infty}\{x\}$ arising in the factorization of the polynomial producing the
algebraic curve $F(x,y)=0$, respectively. For more details see
Theorem~\ref{T1:Lienard_2m+1}. Note that we do not require the polynomial
$F(x,y)$ to be irreducible. The cofactor of such an algebraic curve takes the
form
\begin{equation}
 \label{Lienard_Inegrability_coff_n1_add}
 \begin{gathered}
\lambda(x,y)=-2N_1f(x)-N_1\left\{h^{(1)}_x(x)+h^{(2)}_x(x)\right\}_+.
\end{gathered}
\end{equation}
If there exists another   invariant algebraic curve with the property
$N_1=N_2$, then the system under consideration possesses a rational first
integral. This fact contradicts Theorem \ref{T:Lienard_Integrability1_C}.
\end{proof}

\textit{Corollary 2.} A Li\'{e}nard differential system~\eqref{Lienard_gen}
from family ($C$) cannot have two distinct irreducible invariant algebraic
curves $F_1(x,y)=0$ and $F_2(x,y)=0$ with the property
$N_{1,1}N_{2,2}-N_{1,2}N_{2,1}=0$, where $N_{l,j}$ is the number of
    times the family of Puiseux series $y^{(l)}_{\infty}(x)$ enters the
    factorization of the polynomial $F_j(x,y)$ in the ring
    $\mathbb{C}_{\infty}\{x\}[y]$.

\begin{proof} The proof is by contradiction.
Suppose  a Li\'{e}nard differential system~\eqref{Lienard_gen} satisfying the
condition $\deg g>2\deg f+1$ possesses two distinct irreducible invariant
algebraic curves $F_1(x,y)=0$ and $F_2(x,y)=0$ with the property
$N_{1,1}N_{2,2}-N_{1,2}N_{2,1}=0$.

First of all, let us assume that one of the numbers $N_{l,j}$, where $l$, $j=
1$,~$2$, is zero. Without loss of generality, we choose $N_{2,1}=0$. This
gives $N_{1,1}N_{2,2}=0$. The polynomial $F_j(x,y)$, $j=1$, $2$ should have
at least one Puiseux series near the point $x=\infty$ in its factorization.
Thus, we get $N_{1,1}>0$, $N_{2,2}=0$, and $N_{1,2}>0$. By Theorem
\ref{T1:Lienard_2m+1} there exists at most one irreducible invariant
algebraic curve possessing only one family of Puiseux series near the point
$x=\infty$ in the factorization. This yields a contradiction.

Now suppose that all the numbers $N_{l,j}$, where $l$, $j= 1$,~$2$, are
non-zero. Introducing the variable
\begin{equation}
 \label{Lienard_Inegrability_coff_n1_number}
 \begin{gathered}
\varkappa=\frac{N_{1,1}}{N_{2,1}}=\frac{N_{1,2}}{N_{2,2}},
\end{gathered}
\end{equation}
we represent the cofactors $\lambda_1(x,y)$ and $\lambda_2(x,y)$ of the
invariant algebraic curves $F_1(x,y)=0$ and $F_2(x,y)=0$ in the form
\begin{equation}
 \label{Lienard_Inegrability_coff_n1_add_coff}
 \begin{gathered}
\lambda_j(x,y)=-N_{2,j}\left[(\varkappa+1) f(x)+\left\{\varkappa h^{(1)}_x(x)+
h^{(2)}_x(x)\right\}_+\right],\quad j=1,2.
\end{gathered}
\end{equation}
We conclude that the cofactors are dependent over the ring $\mathbb{Z}$.
Consequently, the Li\'{e}nard differential system under study has a rational
first integral. This is a contradiction.

\end{proof}

We see from Corollary $1$ and Theorem  \ref{T1:Lienard_2m+1} that if $\deg g
$ denoted as $n$ is an even number, then a Li\'{e}nard differential
system~\eqref{Lienard_gen} satisfying the condition $\deg g>2\deg f+1$ cannot
have more than one irreducible invariant algebraic curve simultaneously.

Next, let us investigate the existence of  non-autonomous Darboux first
integrals. The following lemma is valid.

\begin{lemma}\label{L:Lienard_Integrability_time2}
A Li\'{e}nard differential system \eqref{Lienard_gen} from family ($C$) has a
non-autonomous Darboux first integral  with a time-dependent
 exponential factor~\eqref{FI_t_gen} if and only if   $\deg f=0$
 and one of the following assertions is valid.

\begin{enumerate}

\item There exists an irreducible invariant algebraic curve $F(x,y)=0$
    such that the family of Puiseux series $y^{(1)}_{\infty}(x)$ arises
    in the factorization of the polynomial $F(x,y)$ in the ring
 $\mathbb{C}_{\infty}\{x\}[y]$ as many times as so does the family
 $y^{(2)}_{\infty}(x)$, i.e. $N_{1}=N_{2}$. A first integral takes the
 form
\begin{equation}
 \label{Lienard_Inegrability_FI_C1_deg_f_0}
 \begin{gathered}
I(x,y,t)=F(x,y)\exp\left[\frac{2(n+1)f_0N_1 t}{n+3}\right].
\end{gathered}
\end{equation}

\item There exist two distinct irreducible invariant algebraic curves
    $F_1(x,y)=0$ and $F_2(x,y)=0$ such that the following relation
$N_{1,j}\neq N_{2,j}$, $j=1$, $2$ is valid, where  $N_{l,j}$ is the
number of times the family of Puiseux series
    $y^{(l)}_{\infty}(x)$ enters the factorization of the polynomial
    $F_j(x,y)$ in the ring $\mathbb{C}_{\infty}\{x\}[y]$. A first
    integral reads~as
 \begin{equation}
 \label{Lienard_Inegrability_FI_C2_deg_f_0}
 \begin{gathered}
I(x,y,t)=\frac{\left\{F_1(x,y)\right\}^{N_{1,2}-N_{2,2}}}{\left\{F_2(x,y)\right\}^{N_{1,1}-N_{2,1}}}
\exp\left[\frac{2(n+1)f_0\Omega t}{n+3}\right]
\end{gathered}
\end{equation}
with the parameter  $\Omega$ given by the relation $\Omega=
N_{1,2}N_{2,1}-N_{1,1}N_{2,2}$.

\end{enumerate}

There are no other independent non-autonomous Darboux first integrals with a
time-dependent exponential factor \eqref{FI_t_gen}.

\end{lemma}

\begin{proof}
By Lemmas \ref{L:Lienard_exp_inv1} and \ref{L:Lienard_exp_inv2}  exponential
factors cannot enter an explicit expression of  a non-autonomous Darboux
first integral \eqref{FI_t_gen}. It follows from Theorem
\ref{T:L23_Non_aut_FI} that a Li\'{e}nard differential system
\eqref{Lienard_gen} satisfying
 the condition $\,\,$ $\deg g>2\deg f+1$ has a first integral~\eqref{FI_t_gen} if and only if there exists $K\in\mathbb{N}$ pairwise distinct
irreducible invariant algebraic curves $F_1(x,y)=0$, $\ldots$, $F_K(x,y)=0$
and $K+1$ non-zero complex  numbers $d_1$, $\ldots$, $d_K$, $\omega$ such
that the following condition
\begin{equation}
 \label{Lienard_Inegrability_new1_time}
\sum_{j=1}^{K}d_j\lambda_j(x,y)+\omega=0,\quad
 d_1,\ldots, d_{K},\omega\in\mathbb{C}\setminus\{0\}
\end{equation}
is valid. In this expression $\lambda_j(x,y)$ is the cofactor of the
 invariant algebraic curve $F_j(x,y)=0$. The related first integral reads as
\begin{equation}
 \label{Lienard_Inegrability_new1_time_add}
I(x,y,t)=\prod_{j=1}^{K}F_j^{d_j}(x,y)\exp[\omega t].
\end{equation}
 Suppose the family of Puiseux series $y^{(l)}_{\infty}(x)$
 arises $N_{l,j}$ times in the factorization of the polynomial $F_j(x,y)$ in the ring
 $\mathbb{C}_{\infty}\{x\}[y]$.
 Here $l=1$, $2$ and again we use the designations of  Theorem
\ref{T1:Lienard_2m+1}. The cofactor $\lambda_j(x,y)$ is given by
relation~\eqref{Lienard_Inegrability_coff_n1}. The necessary and sufficient
condition for first integral \eqref{Lienard_Inegrability_new1_time_add} to
exist reads~as
\begin{equation}
 \label{Lienard_Inegrability_new2_time} \begin{gathered}
\sum_{j=1}^{K}d_jN_{1,j}\left(f(x)+\left\{h^{(1)}_x(x)\right\}_+\right)
+\sum_{j=1}^{K}d_jN_{2,j}\left(f(x)+\left\{h^{(2)}_x(x)\right\}_+\right)=\omega.
\end{gathered}
\end{equation}
This condition is not satisfied whenever relation
\eqref{Lienard_Inegrability_new3} is not valid.  Further, we rewrite
condition \eqref{Lienard_Inegrability_new2_time} in the form
\begin{equation}
 \label{Lienard_Inegrability_new4_time} \begin{gathered}
B\left(2f(x)+\left\{h^{(1)}_x(x)+h^{(2)}_x(x)\right\}_+\right)=\omega,\quad B=\sum_{j=1}^{K}d_jN_{1,j},
\end{gathered}
\end{equation}
where  unlike the case of Theorem \ref{T:Lienard_Integrability1_C} the
parameter $B$ may be complex-valued. If $m>0$, then we arrive at the
expression
\begin{equation}
 \label{Lienard_Inegrability_new5_time} \begin{gathered}
B\left(2f_0-\frac{4(m+1)f_0}{2m+n+3}\right)=0,
\end{gathered}
\end{equation}
which is not valid. Consequently, we should set $m=0$. Relations
\eqref{Lienard_Inegrability_new3} and \eqref{Lienard_Inegrability_new4_time}
now become
\begin{equation}
 \label{Lienard_Inegrability_new6_time} \begin{gathered}
\sum_{j=1}^{K}d_jN_{1,j}=\sum_{j=1}^{K}d_jN_{2,j},\quad
2f_0(n+1)\sum_{j=1}^{K}d_jN_{1,j}=(n+3)\omega.
\end{gathered}
\end{equation}
If the original Li\'{e}nard differential system has only one irreducible
invariant algebraic curve ($K=1$), then  algebraic system
\eqref{Lienard_Inegrability_new6_time} is satisfied if and only if the
following relation $N_{1,1}=N_{2,1}$ is valid. We note that the parameter
$d_1\neq0$ can be chosen arbitrarily. Thus, we set $d_1=1$. The second
equation in \eqref{Lienard_Inegrability_new6_time} produces the value of
$\omega$. As a result, we obtain non-autonomous Darboux first
integral~\eqref{Lienard_Inegrability_FI_C1_deg_f_0}, where the index $j=1$ is
omitted.

Now let us suppose that the Li\'{e}nard differential system under
consideration has at least two distinct irreducible invariant algebraic
curves $F_1(x,y)=0$ and $F_2(x,y)=0$ satisfying the restriction $N_{1,j}\neq
N_{2,j}$, $j=1$, $2$. We see that algebraic system
\eqref{Lienard_Inegrability_new6_time} is always satisfied. Indeed, setting
$K=2$ and recalling the fact that one of the exponents $d_1$ and $d_2$ can be
chosen arbitrary, we obtain a solution
\begin{equation}
 \label{Lienard_Inegrability_FI_C6_deg_f_0} \begin{gathered}
d_1=N_{1,2}-N_{2,2},\quad d_2=N_{2,1}-N_{1,1},\quad \omega=\frac{2(n+1)f_0\Omega }{n+3},
\end{gathered}
\end{equation}
where we use the designation $\Omega=
    N_{1,2}N_{2,1}-N_{1,1}N_{2,2}$. Hence we have found time-dependent Darboux first integral
\eqref{Lienard_Inegrability_FI_C2_deg_f_0}. By Corollary $2$ to
Theorem~\ref{T:Lienard_Integrability1_C}, the following relation
$\Omega\neq0$ is valid.

 Suppose a Li\'{e}nard differential system has two independent
non-autonomous Darboux first integrals  with a time-dependent
 exponential factor~\eqref{FI_t_gen}, then this system is Darboux integrable.
 This fact contradicts Theorem~\ref{T:Lienard_Integrability1_C}.

\end{proof}

\textit{Remark.} Li\'{e}nard differential systems \eqref{Lienard_gen}
satisfying the conditions of item~$2$ can only arise when $n=\deg g$ is an odd
number.

\smallskip

There exist Li\'{e}nard differential systems \eqref{Lienard_gen} with
non-autonomous Darboux first integrals
\eqref{Lienard_Inegrability_FI_C1_deg_f_0} and
\eqref{Lienard_Inegrability_FI_C2_deg_f_0}. An example was given in article
\cite{Demina16}. Finally, we turn to the Liouvillian integrability.

\begin{theorem}\label{T:Lienard_Integrability3}
A Li\'{e}nard differential system \eqref{Lienard_gen} from family ($C$) is
Liouvillian integrable if
 and only if the relation
\begin{equation}
 \label{Lienard_Inegrability_IF_C3_Cond}
 \begin{gathered}
4(m+1)f(x)+(2m+n+3)\left\{h^{(1)}_x(x)+h^{(2)}_x(x)\right\}_+=0
\end{gathered}
\end{equation}
is identically satisfied and one of the following assertions is valid.

\begin{enumerate}

\item There exists an irreducible invariant algebraic curve $F(x,y)=0$
    such that the family of Puiseux series $y^{(1)}_{\infty}(x)$ arises
    in the factorization of the polynomial $F(x,y)$ in the ring
 $\mathbb{C}_{\infty}\{x\}[y]$ as many times as so does the family
 $y^{(2)}_{\infty}(x)$, i.e. $N_{1}=N_{2}$. In this case the system has
 the unique Darboux integrating factor
\begin{equation}
 \label{Lienard_Inegrability_IF_C1_deg_f_0}
 \begin{gathered}
M(x,y)=\left\{F(x,y)\right\}^{-\frac{2m+n+3}{2(n+1)N_{1}}}.
\end{gathered}
\end{equation}

\item There exist two distinct irreducible invariant algebraic curves
    $F_1(x,y)=0$ and $F_2(x,y)=0$ such that the following relation
$N_{1,j}\neq N_{2,j}$, $j=1$, $2$ is valid, where  $N_{l,j}$ is the
number of times the family of Puiseux series
    $y^{(l)}_{\infty}(x)$ enters the factorization of the polynomial
    $F_j(x,y)$ in the ring $\mathbb{C}_{\infty}\{x\}[y]$. In this case
    the system has the unique Darboux
 integrating factor
 \begin{equation}
 \label{Lienard_Inegrability_IF_C2_deg_f_0}
 \begin{gathered}
M(x,y)=\frac{\left\{F_1(x,y)\right\}^{\frac{(2m+n+3)(N_{2,2}-N_{1,2})}{2(n+1)\Omega}}}
{\left\{F_2(x,y)\right\}^{\frac{(2m+n+3)(N_{2,1}-N_{1,1})}{2(n+1)\Omega}}}.
\end{gathered}
\end{equation}
 The following designation $\Omega= N_{1,2}N_{2,1}-N_{1,1}N_{2,2}$ is
    introduced in expression~\eqref{Lienard_Inegrability_IF_C2_deg_f_0}.

\end{enumerate}
\end{theorem}

\begin{proof}
It follows from Theorem \ref{T:Liouville} that a Liouvillian integrable
 differential system \eqref{DS} has a Darboux integrating factor.
 By Lemmas \ref{L:Lienard_exp_inv1} and \ref{L:Lienard_exp_inv2}  exponential
invariants cannot enter an explicit expression of  a  Darboux integrating
factor. Thus, the Darboux integrating factor reads as
\begin{equation}
 \label{Lienard_Inegrability_IF_Darboux1} \begin{gathered}
M(x,y)=\prod_{j=1}^KF_j^{d_j}(x,y),\quad d_1,\ldots,d_k\in\mathbb{C},\quad K\in\mathbb{N},
\end{gathered}
\end{equation}
where the polynomials $F_1(x,y)$, $\ldots$, $F_K(x,y)$ give pairwise distinct
irreducible invariant algebraic curves $F_1(x,y)=0$, $\ldots$, $F_K(x,y)=0$
of a Li\'{e}nard differential system. Without loss of generality, we suppose
that the numbers $d_1$, $\ldots$, $d_K$ are all non-zero. The cofactor
$\lambda_j(x,y)$ of the invariant algebraic curve $F_j(x,y)=0$ is given by
relation \eqref{Lienard_Inegrability_coff_n1}. The necessary and sufficient
condition $d_1\lambda_1(x,y)+\ldots
d_K\lambda_K(x,y)=-\text{div}\,\mathcal{X}$ for Darboux integrating
factor~\eqref{Lienard_Inegrability_IF_Darboux1} to exist now takes the form
\begin{equation}
 \label{Lienard_Inegrability_IF_C3} \begin{gathered}
\sum_{j=1}^{K}d_jN_{1,j}\left(f(x)+\left\{h^{(1)}_x(x)\right\}_+\right)
+\sum_{j=1}^{K}d_jN_{2,j}\left(f(x)+\left\{h^{(2)}_x(x)\right\}_+\right)
=-f(x).
\end{gathered}
\end{equation}
This condition is not satisfied provided that relation
\eqref{Lienard_Inegrability_new3} is not valid. Using relation
\eqref{Lienard_Inegrability_new3}, we simplify condition
\eqref{Lienard_Inegrability_IF_C3}. Thus, we get
\begin{equation}
 \label{Lienard_Inegrability_IF_C4} \begin{gathered}
B\left(2f(x)+\left\{h^{(1)}_x(x)+h^{(2)}_x(x)\right\}_+\right)=-f(x),\quad B=\sum_{j=1}^{K}d_jN_{1,j},
\end{gathered}
\end{equation}
where  unlike the case of Theorem \ref{T:Lienard_Integrability1_C} the
parameter $B$ may be complex-valued. Doing the same as in the proof of
Theorem \ref{T:Lienard_Integrability1_C}, we find the following equality
\begin{equation}
 \label{Lienard_Inegrability_IF_C5} \begin{gathered}
B\left(2f_0-\frac{4(m+1)f_0}{2m+n+3}\right)=-f_0,
\end{gathered}
\end{equation}
which gives the value of $B$. The result is
\begin{equation}
 \label{Lienard_Inegrability_IF_C3_n1} \begin{gathered}
B=-\frac{2m+n+3}{2(n+1)}.
\end{gathered}
\end{equation}
Substituting this equality into condition \eqref{Lienard_Inegrability_IF_C4}
yields relation  \eqref{Lienard_Inegrability_IF_C3_Cond}. Finally, we are
left with the following algebraic system
\begin{equation}
 \label{Lienard_Inegrability_IF_C3_n2} \begin{gathered}
\sum_{j=1}^{K}d_jN_{1,j}=\sum_{j=1}^{K}d_jN_{2,j},\quad \sum_{j=1}^{K}d_jN_{1,j}=-\frac{2m+n+3}{2(n+1)}
\end{gathered}
\end{equation}
with respect to the unknowns $d_1$, $\ldots$, $d_K$.

Suppose the Li\'{e}nard differential system under study has only one
irreducible invariant algebraic curve ($K=1$). Algebraic system
\eqref{Lienard_Inegrability_IF_C3_n2} is satisfied if and only if
$N_{1,1}=N_{2,1}$. Omitting the index $j$, we find the value of $d$ and
Darboux integrating factor \eqref{Lienard_Inegrability_IF_C1_deg_f_0}.

Now we assume that  the Li\'{e}nard differential system in question possesses
at least two distinct irreducible invariant algebraic curves. Setting $K=2$,
we see that the determinant of algebraic system
\eqref{Lienard_Inegrability_IF_C3_n2} equals $\Omega=
N_{1,2}N_{2,1}-N_{1,1}N_{2,2}$. By Corollary $2$ to Theorem
\ref{T:Lienard_Integrability1_C} we get $\Omega\neq0$. Consequently,
algebraic system \eqref{Lienard_Inegrability_IF_C3_n2} has the unique
solution. As a result we obtain Darboux integrating
factor~\eqref{Lienard_Inegrability_IF_C2_deg_f_0}.

If there are no invariant algebraic curves or there exists the unique
irreducible invariant algebraic curve satisfying the condition $N_1\neq N_2$,
then system \eqref{Lienard_Inegrability_IF_C3_n2} is inconsistent.

Suppose  the Li\'{e}nard differential system has two distinct Darboux
integrating factors. Then their ratio is a Darboux first integral. This fact
contradicts Theorem~\ref{T:Lienard_Integrability1_C}.

\end{proof}

\textit{Corollary.}  A Liouvillian integrable Li\'{e}nard differential system
\eqref{Lienard_gen} from family ($C$) has at most two distinct irreducible
invariant algebraic curves simultaneously provided that $n=\deg g$ is an odd
number.

\begin{proof}
Assuming that a Liouvillian integrable Li\'{e}nard differential system  has
three or more pairwise distinct irreducible invariant algebraic curves, we
use Theorem \ref{T:Lienard_Integrability3} to find at least two distinct
Darboux integrating factors. Consequently, the system possesses a Darboux
first integral.  It is a contradiction.

\end{proof}

\textit{Remark 1.} Item~$2$ can only arise if the number
$n=\deg g$ is odd. Moreover, integrating factor
\eqref{Lienard_Inegrability_IF_C2_deg_f_0} transforms into integrating
factor~\eqref{Lienard_Inegrability_IF_C1_deg_f_0} whenever the following
condition $N_{1,1}+N_{1,2}=N_{2,1}+N_{2,2}$ holds. In this case the
polynomial $F(x,y)$ in~\eqref{Lienard_Inegrability_IF_C1_deg_f_0} is reducible: $ F(x,y)=F_1(x,y)F_2(x,y)$.

\textit{Remark 2.} Relation \eqref{Lienard_Inegrability_IF_C3_Cond} is
identically satisfied whenever $m=0$ ($\deg f=0$). This statement follows from relations \eqref{Lienard_Inegrability_PS_alt} and \eqref{Lienard_Inegrability_ODE_for_coeff}.

Let us obtain all Liouvillian integrable Li\'{e}nard differential systems
\eqref{Lienard_gen} from family ($C$) with a hyperelliptic invariant
algebraic curve $y^2+u(x)y+v(x)=0$, where $u(x)$, $v(x)\in\mathbb{C}[x]$.

\begin{theorem}\label{T:Lienard_IntegrabilityC_partial}
A Li\'{e}nard differential system  \eqref{Lienard_gen} from family ($C$) with
a hyperelliptic invariant algebraic curve $y^2+u(x)y+v(x)=0$, where $u(x)$,
$v(x)\in\mathbb{C}[x]$,
 is Liouvillian
integrable if and only the system is of the form
\begin{equation}
 \label{Lienard_Inegrability_Cpartial_1}
x_t=y,\quad y_t=-\frac{(k+2l)}{4}w^{l-1}w_xy-\frac{k}8\left(w^{2l-1}+4\beta w^{k-1}\right)w_x,
\end{equation}
where $\beta\in\mathbb{C}\setminus\{0\}$, $w(x)$ is a polynomial of degree
$(m+1)/l$, $k$ and $l$ are relatively prime natural numbers such that the
following relation $(m+1)k=(n+1)l$
 is valid. The associated  Li\'{e}nard differential system has the unique
Darboux integrating factor
   \begin{equation}
 \label{Lienard_Inegrability_Int_Fact_C_p1}
M(x,y)=\left\{y^2+w^ly+\frac14w^{2l}+\beta w^k\right\}^{-\left(\frac12+\frac{l}{k}\right)}
\end{equation}
and the hyperelliptic invariant algebraic curve reads as
$4y^2+4w^ly+w^{2l}+4\beta w^k=0$. A Liouvillian first integral is of the form
 \begin{equation}
 \begin{gathered}
 \label{Lienard_Inegrability_FI_C_p1}
I(x,y)=\frac{(2l-k)(2y+w^l)}{4kw^{\frac{k}2}\beta^{\frac12+\frac{l}{k}}}
{}_2F_1\left(\frac12,\frac12+\frac{l}{k};\frac32;-\frac{(2y+w^l)^2}{4\beta w^k}\right)
+\left\{y^2+w^ly+\frac14w^{2l}+\beta w^k\right\}^{\frac12-\frac{l}{k}},
\end{gathered}
\end{equation}
where ${}_2F_1(\alpha,\delta;\sigma;s)$ is the  hypergeometric function.
\end{theorem}

\begin{proof}
It is straightforward to verify that system
\eqref{Lienard_Inegrability_Cpartial_1} is Liouvillian integrable with an
integrating factor and a first integral given by expressions
\eqref{Lienard_Inegrability_Int_Fact_C_p1} and
\eqref{Lienard_Inegrability_FI_C_p1}, respectively. Since $\beta\neq0$ and
$w(x)$ is a polynomial of degree $(m+1)/l$, we conclude that system
\eqref{Lienard_Inegrability_Cpartial_1} is from family ($C$).

Now our goal is to prove the converse statement. Let us suppose that a
Li\'{e}nard differential system  \eqref{Lienard_gen} from family ($C$) is
Liouvillian integrable and possesses a hyperelliptic invariant algebraic
curve $y^2+u(x)y+v(x)=0$.  By Theorem \ref{T:Lienard_Integrability3}
condition \eqref{Lienard_Inegrability_IF_C3_Cond} is identically satisfied
and the system has integrating factor given by
expression~\eqref{Lienard_Inegrability_IF_C1_deg_f_0}, where  the polynomial
$F(x,y)$ can be chosen in the form $F(x,y)=y^2+u(x)y+v(x)$. In addition, we
set $N_1=1$. Substituting integrating factor
\eqref{Lienard_Inegrability_IF_C1_deg_f_0} into the partial differential
equation $yM_x-[f(x)y+g(x)]M_y-f(x)M=0$ and equating to zero the coefficients
of different powers of $y$ yields the relations
\begin{equation}
 \label{Lienard_Inegrability_IF_C3_n3} \begin{gathered}
f(x)=\frac{2m+n+3}{4(m+1)}u_x,\quad g(x)=\frac12 v_x+\frac{n-2m-1}{8(m+1)}uu_x
\end{gathered}
\end{equation}
and the following equation
\begin{equation}
 \label{Lienard_Inegrability_IF_C3_n4} \begin{gathered}
uv_x-\frac{n+1}{m+1}u_xv+\frac{n-2m-1}{4(m+1)}u^2u_x=0.
\end{gathered}
\end{equation}
Let us note that condition \eqref{Lienard_Inegrability_IF_C3_Cond} produces
an explicit expression of the polynomial $f(x)$ similar to that given in
relations \eqref{Lienard_Inegrability_IF_C3_n3}. Integrating equation
\eqref{Lienard_Inegrability_IF_C3_n4} with respect to the function $v(x)$, we
obtain
\begin{equation}
 \label{Lienard_Inegrability_IF_C3_n5} \begin{gathered}
v(x)=\beta u^{\frac{n+1}{m+1}}+\frac14u^2,
\end{gathered}
\end{equation}
where $\beta\in\mathbb{C}$ is a constant of integration. Using Theorem
\ref{T1:Lienard_2m+1} and the arguments given in the proof of Theorem
\ref{T:Lienard_Integrability1_C}, we conclude that $u(x)$ is a polynomial of
degree $m+1$ and $v(x)$ is a polynomial of degree $n+1$. Thus, we see that
$\beta$ is non-zero. Further, we introduce relatively prime natural numbers
$k$ and $l$ satisfying the relation $(m+1)k=(n+1)l$. It follows from expression
\eqref{Lienard_Inegrability_IF_C3_n5} that there exists a polynomial $w(x)$ of degree $(m+1)/l$ such that the polynomial $u(x)$
 can be represented in the form $u(x)=w^l(x)$.
 Hence we obtain the equality
 $v(x)=\beta w^k(x)+w^{2l}(x)$. Substituting the explicit representations of
 the polynomials $u(x)$ and $v(x)$ into relations
 \eqref{Lienard_Inegrability_IF_C3_n3}, we find the polynomials $f(x)$ and $g(x)$
 as given in  \eqref{Lienard_Inegrability_Cpartial_1}. Expressing
 the number $n$ from the relation $(m+1)k=(n+1)l$, we find Darboux integrating
 factor \eqref{Lienard_Inegrability_Int_Fact_C_p1} giving Liouvillian first
 integral \eqref{Lienard_Inegrability_FI_C_p1}.
\end{proof}

\textit{Remark 1.} We do not require that the polynomial $y^2+u(x)y+v(x)$ is
irreducible. See also Remark 1 to Theorem \ref{T:Lienard_Integrability3}.

\textit{Remark 2.}  The family of systems \eqref{Lienard_Inegrability_Cpartial_1} can be transformed to the following simple form
\begin{equation}
 \label{Lienard_Inegrability_Cpartial_1_Sundman}
s_{\tau}=z,\quad z_{\tau}=-\frac{(k+2l)}{4}s^{l-1}z-\frac{k}8\left(s^{2l-1}+4\beta s^{k-1}\right)
\end{equation}
via the generalized Sundman transformation $s(\tau)=w(x)$, $z(\tau)=y$, $d\tau=w_x(x)dt$. Substituting $w(x)=s$, $y=z$ into \eqref{Lienard_Inegrability_FI_C_p1}, we find a Liouvillian first integral for systems \eqref{Lienard_Inegrability_Cpartial_1_Sundman}.

\smallskip

It follows from Theorem \ref{T1:Lienard_2m+1} that equation
\eqref{Lienard_y_x} related to a Li\'{e}nard differential
system~\eqref{Lienard_gen} from family ($C$) may have a polynomial solution
only if $n=\deg g(x)$ is an odd number. Such a polynomial solution gives rise
to an invariant algebraic curve with the   generating polynomial of the first degree with respect to $y$. Let us study the Liouvillian integrability of Li\'{e}nard
differential systems \eqref{Lienard_gen} from family ($C$) possessing
invariant algebraic curves with generating polynomials of the first degree with
respect to $y$. Since arbitrary coefficients arise in the non-polynomial part
of the series $y_{\infty}^{(l)}(x)$, $l=1$, $2$, we conclude that equation
\eqref{Lienard_y_x} has at most two distinct polynomial solutions
simultaneously provided that the  inequality $\deg g> 2\deg f +1$ holds. In
what follows we denote these polynomial solutions as $y=p_1(x)$ and
$y=p_2(x)$. Note that the following relations $p_l(x)=\{h^{(l)}(x)\}_+$,
$l=1$, $2$ are valid, where $h^{(l)}(x)$ is the initial part of the series
$y_{\infty}^{(l)}(x)$.

\begin{theorem}\label{T:Lienard_IntegrabilityC_polynomial}
A Li\'{e}nard differential system  \eqref{Lienard_gen} from family ($C$) with
two distinct invariant algebraic curves given by first-degree polynomials
with respect to $y$  is Liouvillian integrable if and only if $n=\deg g(x)$
is an odd number, the system is of the
form~\eqref{Lienard_Inegrability_Cpartial_1} and other conditions of Theorem
\ref{T:Lienard_IntegrabilityC_partial} are satisfied with the additional
restriction: either $k$ is an even number or otherwise $(m+1)/l$ is an even
number and the polynomial $w(x)$ has only double roots. The polynomials
$p_1(x) $ and $p_2(x)$ producing the invariant algebraic curves  $y-p_1(x)=0$
and $y-p_2(x)=0$ can be represented in the form
\begin{equation}
 \label{Lienard_Inegrability_Cpartial_polynomials1}
p_1(x)=\sqrt{\beta} w^{\frac{k}{2}}+\frac{1}{2}w^l(x),\quad
p_2(x)=-\sqrt{\beta} w^{\frac{k}{2}}+\frac{1}{2}w^l(x),\quad \beta\in\mathbb{C}\setminus\{0\}.
\end{equation}
\end{theorem}

\begin{proof}

We use  item $2$ of Theorem
\ref{T:Lienard_Integrability3} and the arguments given in the proof of Theorem
\ref{T:Lienard_IntegrabilityC_partial}. Let us note that the hyperelliptic invariant
algebraic curve  of Theorem \ref{T:Lienard_IntegrabilityC_partial} with the
generating polynomial $y^2+w^ly+\frac14w^{2l}+\beta w^k=(y+w^l/2)^2+\beta
w^k$ splits into two distinct invariant algebraic curves $y-p_1(x)=0$ and
$y-p_2(x)=0$ if and only if $n$ is an odd number and either $k$ is an even
number or otherwise $(m+1)/l$ is an even number and $w(x)$ is a polynomial
with double roots. In addition, recall that the degree of the polynomial
$w(x)$ equals $(m+1)/l$.

\end{proof}

\textit{Remark.} This theorem can  also be proved directly without using
Theorem \ref{T:Lienard_IntegrabilityC_partial}. As an example, see Theorem
\ref{T:Lienard_IntegrabilityA_partial}.

Further, our goal is to demonstrate that there exist Liouvillian integrable
Li\'{e}nard differential systems \eqref{Lienard_gen} from family ($C$) for
any choice of the numbers $m=\deg f(x)$ and $n=\deg g(x)$. Setting
$u(x)=x^{m+1}$ in expression \eqref{Lienard_Inegrability_IF_C3_n5}, we find
the following Liouvillian integrable Li\'{e}nard differential systems from
family ($C$)
\begin{equation}
 \label{Lienard_Inegrability_IF_C3_n6} \begin{gathered}
x_t=y,\quad y_t=-\frac{2m+n+3}{4}x^my-\frac{n+1}{8}\left(4\beta x^n+x^{2m+1}\right).
\end{gathered}
\end{equation}
The related Darboux integrating factor reads as
\begin{equation}
 \label{Lienard_Inegrability_IF_C3_n7} \begin{gathered}
M(x,y)=\left(y^2+x^{m+1}y+\beta x^{n+1}+\frac14x^{2(m+1)}\right)^{-\frac{2m+n+3}{2(n+1)}}.
\end{gathered}
\end{equation}
The numbers $n=\deg g$ and $m=\deg f$ can be chosen arbitrarily. In addition,
if the following conditions $n=l(m+1)-1$,  $l\in\mathbb{N}$, and $l>2$ hold,
then we obtain another particular family of Liouvillian integrable
 Li\'{e}nard differential systems
\begin{equation}
 \label{Lienard_Inegrability_IF_C3_n8} \begin{gathered}
x_t=y,\quad y_t=-\frac{l+2}{4}u_xy-\frac{l}{8}\left(4\beta u^{l-1}+u\right)u_x,
\end{gathered}
\end{equation}
where $u(x)$ is an arbitrary polynomial of degree $m+1$. The associated
Darboux integrating factor can be represented in the form
\begin{equation}
 \label{Lienard_Inegrability_IF_C3_n9} \begin{gathered}
M(x,y)=\left(y^2+uy+\beta u^{l}+\frac14u^2\right)^{-\frac{l+2}{2l}}.
\end{gathered}
\end{equation}

Now let us study the the existence of non-autonomous Darboux--Jacobi last
 multipliers. The case $\deg f=0$ is simple. There are families of distinct Jacobi last
 multipliers arising as products of integrating factors  \eqref{Lienard_Inegrability_IF_C1_deg_f_0},
 \eqref{Lienard_Inegrability_IF_C2_deg_f_0} and non-autonomous
 first integrals $I^{\varkappa}(x,y,t)$, where $\varkappa\in\mathbb{C}$ and the function $I(x,y,t)$ is given by relations \eqref{Lienard_Inegrability_FI_C1_deg_f_0}
 and \eqref{Lienard_Inegrability_FI_C2_deg_f_0}.

\begin{lemma}\label{L:Lienard_Integrability_t3}

A Li\'{e}nard differential system \eqref{Lienard_gen} satisfying
 the conditions $\deg g>2\deg f+1$ and $\deg f>0$ has a non-autonomous Darboux--Jacobi last
 multiplier of the form \eqref{JLM_gen} if
 and only if there exists a non-zero complex number $\omega$ such that the relation
\begin{equation}
 \label{Lienard_Inegrability_IF_C3_Cond_time}
 \begin{gathered}
4(m+1)f(x)+(2m+n+3)\left\{h^{(1)}_x(x)+h^{(2)}_x(x)\right\}_++2(n+1)\omega=0
\end{gathered}
\end{equation}
is identically satisfied and one of the following assertions is valid.

\begin{enumerate}

\item There exists an irreducible invariant algebraic curve $F(x,y)=0$
    such that the family of Puiseux series $y^{(1)}_{\infty}(x)$ arises
    in the factorization of the polynomial $F(x,y)$ in the ring
 $\mathbb{C}_{\infty}\{x\}[y]$ as many times as so does the family
 $y^{(2)}_{\infty}(x)$, i.e. $N_{1}=N_{2}$. In this case the system has
 the unique Darboux--Jacobi last multiplier
\begin{equation}
 \label{Lienard_Inegrability_IF_C1_deg_f_0_time}
 \begin{gathered}
M(x,y,t)=\left\{F(x,y)\right\}^{-\frac{2m+n+3}{2(n+1)N_{1}}}\exp[\omega t].
\end{gathered}
\end{equation}

\item There exist two distinct irreducible invariant algebraic curves
    $F_1(x,y)=0$ and $F_2(x,y)=0$ such that the following relation
$N_{1,j}\neq N_{2,j}$, $j=1$, $2$ is valid, where  $N_{l,j}$ is the
number of times the family of Puiseux series $y^{(l)}_{\infty}(x)$ enters
    the factorization of the polynomial $F_j(x,y)$ in the ring
    $\mathbb{C}_{\infty}\{x\}[y]$. In this case the system has the unique
    Darboux--Jacobi last multiplier
 \begin{equation}
 \label{Lienard_Inegrability_IF_C2_deg_f_0_time}
 \begin{gathered}
M(x,y,t)=\frac{\left\{F_1(x,y)\right\}^{\frac{(2m+n+3)(N_{2,2}-N_{1,2})}{2(n+1)\Omega}}}{
\left\{F_2(x,y)\right\}^{\frac{(2m+n+3)(N_{2,1}-N_{1,1})}{2(n+1)\Omega}}}\exp[\omega t],
\end{gathered}
\end{equation}
where the parameter $\Omega$ is given by the relation $\Omega=
    N_{1,2}N_{2,1}-N_{1,1}N_{2,2}$.

\end{enumerate}

\end{lemma}

\begin{proof}
We repeat the proof of Theorem \ref{T:Lienard_Integrability3}. The only
difference is in condition~\eqref{Lienard_Inegrability_IF_C3}. In the
non-autonomous case this condition takes the form
\begin{equation}
 \label{Lienard_Inegrability_IF_C3_time} \begin{gathered}
\sum_{j=1}^{K}d_jN_{1,j}\left(f(x)+\left\{h^{(1)}_x(x)\right\}_+\right)
+\sum_{j=1}^{K}d_jN_{2,j}\left(f(x)+\left\{h^{(2)}_x(x)\right\}_+\right)
=\omega-f(x),
\end{gathered}
\end{equation}
where $\omega$ is a non-zero complex constant.
\end{proof}

We have established that Li\'{e}nard differential systems \eqref{Lienard_gen}
from family ($C$) have neither rational nor Darboux first integrals. Let us
note that the famous Duffing oscillators belong to family ($C$). These
oscillators are studied in articles \cite{Demina07, Demina16} in details.

\section{Quartic Li\'{e}nard differential equations with a quadratic damping function}\label{S:Example_L24}

The aim of the present section is to demonstrate that the necessary and
sufficient conditions of Liouvillian integrability presented in the previous sections
 can be used to find all Liouvillian integrable
subfamilies of Li\'{e}nard differential systems without performing the
classification of irreducible invariant algebraic curves. As an example, we
 consider Li\'{e}nard differential systems with the restrictions $\deg f= 2$
and $\deg g =4$:
\begin{equation}
\begin{gathered}
 \label{Lienard1_DS24_main}
x_t=y,\quad y_t=-(\zeta x^2+\beta x+\alpha)y-(\varepsilon x^4+\xi x^3+ex^2+\sigma x+\delta),\quad  \zeta\varepsilon\neq0.
\end{gathered}
\end{equation}
Introducing suitable rescalings and shifts, it is without loss of generality
to set $\zeta=3$, $\varepsilon=-3$, and $\beta=0$. In what follows we  work
with the following systems
\begin{equation}
\begin{gathered}
 \label{Lienard1_DS24}
x_t=y,\quad y_t=-(3 x^2+\alpha)y+3 x^4-\xi x^3-ex^2-\sigma x-\delta,\quad  \zeta\varepsilon\neq0.
\end{gathered}
\end{equation}
Let us solve the
integrability problem for systems  \eqref{Lienard1_DS24}.

\begin{theorem}\label{T:L24_untegrability}
Quartic Li\'{e}nard differential systems with a quadratic damping function
\eqref{Lienard1_DS24} are Liouvillian integrable if and only the tuple of
the parameters ($\alpha, \xi, \delta, \sigma, e $) equals
\begin{equation}
 \label{L24_Integrability_main}
\begin{gathered}
I:\quad(\alpha, \xi, \delta, \sigma, e )=\left(-\frac{25}{12}, -7, \frac{125}{432}, -\frac{25}{36},-5\right);\\
II:\quad(\alpha, \xi, \delta, \sigma, e )=\left( -\frac{61}{12}, -7, \frac{3905}{432},  \frac{623}{36},4\right).\hfill
\end{gathered}
\end{equation}
 The
related Darboux integrating factors can be represented as
\begin{equation}
 \label{L24_Integrability_Integrating_factors}
\begin{gathered}
I:\quad M(x,y)= \frac{1}{\left( y+{x}^{3}+\frac32{x}^{2}+\frac {5}{12}x-{\frac{25}{216}}
 \right)^{\frac23} \left( y-{x}^{2}-\frac53x-{\frac{25}{36}} \right)};\\
II:\quad M(x,y)=\frac{\left(y+x^3 + \frac32x^2 - \frac{31}{12}x  - \frac{781}{216}\right)^{\frac13}}{{y}^{2}
+{\frac { \left( 6x+5 \right)  \left( 6x-13 \right)
 \left( 6x+11 \right) y}{216}}-{\frac { \left( 6x-13 \right)
 \left( 6x+11 \right) ^{2} \left( 6x+5 \right) ^{2}}{7776}}
}.
\end{gathered}
\end{equation}
\end{theorem}

\begin{proof}
Our proof is based on the  results of Theorems \ref{T1:Lienard_gen} and
\ref{T:Lienard_Integrability2}. The Puiseux series given in relation
\eqref{Lienard1_F_series} are now the following
\begin{equation}
\begin{gathered}
 \label{Lienard24_PS1}
y^{(1)}_{\infty}(x)=-x^3-\frac{3}{2}x^2+\left(\xi-\alpha+\frac{9}{2}\right)x+b_3+\sum_{l=1}^{\infty}b_{l+3}x^{-l};\hfill\\
y^{(2)}_{\infty}(x)=x^2-\frac{1}{3}(\xi+2)x+\frac13\left(\xi+2-\alpha-e\right)+\sum_{l=1}^{\infty}a_{l+2}x^{-l}.
\end{gathered}
\end{equation}
The Puiseux series $y^{(1)}_{\infty}(x)$ has an arbitrary coefficient $b_3$
and exists whenever the restriction $e=3(27+6\xi-4\alpha)/4$ holds. The
Puiseux series $y^{(2)}_{\infty}(x)$ possesses uniquely determined
coefficients. Note that we use novel designations for the coefficients of the
Puiseux series $y^{(2)}_{\infty}(x)$. The series $y^{(1)}_{\infty}(x)$ terminates
at  the zero term under the condition
\begin{equation}
\begin{gathered}
 \label{Lienard24_Invc_2con}
\delta=\frac{1}{24}(2\xi+9)(18\alpha+4\xi\alpha-4\sigma-4\xi^2-36\xi-81).
\end{gathered}
\end{equation}
Thus, we see that systems \eqref{Lienard1_DS24} possess the invariant algebraic
curve $F_1(x,y)=0$ of Theorem~\ref{T:Lienard_Integrability2} whenever
$e=3(27+6\xi-4\alpha)/4$ and $\delta$ is of the form
\eqref{Lienard24_Invc_2con}. The related polynomial and the cofactor can be
represented as
\begin{equation}
\begin{gathered}
 \label{Lienard24_Invc_2}
F_1(x,y)=y+x^3+\frac{3}{2}x^2-\left(\xi-\alpha+\frac{9}{2}\right)x-\frac13(\sigma+\xi^2-\xi\alpha)\\
 -\frac14(27+12\xi-6\alpha),\quad \lambda_1(x,y)=3x-\xi-\frac92.
\end{gathered}
\end{equation}
Condition \eqref{Lienard_Inegrability_C1_cond} gives the following
restriction: $\xi=-7$. Finally, we use
Theorem~\ref{T:Darboux_pols_computation_Lienard} to find an irreducible
 invariant algebraic curve that exists simultaneously with $F_1(x,y)=0$ and is given by expression \eqref{Lienard1_F} where
$k=1$ and $N\in\mathbb{N}$. As a result, we obtain the values of the
parameters as presented in relation \eqref{L24_Integrability_main}. The
related invariant algebraic curves are given by the polynomials
\begin{equation}
\begin{gathered}
 \label{Lienard24_Invc_3}
(I):\quad F_2(x,y)=y-{x}^{2}-\frac53x-{\frac{25}{36}},\quad \lambda_2(x,y)=-3x^2-2x+\frac{5}{12};\hfill\\
(II):\quad F_2(x,y)={y}^{2}
+{\frac { \left( 6x+5 \right)  \left( 6x-13 \right)
 \left( 6x+11 \right) y}{216}}-{\frac { \left( 6x-13 \right)}{7776}}\hfill\\
  \times\left( 6x+11 \right)^{2} \left( 6x+5 \right)^{2},\quad\lambda_2(x,y)=-3x^2 + x + \frac{71}{12}.
\end{gathered}
\end{equation}
We calculate explicit expressions of Darboux integrating factors with the
help of expression~\eqref{Lienard_Inegrability_Int_fact}.
\end{proof}
In  case ($I$) a Liouvillian first integral is given by expression
\eqref{Lienard_Inegrability_FI_A_p1}, where one sets
\begin{equation}
\begin{gathered}
 \label{Lienard24_LFI_v}
l=2,\quad k=3,\quad \beta=-1,\quad v(x)=x+\frac56.
\end{gathered}
\end{equation}
In case ($II$) a Liouvillian first integral reads
as
\begin{equation}
\begin{gathered}
 \label{Lienard24_LFI_v2}
I(x,y)=\frac{\sqrt{w(x)}}{w(x)}\sum_{l=0}^2\left(\left\{\sqrt{w(x)}-v(x)\right\}
U(x)\ln\left(z^{\frac13}-U(x)\exp\left(\frac{2\pi l i}{3}\right)\right)\right.\\
\left.+
\left\{\sqrt{w(x)}+v(x)\right\}
V(x)\ln\left(z^{\frac13}-V(x)\exp\left(\frac{2\pi l i}{3}\right)\right)\right)\exp\left(\frac{2\pi l i}{3}\right)+6z^{\frac13},
\end{gathered}
\end{equation}
where we have introduced the notation
\begin{equation}
\begin{gathered}
 \label{Lienard24_LFI_U_V}
U(x)=
\left\{u(x)-v(x)+\sqrt{w(x)}\right\}^{\frac13}, V(x)=
\left\{u(x)-v(x)-\sqrt{w(x)}\right\}^{\frac13},z=y+u(x).
\end{gathered}
\end{equation}
The polynomials $u(x)$, $v(x)$, $w(x)$
take the form
\begin{equation}
\begin{gathered}
 \label{Lienard24_LFI_u_v_w}
u(x)=\frac{\left(6 x +11\right) \left(36 x^{2}-12 x -71\right)}{216},\quad v(x)=\frac{\left(6 x +5\right) \left(6 x -13\right) \left(6 x +11\right)}{432},\\
w(x)=\frac{\left(6 x -13\right) \left(6 x +11\right)^{3} \left(6 x +5\right)^{2}}{186624}.
\end{gathered}
\end{equation}


Concluding this section we note that  the method of Puiseux
series and the explicit expression \eqref{General_Fl_cof} of the cofactor of an invariant algebraic curve greatly facilitate   the classification of integrable multi-parameter
planar differential systems.

\section{Conclusion}\label{S:Conclusion}

This work completely solves  the Liouvillian integrability problem for
polynomial Li\'{e}nard differential systems \eqref{Lienard_gen} satisfying
the condition $\deg g\neq2\deg f+1$. In the case $\deg g=2\deg f+1$ our
results are complete for the non-resonant systems. We say that a Li\'{e}nard
differential system with the restriction $\deg g=2\deg f+1$ is resonant near
infinity if equation~\eqref{eq:DP2_5_2} possesses a positive rational
solution. The resonance condition introduces a restriction on the
highest-degree coefficients $f_0$ and $g_0$ of the polynomials $f(x)$ and
$g(x)$.

We have established that a generic nonlinear polynomial Li\'{e}nard
differential system \eqref{Lienard_gen} with fixed degrees of the polynomials
$f(x)$ and $g(x)$ is not Liouvillian integrable provided that the following
restriction $\deg g>\deg f$ is valid. However, as we have demonstrated,
 Liouvillian integrable subfamilies  exist for any degrees of the polynomials $f(x)$ and $g(x)$ whenever $\deg g>\deg f$. Besides that, we have
classified polynomial Li\'{e}nard differential systems possessing
non-autonomous Darboux first integrals and non-autonomous Jacobi last
multipliers with a time-dependent exponential factor. Some of our results describing
explicit families of Liouvillian integrable Li\'{e}nard differential systems
are gathered in Table~\ref{Tb:Lienard_IC}.

\bigskip

\textit{Remarks to Table \ref{Tb:Lienard_IC}}

\begin{enumerate}

\item Natural numbers $l$ and $k$ are both non-unit.

\item Symbols $D$, $E$, and $L$ mean Darboux, elementary, and Liouvillian,
    respectively. Symbols $\overline{D}$ and $\overline{E}$ mean
    non-Darboux and non-elementary.

\item Family $(A)_1$ gives all Liouvillian integrable families of
    Li\'{e}nard differential systems \eqref{Lienard_gen} such that
    the related equation \eqref{Lienard_y_x} possesses two distinct
    polynomial solutions and the following inequalities $\deg f<\deg
    g<2\deg f+1$ are valid.

\item Families $(B)_1$ and $(B)_2$ produce all non-resonant Darboux
    integrable Li\'{e}nard differential systems \eqref{Lienard_gen}
    satisfying the restriction $\deg g=2\deg f+1$.

\item Families $(B)_1$, $(B)_2$, $(B)_3$, and $(B)_4$ include all
    non-resonant Liouvillian integrable Li\'{e}nard differential systems
    \eqref{Lienard_gen} satisfying the restriction $\deg g=2\deg f+1$.
    Note that families $(B)_3$ and $(B)_4$ also involve integrable resonant
    systems. Consequently, additional  restrictions should be
   imposed if one is interested only in the non-resonant case. These
   restrictions are described in Theorems
   \ref{T:Lienard_degenerate_Liouville1} and
   \ref{T:Lienard_degenerate_Liouville2}.

\item Family $(C)_1$ gives all Liouvillian integrable families of
    Li\'{e}nard differential systems~\eqref{Lienard_gen} with $\deg
    g>2\deg f+1$ possessing either a hyperelliptic invariant algebraic curve or
    two distinct invariant algebraic curves with generating polynomials
    of the first degree with respect to $y$.

\end{enumerate}

\begin{table}[h!]
        \center
        \footnotesize
       \begin{tabular}[pos]{|c|c|c|c| }
        \hline
        Family   &  $f(x)$, $g(x)$ & $M(x,y)$ & $I(x,y)$, type of first integral \\
        \hline
       $(A)_1$ &  $f(x)=-\left[k\beta v^{k-1}+(k+l)v^{l-1}\right]v_x $ &
         $\frac{z^{-\frac{l}{k}}}{y-v^l} $ &
         $\displaystyle \frac{k\beta^{\frac{l}{k}}}{k-l}z^{\frac{k-l}{k}}+
\sum_{j=0}^{m}\exp\left[-\frac{\pi l(2j+1)i}{k}\right]$  \\
Th \ref{T:Lienard_IntegrabilityA_partial} & $g(x)=k\left[\beta v^k+v^l\right]v^{l-1}v_x $ & & $ \times\ln\left\{z^{\frac{1}{m+1}}
-\exp\left[\frac{\pi(2j+1)i}{m+1}\right][\beta v^k]^{\frac{1}{m+1}}\right\}$ \\
& $z=y-\beta v^k-v^l$,
$\frac{m+1}{n-m}=\frac{k}{l}$, $(l,k)=1$ &  & $\overline{D}EL$ \\
& $v(x)\in\mathbb{C}[x]$, $\deg v=\frac{n-m}{l}$, $\beta\in\mathbb{C}\setminus\{0\}$ &  & \\
         \hline
         $(B)_1$ &  $f(x)=-\frac{2f_0}{f_0-\delta}q_{1,\,x} $,
         $g(x)=\frac{f_0+\delta}{f_0-\delta}q_{1,\,x}q_1$ &
         $\frac{1}{y-\frac{(f_0+\delta)}{(f_0-\delta)}q_1} $ &
         $ \left[y-q_1\right]^{\delta-f_0}\left[y-\frac{(f_0+\delta)}{(f_0-\delta)}q_1\right]^{\delta+f_0}$  \\
         Th \ref{T:Lienard_degenerate_Darboux1} & $q_1(x)\in\mathbb{C}[x]$, $\delta$, $f_0\in\mathbb{C}\setminus\{0\}$
          & $\times\frac{1}{y-q_1}$   & $DEL$ \\
          & $q_1(x)=\frac{\delta-f_0}{2\{m+1\}}x^{m+1}+o(x^{m+1})$ & & \\
         \hline
         $(B)_2$ &  $f(x)=-2q_x $,
         $g(x)=qq_x$ &
         $\frac{1}{\left\{y-q\right\}^2} $ &
         $ \left[y-q(x)\right]\exp\left[-\frac{q(x)}{y-q(x)}\right]$  \\
         Th \ref{T:Lienard_degenerate_Darboux2} & $q(x)\in\mathbb{C}[x]$, $\deg q=m+1$
          &   & $DEL$ \\
          \hline
         $(B)_3$ &  $f(x)=-\left[\frac{\{(2d_1+1)l+k\}l}{k-l}u^{l-1}\right. $ &
         $\frac{[y-p_1]^{d_1}}{[y-p_2]^{d_1+1+\frac{k}{l}}}$ &
         $ \frac{p_2B\left(\frac{y-p_1}{p_2-p_1};1+d_1,
         -d_1-\frac{k}{l}\right)}{\{p_2-p_1\}^{\frac{k}{l}}}$  \\
          Th \ref{T:Lienard_degenerate_Liouville_polynomial} & $ \left.+\beta(l+k)u^{k-1}\right]u_x$
           &
          & $-
\frac{B\left(\frac{y-p_1}{p_2-p_1};1+d_1,1-d_1-\frac{k}{l}\right)}{\{p_2-p_1\}^{\frac{k}{l}-1}}$  \\
          & $g(x)=\left[l\beta^2u^{2k-1}+
          \frac{\{(2d_1+1)l+k\}l\beta}{k-l}\right.$ & & $\overline{D} \overline{E}L$ ($d_1\not\in\mathbb{Q}$) \\
         & $\times u^{k+l-1}\left.+\frac{(ld_1+k)(d_1+1)l^2}{(k-l)^2}u^{2l-1}\right]u_x$ & & \\
         & $p_1(x)=\beta u^k(x)+\frac{(d_1+1)l}{k-l}u^l(x)$ & & \\
         & $p_2(x)=\beta u^k(x)+\frac{(ld_1+k)}{k-l}u^l(x)$ & & \\
         & $(l,k)=1$, $d_1$, $\beta\in\mathbb{C}\setminus\{0\}$ & & \\
         & $u(x)\in\mathbb{C}[x]$, $\deg u=\frac{m+1}{\max\{k,l\}}$ & & \\
         \hline
         $(B)_4$ &  $f(x)=\left[\frac{2l^2}{l-k}v^{l-1}-(l+k)\beta v^{k-1}\right]v_x $ &
         $\frac{\exp\left[\frac{v^l}{y-q}\right]}{[y-q]^{\frac{l+k}{l}}}$ &
         $ v^{l-k}\gamma\left(-\frac{l-k}{l},\frac{v^l}{q-y}\right)-\frac{q}{v^{k}}\gamma\left(\frac{k}{l},\frac{v^l}{q-y}\right)$  \\
         Th \ref{T:Lienard_degenerate_Liouville2} &
         $g(x)= \left[\frac{l^3}{(l-k)^2}v^{2l-1}+l\beta^2 v^{2k-1}\right. $ &
          &          $\overline{D} \overline{E}L$  \\
           &             $\left.-\frac{2l^2\beta}{l-k}v^{l+k-1}\right]v_x $ & &          $ $  \\
          &               $q(x)=-\frac{l}{l-k}v^l+\beta v^k $, $\beta\in\mathbb{C}\setminus\{0\}$ & &        \\
          &  $v(x)\in\mathbb{C}[x]$, $\deg v=\frac{m+1}{\max\{k,l\}}$,     $(l,k)=1$ $ $ & & \\
                    \hline
                    $(C)_1$ &  $f(x)=\frac{(k+2l)}{4}w^{l-1}w_x $ &
         $z^{-\left(\frac12+\frac{l}{k}\right)}$ &
         $
{}_2F_1\left(\frac12,\frac12+\frac{l}{k};\frac32;-\frac{(2y+w^l)^2}{4\beta w^k}\right)$\\
         Th \ref{T:Lienard_IntegrabilityC_partial} &
         $g(x)= \frac{k}8\left(w^{2l-1}+4\beta w^{k-1}\right)w_x $ &
            &          $\times\frac{(2l-k)(2y+w^l)}{4kw^{\frac{k}2}\beta^{\frac12+\frac{l}{k}}}+z^{\frac12-\frac{l}{k}}$  \\
          Th \ref{T:Lienard_IntegrabilityC_polynomial} &  $z=\left[y+\frac{w^l}{2}\right]^2+\beta w^k$, $\frac{n+1}{m+1}=\frac{k}{l}$ & & $\overline{D} \overline{E}L$ \\
           &  $w(x)\in\mathbb{C}[x]$, $\deg w=\frac{m+1}{l}$,     $(l,k)=1$ $ $ & & \\
         \hline
             \end{tabular}
    \caption{The explicit Liouvillian integrable families of Li\'{e}nard differential systems.} \label{Tb:Lienard_IC}
\end{table}

Let us note that families  $(B)_1$ and $(B)_2$ are
    those given by the Chiellini integrability condition
    $\{f(x)/g(x)\}_x=\alpha f(x)$,
    see~\cite{Chiellini01}. Remarkably, Chiellini integrable Li\'{e}nard
    differential systems can be linearized via generalized Sundman
    transformations~\cite{Berkovich01}.  Other integrable families from Table \ref{Tb:Lienard_IC} can also be transformed to a more simple form via generalized Sundman
    transformations, see remarks and comments to Theorems \ref{T:Lienard_IntegrabilityA_partial}, \ref{T:Lienard_degenerate_Liouville_polynomial}, \ref{T:Lienard_degenerate_Liouville2}, and \ref{T:Lienard_IntegrabilityC_partial}. These systems with the exception of a number of partial cases  that
    appear in  \cite{Lakshmanan01, Lakshmanan02, Lakshmanan03, Polyanin, Stachowiak}
    seem to be new.

Let us enumerate some unsolved problems related to the integrability and
solvability of Li\'{e}nard differential systems. Despite the fact that the
subset of resonant Li\'{e}nard differential systems is of Lebesgue measure
zero in the set of all polynomial Li\'{e}nard differential systems satisfying the
condition $\deg g=2\deg f+1$, it is an interesting open problem to perform a
classification of invariant algebraic curves and integrable subfamilies of
particular resonant polynomial Li\'{e}nard differential systems provided that
only a resonant condition is imposed on the parameters of the systems.  The
method of Puiseux series~\cite{Demina11, Demina18} can deal with each family of
resonant Li\'{e}nard differential systems characterized by a fixed positive
rational Fuchs index separately. Note that several novel families of
Liouvillian integrable resonant systems are presented in
Theorems~\ref{T:Lienard_degenerate_Liouville_polynomial} and
\ref{T:Lienard_degenerate_Liouville2}, see also Corollary $2$ after the
latter theorem. In addition, a number of integrable resonant families with
$\deg f=2$ and $\deg g =5$ are found via $\lambda$ symmetries in~\cite{Ruiz01}. These families possess an exciting
property: they simultaneously have autonomous and non-autonomous Darboux
first integrals given by expression \eqref{FI_t_gen} with $\omega=0$ and
$\omega\neq0$, respectively. This fact allowed A.~Ruiz and C.~Muriel to
obtain nice expressions of the general solutions. Systems \eqref{Lienard_degenerate_Darboux_system_time} with $\delta=\pm mf_0/(m+2)$  have the same property.

Along with this, it is a difficult open problem to perform a classification
of integrable polynomial Li\'{e}nard differential systems with
non-Liouvillian first integrals. At the moment only particular examples that
can be transformed to linear equations are available, for more details see
articles \cite{Lienard-Riccati, Morales-Ruiz01, Gine_Airy, Demina_Gine_Valls}.

Another important problem is to study  rational Li\'{e}nard differential
systems. If the functions $f(x)$ and $g(x)$ in expression \eqref{Lienard_gen}
are rational, then systems~\eqref{Lienard_gen} give rise to the following
polynomial differential systems in the plane
\begin{equation}
 \label{Lienard_gen_rational}
 x_t=h(x)y,\quad  y_t=-\tilde{f}(x)y-\tilde{g}(x),\quad h(x), \tilde{f}(x), \tilde{g}(x)\in\mathbb{C}[x].
\end{equation}
Investigating the analytic and qualitative properties of these systems with
respect to the degrees of polynomials $h(x)$, $\tilde{f}(x)$, and
$\tilde{g}(x)$ is a future challenge.

\section{Acknowledgments}

This research was  supported by Russian Science Foundation grant
19--71--10003.


\nocite{*}
\bibliographystyle{apa}
\bibliography{D_IL}%

\clearpage



\end{document}